\pdfoutput=1
\documentclass[a4paper,onecolumn,11pt,nopdfoutputerror]{quantumarticle}
\usepackage[T1]{fontenc}
\usepackage[utf8]{inputenc}
\def\techreport{true}

\ifx\papertech\undefined
  \ifx\techreport\undefined
    \newcommand{\onlytech}[1]{\ignorespaces}
    \newcommand{\onlypaper}[1]{#1}
  \else
    \newcommand{\onlytech}[1]{#1}
    \newcommand{\onlypaper}[1]{\ignorespaces}
    \makeatletter
    \def\@copyrightspace{\relax}
    \makeatother
  \fi
\else
  \newcommand{\onlytech}[1]{ {\color{blue}[(TECH ONLY) #1]} }
  \newcommand{\onlypaper}[1]{ {\color{olive}[(PAPER ONLY) #1]} }
\fi

\usepackage{mathtools,amssymb,lipsum}

\usepackage{tikz}
\usetikzlibrary{arrows,chains,matrix,positioning,scopes,shapes,arrows.meta,calc,fit}
\pgfdeclarelayer{background}
\pgfdeclarelayer{behind}
\pgfdeclarelayer{above}
\pgfdeclarelayer{glass}
\pgfsetlayers{background,behind,main,above,glass}
\makeatletter
\tikzset{join/.code=\tikzset{after node path={%
\ifx\tikzchainprevious\pgfutil@empty\else(\tikzchainprevious)%
edge[every join]#1(\tikzchaincurrent)\fi}}}
\makeatother
\tikzset{>=stealth',every on chain/.append style={join},
         every join/.style={->}}
\tikzstyle{labeled}=[execute at begin node=$\scriptstyle,
   execute at end node=$]
\usepackage{amsmath}
\usepackage{color}
\usepackage{amsfonts}
\usepackage[numbers,sort&compress]{natbib}
\usepackage{lettrine}
\usepackage{amsthm}

\allowdisplaybreaks[4]

\usepackage{graphicx}
\usepackage{footnote}
\usepackage{booktabs}
\usepackage{float}
\floatstyle{ruled}
\newfloat{algorithm}{tbp}{loa}
\floatname{algorithm}{Algorithm}

\usepackage{algpseudocode} 

\usepackage{subeqnarray}
\usepackage{cases}
\usepackage{url}
\usepackage{hyperref}
\usepackage{bm}
\usepackage{physics}
\usepackage{xr}
\usepackage{scalerel}
\usepackage{enumitem}
\newtheorem{theorem}{Theorem}

\newtheorem{assumption}[theorem]{Assumption}
\newtheorem{lemma}[theorem]{Lemma}
\newtheorem{corollary}[theorem]{Corollary}
\newtheorem{proposition}[theorem]{Proposition}

\newtheorem{remark}[theorem]{Remark}

\newtheorem{definition}[theorem]{Definition}

\newcommand{\vect}[1]{ \mathbf{#1} }

\newcommand{\edefinition}{\hfill$\triangle$}

\newcommand{\eremark}{\hfill$\triangle$}
\usepackage{etoolbox}
\providecommand{\splitmainmode}{full}
\newif\ifbuildmainpart
\newif\ifbuildappendixpart
\newif\ifsplitmainonly
\newif\ifsplitappendixonly
\buildmainparttrue
\buildappendixparttrue
\splitmainonlyfalse
\splitappendixonlyfalse
\def\splitmodemain{main}
\def\splitmodeappendix{appendix}
\ifx\splitmainmode\splitmodemain
  \buildmainparttrue
  \buildappendixpartfalse
  \splitmainonlytrue
  \splitappendixonlyfalse
\else
  \ifx\splitmainmode\splitmodeappendix
    \buildmainpartfalse
    \buildappendixparttrue
    \splitmainonlyfalse
    \splitappendixonlytrue
  \else
    \buildmainparttrue
    \buildappendixparttrue
    \splitmainonlyfalse
    \splitappendixonlyfalse
  \fi
\fi
\ifsplitmainonly
  \let\splitmodeSavedBibcite\bibcite
  \renewcommand{\bibcite}[2]{}
  \let\bibcite\splitmodeSavedBibcite
\fi
\ifsplitappendixonly
  \let\splitmodeSavedBibcite\bibcite
  \renewcommand{\bibcite}[2]{}
  \let\bibcite\splitmodeSavedBibcite
\fi
\AtEndEnvironment{definition}{\edefinition}
\newcommand{\set}[1]{\mathcal{#1}}

\newcommand{\matr}[1]{\mathbf{#1}}
\newcommand{\imagunit}{\iota}

\newcommand{\setF}{ \mathcal{F} }
\newcommand{\setE}{ \mathcal{E} }
\newcommand{\e}{ \mathrm{e} }
\newcommand{\setpf}{ \partial f }

\newcommand{\sR}{\mathbb{R}}
\newcommand{\sRp}{\mathbb{R}_{\geq 0}}
\newcommand{\sRpp}{\mathbb{R}_{> 0}}
\newcommand{\sC}{\mathbb{C}}
\newcommand{\sZ}{\mathbb{Z}}
\newcommand{\sZpp}{\mathbb{Z}_{>0}}

\newcommand{\Herm}{\mathsf{H}}

\newcommand{\setx}{ \mathcal{X} }

\newcommand{\vx}{ \vect{x} }
\newcommand{\vy}{ \vect{y} }
\newcommand{\vz}{ \vect{z} }

\newcommand{\defeq}{\triangleq}
\newcommand*{\ie}{\textit{i.e.}}
\newcommand*{\eg}{\textit{e.g.}}
\newcommand*{\etal}{\textit{et al.}}

\newcommand{\vk}{ \mathbf{k} }

\newcommand{\np}{ n_{\mathrm{p}} }
\newcommand{\Np}{ N_{\mathrm{p}} }

\newcommand{\setV}{ \set{V} }
\newcommand{\setVb}{ \set{V}_{\mathrm{B}} }
\newcommand{\setVbnsfn}{ \set{V}_{\mathrm{B},\nsfn} }
\newcommand{\setVbnsfnn}[1]{ \set{V}_{\mathrm{B},#1} }

\newcommand{\vqpn}[1]{ \mathbf{q}_{\mathrm{p},#1} }

\newcommand{\qpn}[1]{ q_{\mathrm{p},#1} }

\newcommand{\mUpha}[2]{ \matr{U}_{\mathrm{pe-sub}}\left( #1, #2 \right) }
\newcommand{\mUphamd}[2]{ \matr{U}_{\mathrm{pe}}\left( #1, #2 \right) }
\newcommand{\mUphano}{ \matr{U}_{\mathrm{pe-sub}} }

\newcommand{\graphG}{ \mathsf{G} }

\newcommand{\Nsfn}{ N_{\graphG} }
\newcommand{\Nw}{ N_{\mathrm{w}} }
\newcommand{\nsfn}{ n_{\graphG} }
\newcommand{\nw}{ n_{\mathrm{w}} }

\newcommand{\QPU}{ \mathrm{QPU} }
\newcommand{\QPUc}{ \mathrm{QPU}_{\mathrm{c}} }

\newcommand{\vqfrn}[1]{ \mathbf{q}_{\mathrm{1st},#1} }
\newcommand{\vqsrn}[1]{ \mathbf{q}_{\mathrm{2nd},#1} }

\newcommand{\zdec}{ z_{\mathrm{dec}} }

\newcommand{\vqB}{\mathbf{q}_{\mathrm{B}}}
\newcommand{\vqBauxn}{\mathbf{q}_{\mathrm{Baux},\nsfn}}

\newcommand{\qB}{q_{\mathrm{B}}}

\newcommand{\rmc}{\mathrm{c}}
\newcommand{\rmp}{\mathrm{p}}

\newcommand{\rmini}{\mathrm{ini}}

\newcommand{\qpc}{ q_{\mathrm{p,c}} }
\newcommand{\vqpc}{ \mathbf{q}_{\mathrm{p,c}} }

\newcommand{\vqstn}[1]{\mathbf{q}_{1\mathrm{st},#1}}

\newcommand{\vqloc}{\mathbf{q}_{\mathrm{loc},[\Nsfn]}}
\newcommand{\vqlocn}{\mathbf{q}_{\mathrm{loc},\nsfn}}
\newcommand{\vqcen}{\mathbf{q}_{\mathrm{cen}}}
\newcommand{\vqaux}{\mathbf{q}_{\mathrm{aux}}}
\newcommand{\vqauxc}{\mathbf{q}_{\mathrm{aux,c}}}
\newcommand{\vqauxd}{\mathbf{q}_{\mathrm{aux,d}}}

\newcommand{\vqoran}[1]{\mathbf{q}_{\mathrm{ora},#1}}

\newcommand{\Npe}{N_{\mathrm{pe}}}

\newcommand{\na}{n_{\mathrm{a}}}

\newcommand{\gmax}{g_{\max}}
\newcommand{\gmaxn}[1]{g_{\max}^{(#1)}}

\newcommand{\vxmax}{\mathbf{x}_{\max}}
\newcommand{\svzero}{\ket{\mathbf{0}}}

\newcommand{\szero}{\ket{0}}
\newcommand{\sone}{\ket{1}}

\algrenewcommand\algorithmicindent{1em}

\begin{document}

\title{A Scalable Distributed Quantum Optimization Framework via Factor Graph Paradigm}

\author[a]{Yuwen Huang}
\email{yuwen.huang@ieee.org}
\thanks{Corresponding author.}
\author[b]{Xiaojun Lin}
\email{xjlin@ie.cuhk.edu.hk}
\author[a]{Bin Luo}
\email{binluo@link.cuhk.edu.hk}
\author[a]{John C.S. Lui}
\email{cslui@cse.cuhk.edu.hk}
\affiliation[a]{Department of Computer Science and Engineering, The Chinese University of Hong Kong, Hong Kong SAR, China}
\affiliation[b]{Department of Information Engineering, The Chinese University of Hong Kong, Hong Kong SAR, China}
%



\maketitle
\ifbuildmainpart

\begin{abstract}
Distributed quantum computing (DQC) connects many small quantum processors into a single logical machine, offering a
practical route to scalable quantum computation. However, most existing DQC paradigms are structure-agnostic. Circuit
cutting proposed by Peng \etal~in [Phys. Rev. Lett., Oct. 2020] reduces per-device qubits at the cost of exponential classical post-processing, while
search-space partitioning proposed by Avron \etal~in [Phys. Rev. A., Nov. 2021] distributes the workload but weakens Grover's ideal quadratic speedup. In
this paper, we introduce a structure-aware framework for distributed quantum optimization that resolves this complexity-resource trade-off.
We model the objective function as a factor graph and expose its sparse interaction structure. We cut the graph along its natural
``seams'', \ie, a separator of boundary variables, to obtain loosely coupled subproblems that fit on
resource-constrained processors. We coordinate these subproblems with shared entanglement, so the network executes a
single globally coherent search rather than independent local searches.
We prove that this design preserves Grover-like scaling: for a search space of size $N$, our framework achieves
$O(\sqrt{N})$ query complexity up to processors and separator dependent factors, while relaxing the qubit requirement of
each processor. We extend the framework with a hierarchical divide-and-conquer strategy that scales to large-scale
optimization problems and supports two operating modes: a fully coherent mode for fault-tolerant networks and a hybrid mode that
inserts measurements to cap circuit depth on near-term devices.
We validate the predicted query-entanglement trade-offs through simulations over diverse network topologies, and we show
that structure-aware decomposition delivers a practical path to scalable distributed quantum optimization on quantum
networks.

\end{abstract}

\section{Introduction}
\label{sec:introduction}

Quantum computing offers a potential to solve certain problems that are widely believed to be intractable for even the most powerful classical computers. Present quantum hardware, however, remains constrained by fundamental physical limits. Current devices operate in the noisy intermediate-scale quantum (NISQ) regime~\cite{Preskill2018}, where processors of 50--100s of qubits are available but are susceptible to environmental noise and decoherence~\cite{Kempe2001, Clerk2010}. These imperfections restrict the achievable circuit depth and make the construction of a single, large-scale, monolithic fault-tolerant quantum computer a formidable long-term engineering challenge~\cite{Kim2023}.

An alternative approach to scalable quantum computation is a distributed, modular architecture that connects many smaller, higher-fidelity quantum processing units (QPUs)~\cite{Monroe2014, Wehner2018}. This modular approach, analogous to classical high-performance computing~\cite{Hager2010}, envisions a network of QPUs collaborating on a single computational task~\cite{Kimble2008}. This distributed quantum computing (DQC) model uses quantum networks to establish shared entanglement between modules, enabling non-local operations and allowing the network to function as a single, coherent logical processor~\cite{Pompili2021}. The central challenge, however, lies in designing quantum algorithms and network protocols that do not sacrifice the very quantum advantage they are meant to enable.

The current DQC literature is largely shaped by two \emph{structure-agnostic} approaches that split a computation either at
the circuit level or at the input-space level.
\emph{Circuit cutting} decomposes a monolithic circuit into fragments that fit on separate QPUs and then reconstructs the
global outcome by classical post-processing of fragment statistics~\cite{Peng2020, Tang2021,Piveteau2024}. A common
instance is to replace nonlocal gates by sampling local operations (``virtual gates'')~\cite{Mitarai2021}. Circuit
cutting saves qubits but shifts the burden to classical computation: the reconstruction cost grows exponentially with the
number of inter-processor cuts $K$, typically as $O(4^K)$~\cite{Peng2020, Piveteau2024} or $O(9^K)$~\cite{Mitarai2021}.
When $K$ is large, this overhead can erase the benefit of quantum speedups.
\emph{Search-space partitioning} instead divides a search space of size $N$ across $K$ processors, which run independent
local searches and coordinate only through classical communication~\cite{Avron2021}. This breaks the coherence required
by Grover's algorithm: rather than the $O(\sqrt{N})$ queries of a monolithic search, the total query count becomes
$O(K\sqrt{N/K})=O(\sqrt{KN})$~\cite{Avron2021}. The extra factor $\sqrt{K}$ implies that adding processors to
satisfy local qubit limits makes the computation harder in $K$.
These limitations share a common cause: both paradigms distribute the workload without maintaining global quantum
coherence across processors, which is essential for quadratic speedups.
This motivates a \emph{structure-aware} alternative: decompose the problem along a small set of boundary variables so that
most computation can be localized, while a limited amount of shared entanglement preserves coherence across the boundary
and enables Grover-like global amplification.
Other decomposition methods, such as entanglement forging~\cite{Eddins2022}, target specific settings (\eg, a bipartite
cut with low entanglement) and therefore do not provide a general framework for distributed optimization.

This paper introduces a \emph{structure-aware} paradigm for distributed quantum optimization that resolves the complexity--resource trade-off by using
two resources that structure-agnostic methods leave on the table: \emph{problem structure} and \emph{global coherence}.
We represent the objective function with a factor graph~\cite{Kschischang2001}, which makes sparse interactions explicit through a
sum of local terms. We then cut the graph along a small set of \emph{boundary variables} $\setVb$ and assign each
resulting subproblem to a resource-constrained worker QPU. The key technical point is that the workers do not run as
independent solvers: shared entanglement lets the network implement the required distributed primitives while keeping
the computation coherent across subproblems.

\begin{table*}[t]
  \scriptsize
  \centering
  \caption{Quantitative comparison of DQC paradigms.
  Here $N=2^n$ and $K$ denotes the number of workers (or the number of cuts in circuit cutting).
  Query complexity counts oracle queries. ``Light'' post-processing means polynomial-time aggregation or coordination.
  }
  \label{tab:dqc_comparison_extended}

  \setlength{\tabcolsep}{7pt}
  \renewcommand{\arraystretch}{1.25}

  \begin{tabular}{@{}p{2.2cm}p{2cm}p{3cm}p{6.2cm}@{}}
    \toprule
    \textbf{Work} & \textbf{Query complexity} &
    \textbf{Classical post-processing} &
    \textbf{Inter-node communication} \\
    \midrule

    \textbf{This work} &
    $O(\sqrt{N})$ &
    Light (poly-time aggregation) &
    Shared EPR pairs and classical messages (to maintain global coherence) \\
    \midrule
    \addlinespace[2pt]

    Peng \etal~\cite{Peng2020} &
    $O(\sqrt{N})$ &
    $O\bigl( 4^K\cdot \mathrm{poly}(n) \bigr)$ &
    Classical messages (fragment outcomes for reconstruction) \\
    \midrule
    \addlinespace[2pt]

    Avron \etal~\cite{Avron2021} &
    $O(\sqrt{KN})$ &
    Light (poly-time coordination) &
    Classical messages (coordination only) \\
    \bottomrule
  \end{tabular}
\end{table*}

This design yields a concrete win. For a search space of size $N$, our distributed algorithm preserves Grover-like scaling:
it uses $O(\sqrt{N})$ oracle queries, up to factors that depend on the number of processors and the chosen separator,
while reducing the qubit footprint of each worker to $O\bigl(|\set{V}_{\nsfn}|\bigr)$ for a local search space of size
$2^{|\set{V}_{\nsfn}|}$. Table~\ref{tab:dqc_comparison_extended} makes this advantage explicit by contrasting our trade-offs with circuit cutting and search-space partitioning. We further corroborate the predicted resource drivers with simulations on
quantum networks of diverse topologies.
 \emph{In short, we pay a small boundary search and shared entanglement to cut per-QPU qubits sharply, while retaining Grover-like scaling.}

\noindent\textbf{Main Contributions.}
\begin{itemize}
  \item \textbf{Structure-aware DQC via factor graphs:}
  We propose a distributed framework that decomposes a broad class of structured optimization problems using a factor
  graph separator $\setVb$, producing subproblems that fit on small QPUs while remaining coupled through a coherent
  network protocol.

  \item \textbf{Grover-like query scaling with reduced per-QPU qubits:}
  We show that the resulting distributed algorithm achieves $O(\sqrt{N})$ query complexity up to processor and
  separator dependent factors, while each worker uses only $O\bigl(|\set{V}_{\nsfn}|\bigr)$ qubits. We corroborate these claims with gate-level state vector simulations on small instances, which provide end-to-end sanity checks and confirm the predicted scaling trends.

  \item \textbf{Hierarchical divide-and-conquer for arbitrary size:}
  We extend the single-level framework by a recursive construction that handles oversized subproblems. We define
  (i) a fully coherent mode, suited to fault-tolerant networks, and (ii) a hybrid mode that inserts measurements to cap
  circuit depth on NISQ devices, and we quantify the resulting query--entanglement trade-off in simulation.

  \item \textbf{End-to-end protocol and resource guarantees:}
  We specify the communication protocol and derive performance guarantees for success probability, query complexity,
  and quantum-network resources, including per-QPU qubits and EPR-pair consumption.
\end{itemize}

The paper is organized as follows. Section~\ref{sec:quantum_network_fg_decomposition_setup} introduces the network model,
factor-graph formulation, and decomposition strategy. Section~\ref{sec: distributed_quantum_algorithm} presents the
distributed algorithm and its guarantees. Section~\ref{sec:implementation_of_U_components} describes the distributed
implementation of $\matr{U}_{\mathrm{ini}}$, and Section~\ref{sec:overall_performance} analyzes overall performance and
resource use. Section~\ref{sec:hierarchical_algorithm} extends the framework with a scalable hierarchical design for both
near-term and fault-tolerant hardware. Section~\ref{sec:conclusion} concludes the paper.
Appendix~\ref{apx:basic_notation} lists basic notation. The remaining appendices provide proofs and additional results, including simulation studies (Appendices~\ref{sec:gate_level_validation} and~\ref{sec:hier_chain_multilevel}).


\section{Foundations: Quantum Network, Optimization Framework, and Decomposition}
\label{sec:quantum_network_fg_decomposition_setup}

This section establishes the foundational concepts for the distributed quantum algorithm presented in Section~\ref{sec: distributed_quantum_algorithm}. We first define the quantum network model, then define the factor graph optimization framework, and finally present a problem decomposition strategy. These elements are crucial for structuring and distributing complex optimization tasks across multiple quantum processors.

\begin{figure}[t!]
\centering
\captionsetup{font=scriptsize}
\begin{tikzpicture}[node distance=2cm, remember picture]
\tikzstyle{state_c}=[shape=diamond, draw, minimum size=0.3cm]
\node[state_c] (qc) at (0,0) [label=above: $\QPUc$] {};
\node[state_c] (q1) at (1.5,0) [label=above: $\QPU_{1}$] {};
\node[state_c] (q2) at (3,0) [label=above: $\QPU_{2}$] {};
\node[state_c] (q3) at (4.5,0) [label=above: $\QPU_{3}$] {};
\begin{pgfonlayer}{background}
\draw[-,draw]
 (qc) edge (q1)
 (q1) edge (q2)
 (q2) edge (q3)
 (q1) edge [bend right=20] (q3);
\end{pgfonlayer}
\end{tikzpicture}
\caption{An illustration of an example quantum network, where $\QPUc$ is the coordinator processor and $\QPU_{1}, \QPU_{2}, \QPU_{3}$ are worker processors.}
\label{fig: example quantum network}
\end{figure}
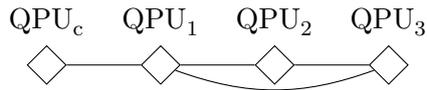

\subsection{The Quantum Network Model}
\label{def: the quantum network setup}

A quantum network comprises one coordinator quantum processor ($\QPUc$) and $\Nw$ worker quantum processors ($\QPU_1, \ldots, \QPU_{\Nw}$). Network operations are restricted to local operations and classical communications (LOCCs) (see, \eg,~\cite{Chitambar2014}): each processor executes quantum operations only on its local qubits. Processors communicate classically and can share Einstein-Podolsky-Rosen (EPR) pairs with their direct neighbors.

We represent the network as a graph $\mathsf{N}(\set{V}_{\mathsf{N}}, \set{E}_{\mathsf{N}})$. (Fig.~\ref{fig: example quantum network} illustrates a simple example where $\Nsfn = 3$.)
\begin{enumerate}
    \item The vertex set $\set{V}_{\mathsf{N}} = \{\mathrm{c}\} \cup \{\nw \mid \nw \in [\Nw]\}$ corresponds to the quantum processors, with $\mathrm{c}$ denoting $\QPUc$ and $\nw$ denoting $\QPU_{\nw}$.
    \item An edge $(v_1, v_2) \in \set{E}_{\mathsf{N}} \subseteq \set{V}_{\mathsf{N}} \times \set{V}_{\mathsf{N}}$ indicates that processors $\mathrm{QPU}_{v_1}$ and $\mathrm{QPU}_{v_2}$ are direct neighbors, enabling them to establish shared EPR pairs and a classical communication channel.
\end{enumerate}
We assume $\mathsf{N}$ is a connected graph. This connectivity ensures two properties:
\begin{enumerate}
    \item Any two processors can communicate classically, possibly via intermediate nodes.
    \item Shared EPR pairs can be established between $\QPU_{v_1}$ and $\QPU_{v_2}$ either directly if $(v_1, v_2) \in \set{E}_{\mathsf{N}}$, or indirectly using entanglement swapping~\cite{ZelseZfiukowski1993} along a path in $\mathsf{N}$.
\end{enumerate}
Throughout this paper, we assume that any pair of processors in the network can share EPR pairs and communicate classically, thereby enabling operations such as quantum teleportation~\cite[Chapter 1.3.7]{Nielsen2010}.

\subsection{Factor Graph-Based Optimization Framework}
\label{sec:factor_graph_main_definition}

Factor graphs can be used for representing the factorization of multivariate functions. Such representation is instrumental in defining a general optimization framework applicable to diverse problems. In this paper, factor graphs serve to visualize problem structure and to guide a decomposition strategy (Section~\ref{sec:decomposition_of_factor_graph}), which is essential for our proposed distributed quantum algorithm (Section~\ref{sec: distributed_quantum_algorithm}).

\begin{definition}[Factor graph-based optimization problem]
\label{def: def of snfg}
A factor graph $\graphG(\set{F}, \set{V}, \set{E}, \set{X})$ consists of:
\begin{enumerate}
    \item A bipartite graph $(\set{F} \cup \set{V}, \set{E})$, where $\setF$ is the set of function nodes and $\set{V}$ is the set of variable nodes. Each edge $(v,f) \in \set{E}$ connects a variable node $v \in \set{V}$ to a function node $f \in \set{F}$.
    \item An alphabet $\set{X} \defeq \prod_{v \in \set{V}} \setx_{v}$, where $\setx_{v} = [N_{v}^{\mathrm{card}}]$ is the finite alphabet for variable $x_{v}$ of cardinality $N_{v}^{\mathrm{card}} \in \sZpp$.
\end{enumerate}
Based on $\graphG$, we make the following definitions.
\begin{enumerate}
    \setcounter{enumi}{2} 
    \item We define a configuration $\mathbf{x} \defeq (x_{v})_{v \in \set{V}} \in \set{X}$ to be an assignment of values to all variable nodes.
    \item For any $\set{V}' \subseteq \set{V}$, we define $\setx_{ \set{V}' } \defeq \prod_{ v \in \set{V}' } \setx_{v}$ and $\mathbf{x}_{\set{V}'} \defeq ( x_{v} )_{v \in \set{V}'} \in \setx_{\set{V}'}$ to be the alphabet and a partial configuration for variables in $\set{V}'$.
    \item For each function node $f \in \setF$, we define the set $\setpf \subseteq \set{V}$ to be the set of variable nodes adjacent to $f$.
    \item Each function node $f \in \setF$ corresponds to a local function $f: \setx_{\setpf} \to \sR$ that maps configurations of its neighboring variables $\mathbf{x}_{\setpf}$ to a value $f(\mathbf{x}_{\setpf}) \in \sR$.
    \item We define the global function $g: \setx \to \sRp$ to be $g(\mathbf{x}) \defeq \sum_{f \in \set{F}} f(\mathbf{x}_{\setpf})$.
    \item We define $\mathbf{x}_{\max}$ via an optimization problem: finding the $ \vx $ that maximize $ g(\vx) $, \ie, finding
    \begin{align}
        \mathbf{x}_{\max} \defeq \arg \max\limits_{ \mathbf{x} \in \set{X} } \, g( \mathbf{x} ).
        \label{eqn: def of the location of maximum of the global fun}
    \end{align}
    and the associated maximum value is 
  \begin{align}
          \gmax \defeq \max\limits_{ \mathbf{x} \in \set{X} } \, g( \mathbf{x} ).
          \label{eqn: def of maximum of the global fun}
      \end{align}
  \end{enumerate}
  \end{definition}

Factor graphs have wide applications~\cite{Kschischang2001,Loeliger2004,J.Wainwright2008}, including theoretical computer science (\eg,  satisfiability (SAT) problems), coding theory (\eg, decoding algorithms for low-density parity-check (LDPC) or turbo codes), and portfolio optimization.
Appendix~\ref{ex:local_portfolio_investment_problem_factor_graph} presents a factor-graph formulation of a local portfolio investment problem. The formulation is readily generalized to a broader class of local portfolio investment problems, including NP-complete variants considered in~\cite{Mansini1999}.
One can use factor graphs to represent the complicated dependencies among local functions with a structured graphical form, which is particularly useful for the case where the global function naturally factorizes over local interactions.

Note that the factor graph defined above is based on the semigroup $(\sR, +)$. We emphasis that 
the main contribution in this paper, \ie, our proposed distributed quantum algorithms, can be adapted for the factor graphs defined based on one of the semigroups in $(\sR, \max)$, $(\sR, \min)$, and $(\sRp, \times)$, by appropriately modifying how values are combined or compared within the algorithmic steps.

\textbf{It is sufficient to consider binary variables and normalized global function}.
  For compatibility with quantum algorithm design, particularly for encoding values into quantum states in qubits, in the remaining part of this paper, we consider factor graphs with \emph{binary variables} ($x_v \in \{0,1\}$ for all $v \in \set{V}$), \emph{nonnegative-valued} local functions, \ie, $f(\mathbf{x}_{\setpf}) \in \sR_{\geq 0}$ for all $f \in \setF$ and $\mathbf{x}_{\setpf} \in \set{X}_{\setpf}$, and a global function $g(\mathbf{x})$ normalized to the \emph{interval} $[0,1]$ for all $\mathbf{x}\in \set{X}$.
  Appendix~\ref{apx: binarization_normalization} (Proposition~\ref{prop: appendix_prop_simplification}) details a transformation that converts any factor graph under Definition~\ref{def: def of snfg} into an equivalent factor graph possessing these binary, nonnegativity, and normalized properties. This transformation preserves the essential structure of the optimization problem, allowing the solution $(\gmax, \mathbf{x}_{\max})$ of the original factor graph to be recovered from the solution of the transformed factor graph.

\subsection{Decomposition for Distributed Execution}
\label{sec:decomposition_of_factor_graph}

Our goal is to break a large factor-graph optimization problem into subproblems that fit on QPUs with limited qubits and limited connectivity. We encode each discrete binary variable into one qubit; for large graphs, this quickly exceeds the capacity of any single processor. In order to make the problem distributable, we choose a small set of boundary variables and condition on them so that the remaining graph splits into smaller components. Each component can then be solved on a different QPU, in parallel, with coordination only through the boundary.

Towards this end, our key observation is that the factor graphs of many applications are \emph{local}: each factor involves only nearby variables (\eg, nearest-neighbor Markov random fields and sparse parity-check codes~\cite{J.Wainwright2008}). Under such locality, small vertex separators typically disconnect large regions, so a modest boundary set $\setVb$ suffices to yield balanced subgraphs that can be mapped to distinct QPUs with limited qubits and connectivity. See, for example, standard treatments of factor graphs and localized interactions~\cite{Kschischang2001,J.Wainwright2008}. This section formalizes that split and links it directly to the distributed algorithm in Section~\ref{sec: distributed_quantum_algorithm}.

\begin{definition}[Boundary-based split of a factor graph]
\label{def: setup of the distributed quantum}
Let $\graphG=(\setF,\setV,\setE,\setx)$ be a factor graph (Definition~\ref{def: def of snfg}).
We fix a boundary set $\setVb\subseteq\setV$ of \emph{variable} nodes and form the reduced bipartite graph
\[
\graphG\setminus\setVb \defeq \bigl(\setF,\setV\setminus\setVb,\{(f,v)\in\setE:\ v\notin\setVb\}\bigr),
\]
\ie, delete all boundary variables in $\setVb$ and their incident edges. Let
\[
\graphG_{\mathrm s,1},\ldots,\graphG_{\mathrm s,\Nsfn}
\]
denote the connected components of $\graphG\setminus\setVb$ that contain at least one variable node.
For each $\nsfn\in[\Nsfn]$, we write $\setV_{\nsfn}\subseteq \setV\setminus\setVb$ and $\setF_{\nsfn}\subseteq\setF$
for the variable-node set and function-node set in $\graphG_{\mathrm s,\nsfn}$, respectively.
By construction, the factor sets $\{\setF_{\nsfn}\}_{\nsfn\in[\Nsfn]}$ are disjoint; hence each factor (function) node
belongs to exactly one subgraph $\graphG_{\mathrm s,\nsfn}$.
Finally, we define the boundary neighborhood of $\graphG_{\mathrm s,\nsfn}$ by
\begin{align*}
  \setVbnsfn &\defeq \left\{\, v\in\setVb \ \middle| \ \exists f\in\setF_{\nsfn}\ \text{s.t.}\ (f,v)\in\setE \, \right\}
  = \setVb \cap \bigcup_{f\in\setF_{\nsfn}} \partial f,
\end{align*}
\ie, the boundary variables adjacent to at least one factor in $\setF_{\nsfn}$.

For each boundary assignment $\vx_{\setVb}\in\setx_{\setVb}$ and each $\nsfn \in [\Nsfn]$, we define the \emph{local conditioned objective} based on its internal variables and the boundary variables to be
\begin{align*}
g_{\nsfn,\vx_{\setVbnsfn}}(\vx_{\setV_{\nsfn}})
\defeq \sum_{f\in\setF_{\nsfn}} f(\vx_{\setpf}),
\qquad \vx_{\setV_{\nsfn}}\in\prod_{v\in\setV_{\nsfn}}\setx_v,
\end{align*}
and define its maximum to be
$g^{\max}_{\nsfn,\vx_{\setVbnsfn}}\defeq\max\limits_{\vx_{\setV_{\nsfn}}} g_{\nsfn,\vx_{\setVbnsfn}}(\vx_{\setV_{\nsfn}})$.
The overall optimization objective, defined as a function of the \emph{boundary variables}, is then
\begin{align*}
g_{\mathrm b}(\vx_{\setVb})
\defeq \sum_{\nsfn\in[\Nsfn]} g^{\max}_{\nsfn,\vx_{\setVbnsfn}}.
\end{align*}
Maximizing $g$ reduces to a search over the boundary:
\begin{align*}
\gmax = \max_{\vx_{\setVb}\in\setx_{\setVb}} g_{\mathrm b}(\vx_{\setVb}).
\end{align*}
\end{definition}

See Fig.~\ref{fig:decomposition_example_factor_graph_DQC_boundary_nodes_highlighted} in Appendix~\ref{ex:local_portfolio_investment_problem_factor_graph} for an illustration of the decomposition in Definition~\ref{def: setup of the distributed quantum}.

\noindent\textbf{Choosing the boundary.}
In practice, there are several reasonable strategies for selecting $\setVb$.
A useful guiding principle is to choose a \emph{small} boundary that separates the factor graph into subgraphs of comparable size. This choice balances the computational load across subproblems while keeping the boundary-induced overhead small.
One principled method uses a junction-tree representation: for two adjacent cliques, the separator is given by the intersection of their variable sets, and a small intersection yields a compact separator~\cite{lauritzen1988local}.
Accordingly, the example in Fig.~\ref{fig:decomposition_example_factor_graph_DQC_boundary_nodes_highlighted} in Appendix~\ref{ex:local_portfolio_investment_problem_factor_graph} adopts the boundary suggested by the corresponding junction-tree separator.

\section{A Distributed Quantum Algorithm for Solving the Factor Graph-based Optimization Problem}
\label{sec: distributed_quantum_algorithm}

This section presents $A_{\mathrm{dist}}$, a distributed quantum algorithm to solve the general factor graph-based optimization problem from expression~\eqref{eqn: def of maximum of the global fun}.
The algorithm operates on the quantum network (Section~\ref{def: the quantum network setup}) and leverages the problem decomposition (Definition~\ref{def: setup of the distributed quantum}).
This alignment of the computation with the problem's topology enables a scalable approach to distributed quantum optimization.

The algorithm performs a quantum-accelerated binary search to find the global maximum, $\gmax$.
The core of this search is a quantum phase estimation subroutine that serves as a comparison mechanism.
This global search, coordinated by the $\QPUc$, is built upon a nested architecture: each query coherently executes local quantum optimization routines in parallel on all worker processors. This coherent parallelism is enabled by a crucial state preparation unitary, $\matr{U}_{\mathrm{ini}}$.
The complex distributed protocol for constructing this operator is a cornerstone of our framework and is detailed in Section~\ref{sec:implementation_of_U_components}.

For clarity, a comprehensive summary of all quantum registers is provided in Appendix~\ref{apx:registers}, the distributed implementations of core unitary components are detailed in Appendix~\ref{apx:implementationof_Upe}, and the rigorous proof of this section's main theorem is presented in Appendix~\ref{apx: proof_of_thm_overall_algorithm_properties}.

\begin{assumption}
In Sections~\ref{sec: distributed_quantum_algorithm}--\ref{sec:overall_performance}, we suppose that the number of worker QPUs is equal to the number of sub-graphs, \ie, $\Nw = \Nsfn$.
\eremark
\end{assumption}

\subsection{The Core Unitary: An Amplitude Amplification Operator}
\label{subsec:core_unitary_amp_amp}

The iterative binary search is driven by a family of amplitude amplification operators.
Each operator performs a Grover-type rotation that increases the amplitude of states
whose encoded objective value exceeds a prescribed threshold.

\textbf{Registers induced by the decomposition.}
Recall that removing the boundary-variable set $\setVb$ decomposes the factor graph into
$\Nsfn$ disjoint subgraphs $\{\graphG_{\mathrm s,\nsfn}\}_{\nsfn\in[\Nsfn]}$.
We allocate three groups of qubits:
\begin{itemize}
\item $\vqB$: the \emph{boundary register}, which stores an encoding of a boundary configuration
(\ie, an assignment to the boundary variables in $\setVb$).
\item $\vqloc$: the collection of \emph{local registers} associated with the $\Nsfn$ partitions.
In particular, $\vqloc$ contains the per-partition \emph{value registers}
$(\vqpn{\nsfn})_{\nsfn\in[\Nsfn]}$, where each $\vqpn{\nsfn}\in(\mathbb{C}^2)^{\otimes \Np}$
stores an $\Np$-bit encoding of the local optimum value for partition $\nsfn$.

\item $\vqaux$: auxiliary / work registers used by state preparation and by the marking circuitry.
We single out $\vqauxc\subseteq \vqaux$ for ancillas that are required to be returned to $\ket{\mathbf 0}$.
\end{itemize}
The exact qubit-level layouts are given in Appendix~\ref{apx:registers}; for the discussion below, it suffices to know
which information is stored in $\vqB$ and in the value registers $(\vqpn{\nsfn})_{\nsfn}$. We next introduce two important unitary operators that form the basis of our algorithm.

\textbf{A key state-preparation unitary.}
Given a boundary configuration, each partition $\graphG_{\mathrm s,\nsfn}$ faces a local sub-problem.
The optimal value of this local sub-problem generally depends on the boundary configuration.
We therefore introduce a crucial unitary $\matr U_{\mathrm{ini}}$ that prepares, \emph{coherently},
the local optima for \emph{all} partitions and across the superposition of \emph{all} boundary configurations.
Formally, $\matr U_{\mathrm{ini}}$ acts on $(\vqB,\vqloc,\vqaux)$ and satisfies
\begin{align}
\matr U_{\mathrm{ini}}\cdot
\ket{\mathbf 0}_{\vqB}\ket{\mathbf 0}_{\vqloc}\ket{\mathbf 0}_{\vqaux}
=\ket{\psi_{\mathrm{ini}}}\ket{\mathbf 0}_{\vqauxc}.
\label{eqn: generated state vector}
\end{align}
We assume that the prepared state $\ket{\psi_{\mathrm{ini}}}$ on the RHS of~\eqref{eqn: generated state vector} satisfies the following two properties.
Let $\vx_{\setVb}^{\ast}$ be an optimal boundary assignment attaining $\gmax$, and let
$\mathbf{z}_{\max,\rmp,[\Nsfn]} \triangleq (\mathbf{z}_{\max,\rmp,\nsfn})_{\nsfn\in[\Nsfn]}$
denote the corresponding tuple of $\Np$-bit encodings of the \emph{local} optimal objective values $\bigl(g^{\max}_{\nsfn,\vx_{\setVbnsfn}^{\ast}} \bigr)_{\nsfn \in [\Nsfn]}$ under $\vx_{\setVb}^{\ast}$.\footnote{Our distributed quantum algorithm aims to recover both the global optimum value $\gmax$ and an optimizer $\vx_{\max}$.
For this reason, we track $\mathbf{z}_{\max,\rmp,[\Nsfn]}$, which compactly represents the local optima induced by the optimal boundary assignment $\vx_{\setVb}^{\ast}$.
During the algorithm, for each candidate boundary assignment $\vx_{\setVb}$, the local optimal values are coherently encoded into an entangled
register state. The algorithm then uses a quantum phase-estimation routine to extract the values corresponding to $\vx_{\setVb}^{\ast}$.}
\begin{enumerate}
    \item \textbf{Sufficient weight on the optimal outcome.}
    Define $ p_{\min,\rmc} \defeq (1-\delta)^{2\Np\Nsfn}/|\setx_{\setVb}| $.
    The state $\ket{\psi_{\rmini}}$ assigns a non-negligible probability $p_{\min,\rmc}$ to the basis outcome where the value registers
    equal $\mathbf{z}_{\max,\rmp,[\Nsfn]}$:
    \begin{align}
        \bra{\psi_{\rmini}}
        \bigl(\matr{I}_{\mathrm{rem}}\otimes\matr{P}_{\{\mathbf{z}_{\max,\rmp,[\Nsfn]}\}}\bigr)
        \ket{\psi_{\rmini}}
        \geq p_{\min,\rmc} \geq 0,
    \label{eqn: property of psi vect with positive magnitude: final}
    \end{align}
    where the projection operator $\matr{P}_{\{\mathbf{z}_{\max,\rmp,[\Nsfn]}\}}$, as defined in Definition~\ref{def:marking_operators} in Appendix~\ref{apx:basic_notation}, projects to the subspace consisting of the branch associated with $\ket{\mathbf{z}_{\max,\rmp,[\Nsfn]}}_{(\vqpn{\nsfn})_{\nsfn}}$.

    \item \textbf{Accuracy of the encoded objective value.}
    The decoded sum of these $\Np$-bit values is a tight lower approximation of $\gmax$:
    \begin{align}
        0 \le \gmax - \sum_{\nsfn\in[\Nsfn]} \zdec\bigl(\mathbf{z}_{\max,\rmp,\nsfn}\bigr) 
        \le \Nsfn \cdot 2^{-\Np},
        \label{eqn: bound of gstar and znsfn}
    \end{align}
    where for an $\Np$-bit string $\mathbf z=\bigl(z(1),\ldots,z(\Np)\bigr)\in\{0,1\}^{\Np}$, the threshold encoding $\zdec$ is defined to be
    \begin{align}
      \zdec(\mathbf z)\defeq\sum_{k=1}^{\Np} z(k)\cdot 2^{-k}\in[0,1).
      \label{eqn:zdec_inline}
    \end{align}
    (See Appendix~\ref{apx:basic_notation} for further discussion and related operators.)
\end{enumerate}
Section~\ref{sec:implementation_of_U_components} constructs a state-preparation unitary $\matr{U}_{\mathrm{ini}}$ that
satisfies these requirements. In this section, we treat $\matr{U}_{\mathrm{ini}}$ as a black box and assume only the
specification in \eqref{eqn: generated state vector}--\eqref{eqn: bound of gstar and znsfn}.


\textbf{Amplitude amplification operator.}
We define the set of local-value strings whose \emph{sum} exceeds the threshold $\mathbf z$ to be
\begin{align}
  \set{Z}_{\ge \mathbf z}
  \defeq
  \left\{
  \mathbf z_{\mathrm p,[\Nsfn]}\in\!\!\!\prod_{\nsfn\in[\Nsfn]}\!\!\!\{0,1\}^{\Np}
  \ \middle| \
  \sum_{\nsfn\in[\Nsfn]}\!\!\!\zdec(\mathbf z_{\mathrm p,\nsfn}) \ge \zdec(\mathbf z)
  \right\}.
  \label{eqn: def of Zpz_sum}
\end{align}
Then, we define the following amplitude amplification operator, which generalizes the idea behind Grover's search algorithm~\cite{Grover1998}.
\begin{definition}[\textbf{Amplitude amplification operator}~\cite{Brassard1997}]
\label{def: the set of operators UAA}
  For each threshold string $\mathbf z\in\{0,1\}^{\Np}$, we define
  $\matr U_{\mathrm{AA},\mathsf{N},\mathbf z}$ on the registers $(\vqB,\vqloc,\vqaux)$ to be
  \begin{align}
  \matr U_{\mathrm{AA},\mathsf{N},\mathbf z}
  \defeq
  \matr U_{\mathrm{ini}}
  \cdot\bigl(2\ket{\mathbf 0}\bra{\mathbf 0}-\matr I\bigr)_{\mathrm{all}}
  \cdot \matr U_{\mathrm{ini}}^{\Herm}
  \cdot\bigl(\matr I_{\mathrm{rem}}\otimes \matr U_{\set{Z}_{\ge \mathbf z}}\bigr).
  \label{eqn: def of UAANz}
  \end{align}
  where
  \begin{itemize}
  \item $\bigl(2\ket{\mathbf 0}\bra{\mathbf 0}-\matr I\bigr)_{\mathrm{all}}$ is the reflection about the
  all-zero computational basis state on the full workspace $(\vqB,\vqloc,\vqaux)$,
  \item $\matr U_{\set{Z}_{\ge \mathbf z}}$ is a phase oracle that flips the sign of computational-basis states
  whose local value registers satisfy~\eqref{eqn: def of Zpz_sum}, \ie,
  \begin{align*}
  \matr U_{\set{Z}_{\ge \mathbf z}}\ket{\mathbf z_{\mathrm p,[\Nsfn]}}
  \defeq
  \begin{cases}
  -\ket{\mathbf z_{\mathrm p,[\Nsfn]}} & \text{if }\mathbf z_{\mathrm p,[\Nsfn]}\in \set{Z}_{\ge \mathbf z}\\
  \phantom{-}\ket{\mathbf z_{\mathrm p,[\Nsfn]}} & \text{otherwise}
  \end{cases}\quad,
  \end{align*}
  and act on the value registers $(\vqpn{\nsfn})_{\nsfn\in[\Nsfn]}$ (and any required ancillas),
  \item $\matr I_{\mathrm{rem}}$ denotes the identity on the remaining registers.
  \end{itemize}
\end{definition}

Implementation details for the reflection and oracle are provided in
  Section~\ref{sec:implementation_of_U_components} and Appendix~\ref{apx:implementationof_Upe}.

\subsection{The Distributed Algorithm: Procedure and Formalism}
\label{subsec:dist_algorithm}

\begin{algorithm}[t!]
\caption{Distributed quantum algorithm for factor graph optimization (bit-by-bit refinement)}
\label{alg:dist-opt}
\begin{algorithmic}[1]
\Statex \textbf{Input:} Factor graph $\graphG$, network $\mathsf{N}$, error parameter $\delta \in (0,1)$, precision bits $\Np$.
\Statex \textbf{Output:} A measured string $\mathbf{z}_{\max,\rmp,\rmc}\in\{0,1\}^{\Np}$ in $\vqpc$ that lower-approximates $\gmax$.
\State \textbf{Initialization:} Initialize all quantum registers across the network to $\ket{\mathbf 0}$ and set $\vqpc=\ket{\mathbf 0}$.
\For{$\np=1,\ldots,\Np$}
    \State \textbf{Form candidate threshold:} interpret each computational basis state associated with $\vqpc(1:\np-1)$ as $\mathbf z_{\mathrm{prefix}}$ and set
    $\mathbf z^{(\np)} \defeq (\mathbf z_{\mathrm{prefix}},1,0,\ldots,0)$.
    \State \label{step:threshold-test} \textbf{Threshold test (writes $\qpc(\np)$):} for each computational basis state associated with $\vqpc(1:\np-1)$, coherently implement
    $\mUpha{\matr U_{\mathrm{ini}},\,\matr U_{\mathrm{AA},\mathsf N,\mathbf z^{(\np)}}}{p_{\min,\rmc},\,\delta}$ in a distributed manner 
    so that the decision bit is written to $\qpc(\np)$ and all other registers are (approximately) restored.
\EndFor
\State \textbf{Measure and output:} Measure $\vqpc$ in the computational basis to obtain $\mathbf{z}_{\max,\rmp,\rmc}$.
\end{algorithmic}
\end{algorithm}

Algorithm~\ref{alg:dist-opt} is the main algorithmic contribution of this paper.
It outputs an $\Np$-bit string $\mathbf z_{\max,\rmp,\rmc}$ that lower-approximates the global optimum $\gmax$.
The algorithm proceeds \emph{by increasing precision}: it determines the bits of $\mathbf z_{\max,\rmp,\rmc}$ from the most
significant bit to the least significant bit, refining the approximation by one bit per iteration.

We first outline the high-level flow of the algorithm and its use of the amplitude-amplification operator defined earlier. The coordinator maintains the $\Np$-qubit register $\vqpc$, which stores the running $\Np$-bit estimate of $\gmax$.
At iteration $\np\in[\Np]$, we assume that the first $\np-1$ qubits of $\vqpc$ encode a prefix
$\mathbf z_{\mathrm{prefix}}\in\{0,1\}^{\np-1}$ (possibly in superposition).
In order to determine the next bit, the coordinator performs a threshold test against the candidate value
\begin{align*}
  \mathbf z^{(\np)} \defeq \bigl(\mathbf z_{\mathrm{prefix}},\,1,\,0,\ldots,0\bigr)\in\{0,1\}^{\Np}.
\end{align*}
If $\gmax \ge \zdec\bigl(\mathbf z^{(\np)}\bigr)$, it sets the state associated with $\qpc(\np)$ to be $\ket{1}$; otherwise it sets $\ket{0}$.
Repeating this procedure for $\np=1,\ldots,\Np$ yields an $\Np$-bit approximation of $\gmax$.
This threshold test is implemented coherently by invoking the
amplitude-amplification operator~\cite{brassard2000} $\matr U_{\mathrm{AA},\mathsf N,\mathbf z^{(\np)}}$, which marks states that exceeds the threshold via a phase flip and amplifies their support.

The comparison ``$\gmax \ge \zdec(\mathbf z)$?'' is implemented via the amplitude-amplification operator
$\matr U_{\mathrm{AA},\mathsf N,\mathbf z}$ (Definition~\ref{def: the set of operators UAA}).
Recall that $\matr U_{\mathrm{AA},\mathsf N,\mathbf z}$ is built from $\matr U_{\mathrm{ini}}$ and a marking oracle
$\matr U_{\set Z_{\ge \mathbf z}}$ that phase-flips precisely those computational-basis states whose encoded local
values exceeds the threshold $\mathbf z$. The key point is that $\matr{U}_{\mathrm{AA},\mathsf N,\mathbf z}$ does  not need to know $\gmax$; instead, the inequality $\gmax \ge \zdec(\mathbf z)$ can be determined by checking whether the marked subspace has nonzero overlap with $\ket{\psi_{\rmini}}$.
Intuitively:
\begin{itemize}
\item if $\gmax < \zdec(\mathbf z)$, then no assignment satisfies the threshold, so the oracle marks nothing;
\item if $\gmax \ge \zdec(\mathbf z)$, then the marked subspace has nonzero weight under the prepared state,
and $\matr U_{\mathrm{AA},\mathsf N,\mathbf z}$ induces a nontrivial Grover-type rotation.
\end{itemize}
Thus, deciding whether $\gmax \ge \zdec(\mathbf z)$ reduces to testing whether the relevant eigenphase of
$\matr U_{\mathrm{AA},\mathsf N,\mathbf z}$ is zero or nonzero. 

In order to implement the above threshold test coherently within Algorithm~\ref{alg:dist-opt}, we invoke an approximately reversible \emph{controlled phase-test} unitary,
denoted by $$\mUpha{\matr U_{\mathrm{ini}},\,\matr U_{\mathrm{AA},\mathsf N,\mathbf z}}{p_{\min,\rmc},\delta}.$$
It acts on the joint registers
$$\bigl(\vqstn{\rmc,\np}, \vqB, \vqloc, \vqaux, \qpc(\np)\bigr)$$ and serves as a black-box decision primitive:
conditioned on the current candidate threshold $\mathbf z$, it writes a single decision bit to the designated output qubit
(here, $\qpc(\np)$), while uncomputing its workspace so that all other registers are (approximately) returned to their
pre-call states. We defer the circuit-level construction and error analysis of this primitive to
Appendix~\ref{apx:reversible_phase_estimation} and also a distributed implementation to Appendix~\ref{apx:implementationof_Upe}; the details are not needed to follow the main control flow in
Algorithm~\ref{alg:dist-opt}.

Note that the implementation of the $\np$-th iteration of Algorithm~\ref{alg:dist-opt} can be equivalently written as a block-diagonal controlled unitary
with respect to the already-determined prefix register $\vqpc(1:\np-1)$.
For each $\np\in[\Np]$, the network applies
\begin{align}
\sum_{\mathbf z_{\mathrm{prefix}}\in\{0,1\}^{\np-1}}
&\biggl(
\mUpha{\matr U_{\mathrm{ini}},\,\matr U_{\mathrm{AA},\mathsf N,\mathbf{z}^{(\np)}}}{p_{\min,\rmc},\,\delta}
\biggr)
\otimes
\ket{\mathbf z_{\mathrm{prefix}}}\bra{\mathbf z_{\mathrm{prefix}}}_{\vqpc(1:\np-1)}.
\label{eqn: controll phase estimation_main_alg}
\end{align}
where $\mathbf z^{(\np)}$ is the candidate threshold used at iteration $\np$.
The phase-test primitive writes its binary outcome to $\qpc(\np)$, thereby extending the prefix by one bit.
After $\Np$ iterations, measuring $\vqpc$ yields the final estimate $\mathbf z_{\max,\rmp,\rmc}$.

\subsection{Performance Guarantees}

The performance guarantees for Algorithm~\ref{alg:dist-opt} depend on the properties of the state $\ket{\psi_{\rmini}} = \matr{U}_{\mathrm{ini}}\svzero$, \ie,~\eqref{eqn: property of psi vect with positive magnitude: final} and~\eqref{eqn: bound of gstar and znsfn}.

\begin{theorem}[\textbf{Algorithm Performance Guarantee}]
\label{thm: overall_algorithm_properties}
Under the conditions in~\eqref{eqn: property of psi vect with positive magnitude: final} and~\eqref{eqn: bound of gstar and znsfn}, for a precision of $\Np$ bits and a failure parameter $\delta \in (0,1)$, Algorithm~\ref{alg:dist-opt} has the following performance guarantees:

\begin{enumerate}
    
    \item \label{prop: overall_query_complexity: final} \textbf{Query complexity.} The total number of applications of $\matr{U}_{\rmini}$ or $\matr{U}_{\rmini}^{\Herm}$ is
    \begin{align*}
        N_{\mathrm{queries}} &= \Np \cdot \Bigl( 2^{\,t(p_{\min,\rmc},\delta)+2} - 2 \Bigr)\in  O\Bigl( \Np \cdot (1-\delta)^{-\Np\Nsfn} \cdot \sqrt{|\setx_{\setVb}|} \Bigr),
    \end{align*} 
    where the qubit-counting function $t$ for phase estimation is defined to be
    \begin{align}
      t(c, \delta) \defeq
      \left\lceil -(1/2) \cdot \log_{2}\left(
          \delta \cdot c
      \right) - 1/2 \right\rceil. \label{eqn: def of fun t}
    \end{align}

    \item \label{prop: overall_success_probability: final} \textbf{Success probability and output.} With probability at least $(1-\delta)^{2\Np}$, a measurement of the coordinator's result register, $\vqpc$, yields the $\Np$-bit string $\mathbf{z}_{\max,\rmp,\rmc}$, which is the largest $\Np$-bit value not exceeding the sum of the local maxima:
    \begin{align}
        \begin{aligned}
            \mathbf{z}_{\max,\rmp,\rmc} \defeq &\arg\max_{\mathbf{z}_{\rmc}\in\{0,1\}^{\Np}} \quad \zdec(\mathbf{z}_{\rmc}) \\
            &\hspace*{1.2 cm} \text{s.t.} \quad \zdec(\mathbf{z}_{\rmc}) \le \sum_{\nsfn\in[\Nsfn]} &\zdec\bigl(\mathbf{z}_{\max,\rmp,\nsfn}\bigr)
         \end{aligned} \quad . \label{eqn: def_z_max_p_c_thm_main_overall}
    \end{align}
    Recall that $ \zdec $ is defined in~\eqref{eqn:zdec_inline}. 
    In order to ensure \((1-\delta)^{2\Np} \ge 1-\delta'\) for any $0<\delta'<1$, Bernoulli's inequality implies that it suffices to choose
    $
      \delta \le \delta'/(2\Np).
    $

    \item \label{prop: overall_approximation_precision} \textbf{Additive-error certificate.} The reported value $\mathbf{z}_{\max,\rmp,\rmc}$ is a certified lower bound on the true global maximum $\gmax$, satisfying:
    \begin{align}
        0 \le \gmax - \zdec\bigl(\mathbf{z}_{\max,\rmp,\rmc}\bigr) \le (\Nsfn+1) \cdot 2^{-\Np}.
    \label{eqn:prop_approx_gmax_revised}
    \end{align}
\end{enumerate}
\end{theorem}
\begin{proof}
See Appendix~\ref{apx: proof_of_thm_overall_algorithm_properties}.
\end{proof}
Theorem~\ref{thm: overall_algorithm_properties} thus characterizes both the complexity and the performance of Algorithm~\ref{alg:dist-opt}.
Part~\ref{prop: overall_query_complexity: final} states the quantum advantage of the algorithm in terms of the query complexity.
Part~\ref{prop: overall_success_probability: final} defines the output as the best $\Np$-bit approximation of the \textit{summed local solutions}, not the true global maximum.
Part~\ref{prop: overall_approximation_precision} connects this output to the true global maximum $\gmax$ by incorporating the local encoding error from~\eqref{eqn: bound of gstar and znsfn}.

\section{Distributed Construction of the Unitary \texorpdfstring{$\matr{U}_{\mathrm{ini}}$}{Uini}}
\label{sec:implementation_of_U_components}

This section details the distributed implementation of the state preparation unitary $\matr{U}_{\mathrm{ini}}$, a cornerstone of the main algorithm. 
A primary challenge in distributed quantum computing is parallelizing a computation across multiple processors while maintaining the global quantum coherence required for a collective speedup. 
The $\matr{U}_{\mathrm{ini}}$ operator is our solution to this challenge. It constructs the state $\ket{\psi_{\mathrm{initial}}}$ (defined in~\eqref{eqn: generated state vector}), which is a coherent superposition across all boundary configurations where each branch is entangled with the optimal solutions of the corresponding local sub-problems. 
This procedure, a quantum analogue of classically solving all sub-problems for every possible boundary assignment, is performed in a single, coherent operation, thereby preserving the potential for a global quantum advantage.

Our presentation follows a top-down approach. First, we introduce the generic algorithmic building block for function maximization. 
Second, we detail the distributed protocol that leverages this block to construct $\matr{U}_{\mathrm{ini}}$. 
Finally, we analyze the performance and resource costs of this construction.

\subsection{A Generic Subroutine for Quantum Maximization}
\label{sec: setup of quantum optimization algorithm}

The local optimization at each worker QPU is a standard quantum task: finding the maximum value of a function. 
Abstracting this task into a generic, reusable subroutine allows the analysis of local computational complexity to be separated from the complexities of distributed communication and coordination.

Consider the objective of finding the maximum value, $F^{\max} \defeq \max_{\mathbf{y} \in \{0,1\}^{N}} F(\mathbf{y})$, of an arbitrary function $F: \{0,1\}^{N} \to \sR$. 
A quantum algorithm for this task requires access to two types of operators.
\begin{enumerate}
    \item \textbf{State Preparation Oracle $\matr{U}_F$:} This unitary operation prepares an initial state $\ket{\psi_{F,\mathrm{ini}}}$ that has a non-zero overlap with the subspace of optimal solutions. 
    Specifically, we assume that
    \begin{align}
        \bra{\psi_{F,\mathrm{ini}}} (\matr{P}_{ \set{Y}^{\max} } \otimes \matr{I}_{\mathrm{rem}}) \ket{\psi_{F,\mathrm{ini}}} \geq p_{\min,F} > 0, \label{eqn: state preparation property}
    \end{align}
    where $\set{Y}^{\max} \defeq \arg\max\limits_{\mathbf{y} \in \{0,1\}^{N}} F(\mathbf{y})$ is the set of configurations that maximize $F$, the operator $\matr{P}_{\set{Y}^{\max}}$ is the projector onto this set, and $p_{\min,F}$ is a lower bound on the success probability. 
    As established in Appendix~\ref{apx:quantum_is_universal_for_classical}, which details how classical randomized algorithms can be unitarily implemented, the existence of such an oracle is justified. 

    \item \textbf{Marking Unitary Operator $\matr{U}_{\mathcal{Z}_{F \ge \mathbf{z}}}$:} This unitary, formally defined in Appendix~\ref{apx:generic_maximization_subroutine}, applies a phase of $-1$ to any computational basis state $\ket{\mathbf{y}}$ for which the function value $F(\mathbf{y})$ is greater than or equal to a specified threshold $\zdec(\mathbf{z})$.
\end{enumerate}
Equipped with these oracles, a generic subroutine, $\matr{U}_{\max}(\Np, \delta, p_{\min,F}, \matr{U}_F)$, exists to find an $\Np$-bit approximation of $F^{\max}$ with high probability. 
The detailed construction of this subroutine is provided in Appendix~\ref{apx:generic_maximization_subroutine}.
The performance of this generic subroutine is summarized in the following proposition.

\begin{proposition}[Performance of Generic Quantum Maximization]
\label{prop:generic_maximization_performance}
For an arbitrary function $F: \{0,1\}^{N} \to [0,1]$, the generic quantum maximization algorithm \\ $\matr{U}_{\max}(\Np, \delta, p_{\min,F}, \matr{U}_F)$ has the following properties:
\begin{enumerate}
    \item \textbf{Query Complexity:} The total number of queries to the oracles $\matr{U}_F$ and its inverse is $\Np \cdot (2^{t(p_{\min,F}, \delta)+2} - 2)$, where the function $t$,
    as defined in~\eqref{eqn: def of fun t}, determines the query precision.
    \item \textbf{Success Probability and Outcome:} By measuring the qubits that store the approximation, the algorithm produces the best $\Np$-bit lower approximation of $F^{\max}$, denoted $\mathbf{z}_{\max,\rmp,F}$, with a probability of at least $(1-\delta)^{2\Np}$.

    \item \textbf{Approximation Precision:} The final approximation satisfies the bound $0 \le F^{\max} - \zdec(\mathbf{z}_{\max,\rmp,F}) \le 2^{-\Np}$.
\end{enumerate}
\end{proposition}
\begin{proof}
See Appendix~\ref{apx:generic_maximization_subroutine}.
\end{proof}

\subsection{Distributed Protocol for Constructing \texorpdfstring{$\matr{U}_{\mathrm{ini}}$}{Uini}}
The construction of $\matr{U}_{\mathrm{ini}}$ is itself a distributed algorithm, built upon the generic quantum maximization subroutine from Proposition~\ref{prop:generic_maximization_performance}. 
This subroutine is instantiated at each worker processor with a set of local oracles, which we define first.

\begin{definition}[Local Oracles]
\label{def: oracles in the network}
For each subgraph $\nsfn \in [\Nsfn]$ and each fixed configuration $\mathbf{x}_{\setVbnsfn}$ of its relevant boundary variables, we assume the corresponding worker processor $\QPU_{\nsfn}$ can implement a local state preparation oracle, denoted $\matr{U}_{g_{\nsfn,\mathbf{x}_{\setVbnsfn}}}$.
This unitary is a specific instance of the generic oracle $\matr{U}_F$ from Proposition~\ref{prop:generic_maximization_performance}, where the function $F$ is the conditioned local objective $g_{\nsfn,\mathbf{x}_{\setVbnsfn}}$.
Acting on the local registers for internal variables ($\vqsrn{\nsfn}$) and oracle ancillas ($\vqoran{\nsfn}$), both initialized to $\svzero$, its action is to prepare an initial state:
  \begin{align}
  \matr{U}_{g_{\nsfn,\mathbf{x}_{\setVbnsfn}}} \cdot \left( \svzero_{\vqsrn{\nsfn}} \otimes \svzero_{\vqoran{\nsfn}} \right) = \ket{ \psi_{\nsfn,\mathrm{ini}} }.
  \end{align}
This prepared state must have a non-trivial overlap with the subspace of optimal solutions. We assume a minimum success probability $p_{\min,\nsfn} > 0$ such that for all $\mathbf{x}_{\setVbnsfn}$:
\begin{align}
\bra{ \psi_{\nsfn,\mathrm{ini}} } \bigl( \matr{P}_{ \set{Y}_{\nsfn,\mathbf{x}_{\setVbnsfn}}^{\max} } \otimes \matr{I}_{\vqoran{\nsfn}} \bigr) \ket{ \psi_{\nsfn,\mathrm{ini}} } \geq p_{\min,\nsfn},
\end{align}
  where $\set{Y}_{\nsfn,\mathbf{x}_{\setVbnsfn}}^{\max} \defeq \arg\max\limits_{ \mathbf{x}_{\set{V}_{\nsfn}}} g_{\nsfn,\mathbf{x}_{\setVbnsfn}}( \mathbf{x}_{\set{V}_{\nsfn}} )$ is the set of optimal local configurations, and $\matr{P}_{ \set{Y}_{\nsfn,\mathbf{x}_{\setVbnsfn}}^{\max} }$ is the projector onto that subspace (formally defined in Definition~\ref{def:marking_operators}, Appendix~\ref{apx:basic_notation}).
\end{definition}


\begin{algorithm}[t!]
\caption{Distributed Protocol for Implementing $\matr{U}_{\mathrm{ini}}$}
\label{alg:u_ini_protocol}
\begin{algorithmic}[1]
\Statex \textbf{Stage 1: Distribute Boundary Information (at Coordinator $\QPUc$)}
\State Initialize a uniform superposition over all boundary configurations in register $\vqB$.
\For{each worker $\QPU_{\nsfn} \in [\Nsfn]$}
\State Coherently fan out boundary bits by entangling $\vqB$ with the auxiliary register $\vqBauxn$ via CNOTs.
\State Transmit the register $\vqBauxn$ to $\QPU_{\nsfn}$ via quantum teleportation.
\EndFor
\Statex \textbf{Stage 2: Parallel Local Maximization (at all Workers $\QPU_{\nsfn}$)}
\State In parallel, each worker executes the generic maximization subroutine (Proposition~\ref{prop:generic_maximization_performance}) on its local function, conditioned on the state of its received register $\vqBauxn$. The $\Np$-bit result is stored in the local qubit register $\vqpn{\nsfn}$.
\Statex \textbf{Stage 3: Aggregate Quantum Results (at Coordinator $\QPUc$)}
\For{each worker $\QPU_{\nsfn} \in [\Nsfn]$}
\State Receive the registers $(\vqBauxn, \vqpn{\nsfn})$ from $\QPU_{\nsfn}$ via quantum teleportation.
\EndFor
\State Reverse the entanglement operations from Stage 1 to disentangle $\vqB$ and return the auxiliary registers $(\vqBauxn)_{\nsfn}$ to the $\svzero$ state.
\end{algorithmic}
\end{algorithm}

With the local oracles defined, we present the distributed protocol for implementing the unitary operation $\matr{U}_{\mathrm{ini}}$ in Algorithm~\ref{alg:u_ini_protocol}, which acts on the qubits specified in Appendix~\ref{apx:registers}.
The protocol's core is the parallel local maximization (Stage 2), formally described by the application of the following controlled operation at each worker $\QPU_{\nsfn}$:
\begin{align}
\sum_{\mathbf{x}_{\setVbnsfn} \in \{0,1\}^{|\setVbnsfn|}} &\ket{\mathbf{x}_{\setVbnsfn}}\bra{\mathbf{x}_{\setVbnsfn}}_{\vqBauxn}  \otimes \matr{U}_{\max}\bigl(\Np, \delta, p_{\min,\nsfn}, \matr{U}_{g_{\nsfn,\mathbf{x}_{\setVbnsfn}}} \bigr).
\label{eqn: Uini_distributed_optimization_impl_alg}
\end{align}
This expression describes a coherent, parallel execution: for each computational basis state $\ket{\mathbf{x}_{\setVbnsfn}}$ of the received register $\vqBauxn$, the generic maximization unitary $\matr{U}_{\max}$ (from Proposition~\ref{prop:generic_maximization_performance}) is applied to the local registers of $\QPU_{\nsfn}$.
The inverse operation, $\matr{U}_{\mathrm{ini}}^{\Herm}$, is constructed by reversing the stages of Algorithm~\ref{alg:u_ini_protocol} and taking the Hermitian conjugate of each unitary.

\subsection{Performance Analysis of the \texorpdfstring{$\matr{U}_{\mathrm{ini}}$}{Uini} Construction}
The resource costs and formal properties of the $\matr{U}_{\mathrm{ini}}$ construction (Algorithm~\ref{alg:u_ini_protocol}) are summarized below.

\begin{lemma}[Performance of $\matr{U}_{\mathrm{ini}}$ Construction]
\label{lem: properties_approximation_of_oracle_Uini}
The distributed implementation of $\matr{U}_{\mathrm{ini}}$ has the following properties:
\begin{enumerate}
\item \textbf{Query Complexity:} One application of $\matr{U}_{\mathrm{ini}}$ requires a total of
\begin{align}
\sum_{\nsfn \in [\Nsfn]} \Np \cdot (2^{t(p_{\min,\nsfn}, \delta)+2} - 2)
\in O\Biggl( \Np \cdot \sum_{\nsfn}\sqrt{|\set{X}_{\set{V}_{\nsfn}}|} \Biggr)
\label{eqn: Uini_query_complexity_of_oracle_lemma}
\end{align}
queries to the fundamental oracles in $\{ \matr{U}_{g_{\nsfn,\mathbf{x}_{\setVbnsfn}}} \}_{\nsfn,\, \mathbf{x}_{\setVbnsfn}}$ and their inverse across all worker QPUs.

\item \textbf{EPR Pair Consumption:} The protocol consumes a total of $ \sum_{\nsfn \in [\Nsfn]} (2 |\setVbnsfn| + \Np) $ EPR pairs for quantum teleportation.

\item \textbf{State Preparation Guarantees:} The resulting state $\ket{\psi_{\rmini}}$ satisfies the required overlap and accuracy conditions given in~\eqref{eqn: property of psi vect with positive magnitude: final} and~\eqref{eqn: bound of gstar and znsfn}, which are necessary for the main algorithm's performance (detailed in Theorem~\ref{thm: overall_algorithm_properties}).
\end{enumerate}
\end{lemma}
\begin{proof}
See Appendix~\ref{apx: properties_approximation_of_oracle}.
\end{proof}

\section{Performance Guarantees and Resource Analysis}
\label{sec:overall_performance}

This section synthesizes the results from the preceding sections to provide a complete, end-to-end performance analysis of our distributed quantum algorithm, $A_{\mathrm{dist}}$.
Section~\ref{sec: distributed_quantum_algorithm} established the algorithm's theoretical power assuming access to the state preparation unitary $\matr{U}_{\mathrm{ini}}$. Recall that Section~\ref{sec:implementation_of_U_components} then detailed the distributed protocol for constructing this unitary.
We now integrate the abstract performance of the main algorithm with the concrete implementation costs (derived in Lemma~\ref{lem: properties_approximation_of_oracle_Uini} and Appendix~\ref{apx:implementationof_Upe}) to quantify the framework's true computational complexity and resource requirements.
This analysis establishes a holistic view of the algorithm's practicality and scalability.


\textbf{Setup}:
We analyze the distributed algorithm $A_{\mathrm{dist}}$ from Algorithm~\ref{alg:dist-opt}.
This analysis assumes each constituent unitary ($\matr{U}_{\mathrm{ini}}$, $\matr{U}_{\mathrm{PE}}$, and $\matr{U}_{\mathrm{check}}$) is implemented by the distributed protocols detailed in Section~\ref{sec:implementation_of_U_components} and Appendix~\ref{apx:implementationof_Upe}.
The performance guarantees depend on the precision parameters $\delta$ and $\Np$, the graph decomposition ($\{\graphG_{\mathrm s,\nsfn}\}_{\nsfn}$, $\setVb$), and the minimum success probabilities $p_{\min,\mathrm{c}}$ and $p_{\min,\nsfn}$.

\begin{theorem}[\textbf{Overall Algorithm Performance}]
\label{thm: overall_algorithm_performance}
The algorithm $A_{\mathrm{dist}}$, executed under the previously specified Setup, has the following resource requirements and performance guarantees.
\begin{enumerate}
    \item \label{prop:total_qubits_count_overall}
    \textbf{Total qubit count.}
    The total number of qubits in the network, $N_{\mathrm{qubits}}$, equals the sum of the register sizes in Table~\ref{table: number of qubits} (Appendix~\ref{apx:registers}).

    \item \label{prop:qubits_per_processor_thm_overall_alg_perf}
    \textbf{Qubits per processor.}
    The coordinator qubit count is
    \begin{align*}
            N_{\mathrm{qubits}}(\QPUc)
            &= |\setVb|
            + t(p_{\min,\mathrm{c}},\delta) + \Np + \Nsfn\cdot(\Np+1)
            + \sum_{\nsfn\in[\Nsfn]} |\setVbnsfn| \\&\quad+ \Nsfn\cdot t(p_{\min,\mathrm{c}},\delta) + 3
            \\&
            \in O\Bigl(|\set{V}|+\mathrm{Poly}\bigl(\Nsfn,\Np,\log(\delta) \bigr) \Bigr).
    \end{align*}
    The worker qubit count (for each $\QPU_{\nsfn}$) is
    \begin{align*}
            N_{\mathrm{qubits}}(\QPU_{\nsfn})
            &= |\set{V}_{\nsfn}| + |\setVbnsfn|
            + \Np + t(p_{\min,\nsfn},\delta)
            + N_{\mathrm{ora},\nsfn} + t(p_{\min,\mathrm{c}},\delta) + 1.
    \end{align*}
    In particular, the dominant, problem-dependent term at worker $\nsfn$ scales with the local instance size
    $|\set{V}_{\nsfn}|$ and its boundary interface $|\setVbnsfn|$.

    \item \label{prop:overall_query_complexity_main_thm}
    \textbf{Overall query complexity.}
    The total number of queries to the leaf oracles \\ $\bigl\{ \matr{U}_{g_{\nsfn,\mathbf{x}_{\setVbnsfn}}} \bigr\}_{\nsfn \in [\Nsfn],\, \mathbf{x}_{\setVbnsfn} \in \set{X}_{\setVbnsfn}}$ and their inverse (Definition~\ref{def: oracles in the network}) is
    \begin{align}
        C_{\mathrm{distr}}(\graphG,\mathsf{N},\delta,\Np)
        &= \Np^{2} \cdot \bigl( 2^{\,t(p_{\min,\mathrm{c}},\delta)+2} - 2 \bigr)
        \cdot \sum_{\nsfn\in[\Nsfn]}  \bigl( 2^{\,t(p_{\min,\nsfn},\delta)+2} - 2 \bigr)
        \nonumber\\&\in 
        O\Biggl( \! \Np^{2} \cdot (1-\delta)^{-\Np\Nsfn} 
        \cdot \sum_{\nsfn\in[\Nsfn]} \sqrt{|\set{X}_{\setVb}| \cdot |\set{X}_{\setV_{\nsfn}}|} \Biggr),
        \label{eqn: overall_query_complexity_discussion_revised}
    \end{align}
    which reflects a nested search: a global stage scaling with the boundary search space $|\set{X}_{\setVb}|$,
    multiplied by parallel local stages scaling with $|\set{X}_{\setV_{\nsfn}}|$.

    \item \label{prop:overall_EPR_consumption_main_thm}
    \textbf{Entanglement and communication.}
    For a single-hop topology (the coordinator is directly connected to all workers), the total number of EPR pairs consumed is
    \begin{align}
            N_{\mathrm{EPR}}
            &= \Np \cdot \bigl( 2^{\,t(p_{\min,\mathrm{c}},\delta)+2} - 2 \bigr)
            \cdot 
            \sum_{\nsfn\in[\Nsfn]} \bigl( 2\,|\setVbnsfn| + \Np \bigr)
            \nonumber\\&\quad+
            \Np \cdot \bigl(2^{t(p_{\min,\mathrm{c}}, \delta)+2} - 4\bigr) \cdot \Nsfn
            +
            \Np \cdot t(p_{\min,\rmc}, \delta) \cdot (2\Nsfn)
            \nonumber\\&\in O\Bigl(\Np \cdot (1-\delta)^{-\Np\Nsfn} \cdot |\setV| \cdot \sqrt{|\set{X}_{\setVb}|}\Bigr).
        \label{eqn:total_epr_consumption}
    \end{align}
\end{enumerate}
\end{theorem}
\begin{proof}
A detailed proof is provided in Appendix~\ref{apx: overall_algorithm_performance}.
\end{proof}

We make three remarks that clarify the interpretation of Theorem~\ref{thm: overall_algorithm_performance}.

\textbf{Qubit scaling.} The theorem cleanly separates the coordinator and worker memory costs. The coordinator's dominant problem-dependent load is governed by the boundary size $|\setVb|$ and the number of workers $\Nsfn$, whereas worker $\QPU_{\nsfn}$ mainly scales with the size of its local subproblem, $|\set{V}_{\nsfn}|$. The partitioning trade-off is therefore explicit: a finer decomposition can reduce per-worker qubits, but only at the price of a larger coordination layer at $\QPUc$. In this sense, the separator is the key structural parameter that determines whether the problem can be distributed across small QPUs. Unlike circuit-cutting approaches, this reduction in per-QPU qubits does not rely on exponential classical reconstruction~\cite{Peng2020,Piveteau2024}.

\textbf{Query complexity.} The product form in~\eqref{eqn: overall_query_complexity_discussion_revised} reflects a nested search: an outer search over boundary assignments and, for each candidate boundary, parallel local maximization subroutines. The theorem therefore identifies two quantities that control the cost: the size of the boundary search space and the difficulty of the induced local problems. This is precisely where structure enters the analysis. Better partitions improve $p_{\min,\rmc}$ and $p_{\min,\nsfn}$ and thereby reduce both global and local costs. Compared with black-box search-space partitioning, which typically yields $O(\sqrt{KN})$ query complexity on $K$ processors~\cite{Avron2021}, the present scheme retains the Grover-type advantage up to explicit structure-dependent factors.

\textbf{Communication cost.} Equation~\eqref{eqn:total_epr_consumption} isolates the entanglement cost of maintaining coherence across the network. For one call to $\matr{U}_{\mathrm{ini}}$, the term $\sum_{\nsfn} (2|\setVbnsfn|+\Np)$ counts the EPR pairs needed to distribute the boundary information and return the $\Np$-bit local outputs. The theorem therefore makes the communication overhead explicit in the cut sizes $|\setVbnsfn|$. This is important because query complexity and communication cost are distinct: an instance may have favorable oracle complexity while still being limited by network resources. Compared with circuit-cutting methods, the overhead here remains natively quantum and directly interpretable in terms of the chosen decomposition~\cite{Peng2020,Piveteau2024}.


\subsection{Impact of Network Topology on Communication Overhead}
\label{sec:generalization_and_resource_analysis}

The preceding analysis in Theorem~\ref{thm: overall_algorithm_performance} assumes an idealized, single-hop network where the coordinator can directly communicate with every worker.
Realistic DQC systems, however, will be built on physical networks with specific, and often constrained, topologies~\cite{Caleffi2024,Barral2025}.
In such scenarios, communication between non-adjacent processors relies on protocols such as \emph{entanglement swapping}. Entanglement swapping consumes entangled resources at each intermediate node (or ``hop'') to establish a long-distance entangled pair~\cite{ZelseZfiukowski1993, Briegel1998}, making communication a primary performance bottleneck.

We now generalize our analysis to an arbitrary connected network topology.
We show that the algorithm's computational complexity remains invariant, but its communication cost scales linearly with the network's diameter.
This result cleanly decouples the computational logic from the communication overhead and provides a quantitative tool for the co-design of distributed quantum algorithms and the physical networks on which they run.

\textbf{Setup}:
Let $A_{\mathrm{dist}}(\graphG, \mathsf{N}, \delta, \Np)$ be the algorithm operating on an idealized single-hop (fully connected) network $\mathsf{N}$ with $\Nsfn+1$ processors.
Let $\mathsf{N}'(\set{V}_{\mathsf{N}'}, \set{E}_{\mathsf{N}'})$ be any other quantum network that is a connected graph with at least $\Nsfn+1$ processors.
We construct a corresponding algorithm, $A_{\mathrm{dist}}'(\graphG, \mathsf{N}', \delta, \Np)$, by mapping the coordinator $\mathrm{c}$ and workers $\nsfn$ to distinct nodes in $\mathsf{N}'$. All communication between non-adjacent nodes is implemented using multi-hop entanglement swapping.

\begin{proposition}[\textbf{Performance on General Connected Networks}]
\label{prop:generalization_to_arbitrary_topologies}
The algorithm $A_{\mathrm{dist}}'$ on the general network $\mathsf{N}'$ has the following properties:
\begin{enumerate}
    \item \textbf{Performance Invariance.}
    The algorithm $A_{\mathrm{dist}}'$ achieves the same output guarantees as $A_{\mathrm{dist}}$.
    Its total qubit count, query complexity to local oracles, success probability, and approximation precision are identical to those stated in Theorem~\ref{thm: overall_algorithm_performance}.

    \item \textbf{Communication Cost Scaling.}
    The total number of EPR pairs consumed by $A_{\mathrm{dist}}'$, denoted $N'_{\mathrm{EPR}}$, is upper-bounded by
    \begin{align*}
        N'_{\mathrm{EPR}} \le d(\mathsf{N}') \cdot N_{\mathrm{EPR}},
    \end{align*}
    where $N_{\mathrm{EPR}}$ is the consumption for the idealized network from \eqref{eqn:total_epr_consumption}.
    The term $d(\mathsf{N}')$ is the coordinator-centric diameter of the network $\mathsf{N}'$, defined as the maximum shortest-path distance from the coordinator to any worker:
    \begin{align*}
        d(\mathsf{N}') \defeq \max_{\nsfn \in [\Nsfn]} \bigl\{ 
        \mathrm{dist}_{\mathsf{N}'}(\mathrm{c}, \nsfn) \bigr\},
    \end{align*}
    where $\mathrm{dist}_{\mathsf{N}'}(\mathrm{c}, \nsfn)$ is the number of \emph{edges} in the shortest path connecting nodes $\mathrm{c}$ and $\nsfn$ in $\mathsf{N}'$.
\end{enumerate}
\end{proposition}
\begin{proof}
See Appendix~\ref{apx:generalization_to_arbitrary_topologies}.
\end{proof}

\begin{remark}
This proposition establishes the network diameter $d(\mathsf{N}')$ as a direct, linear ``penalty factor'' on the required entanglement resources. This leads to several practical design principles for DQC systems:
\begin{enumerate}
    \item \textbf{Strategic Coordinator Placement.} In order to minimize communication overhead, the coordinator role should be assigned to a processor at a graph-theoretically central node of the network (\ie, one that minimizes the maximum distance to all other nodes).
    \item \textbf{Topology-Aware Task Allocation.} For a fixed physical network, worker roles that are more communication-intensive (\eg, those with larger boundary sets $|\setVbnsfn|$), are suggested to be assigned to physically closer QPUs to minimize $d(\mathsf{N}')$.
    \item \textbf{Network Design Principles.} This analysis provides a quantitative basis for preferring low-diameter topologies (\eg, star or fully-connected) over high-diameter ones (\eg, line) when designing dedicated DQC clusters.
\end{enumerate}
The analysis in this paper assumes perfect, deterministic entanglement swapping.
In practice, entanglement swapping is probabilistic and degrades fidelity~\cite{Wehner2018}.
The success probability and fidelity of an end-to-end entangled link decrease exponentially with the path length, which would further penalize high-diameter networks.
This provides a rich avenue for future research into the co-design of fault-tolerant distributed algorithms and realistic network topologies.
\eremark
\end{remark}

\subsection{Gate-Level State Vector Simulation}

\begin{figure}[t]
\centering
\includegraphics[width=0.6\linewidth]{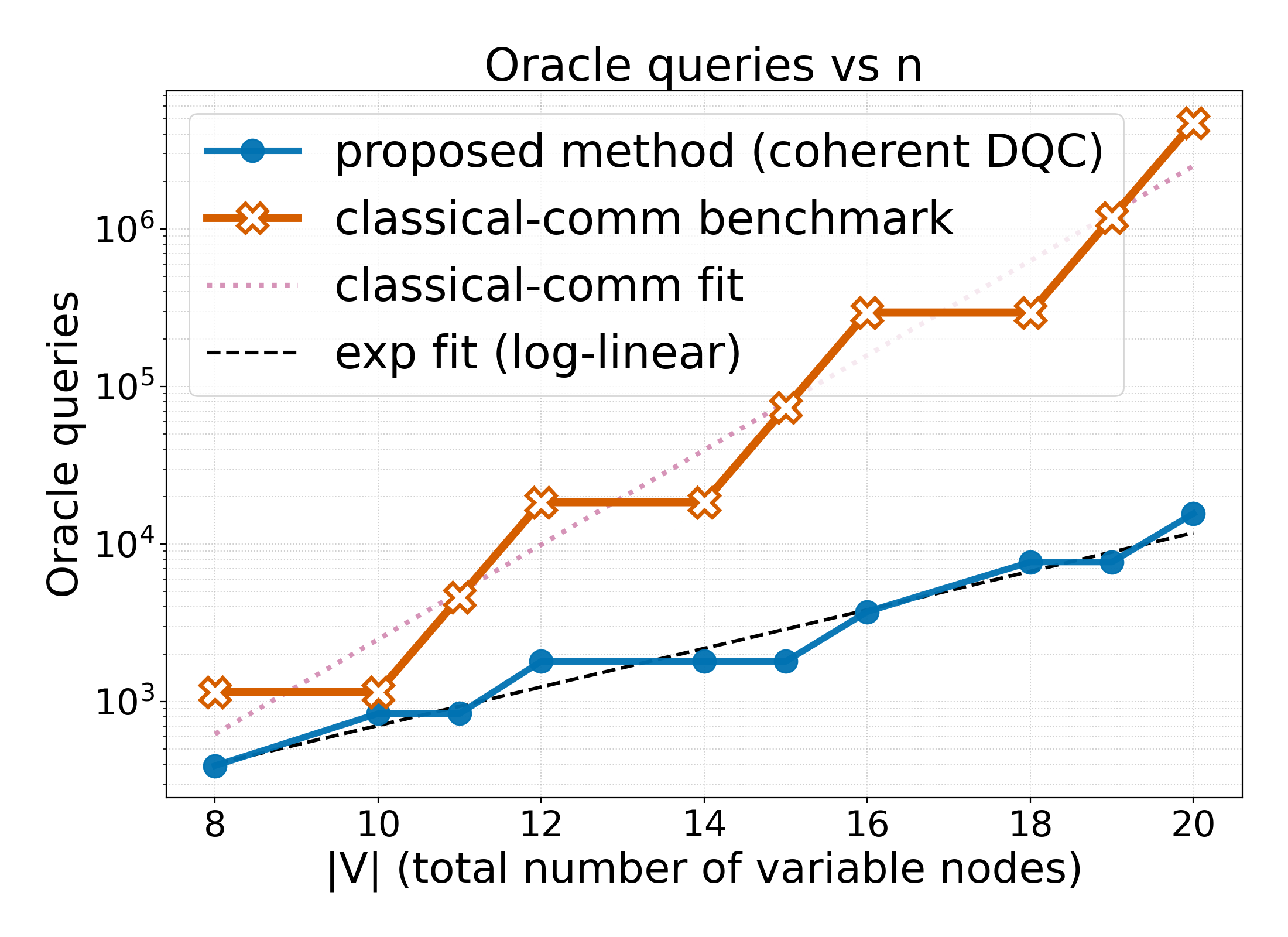}
\caption{Total number of leaf-oracle evaluations per execution of the distributed primitive versus the number of variables $|\setV|$. The $y$-axis is logarithmic. The dashed curve is a log-linear fit (linear in $\log_{10}$).}
\label{fig:queries_vs_n_gate_main}
\end{figure}

\begin{figure}[t]
\centering
\includegraphics[width=.6\linewidth]{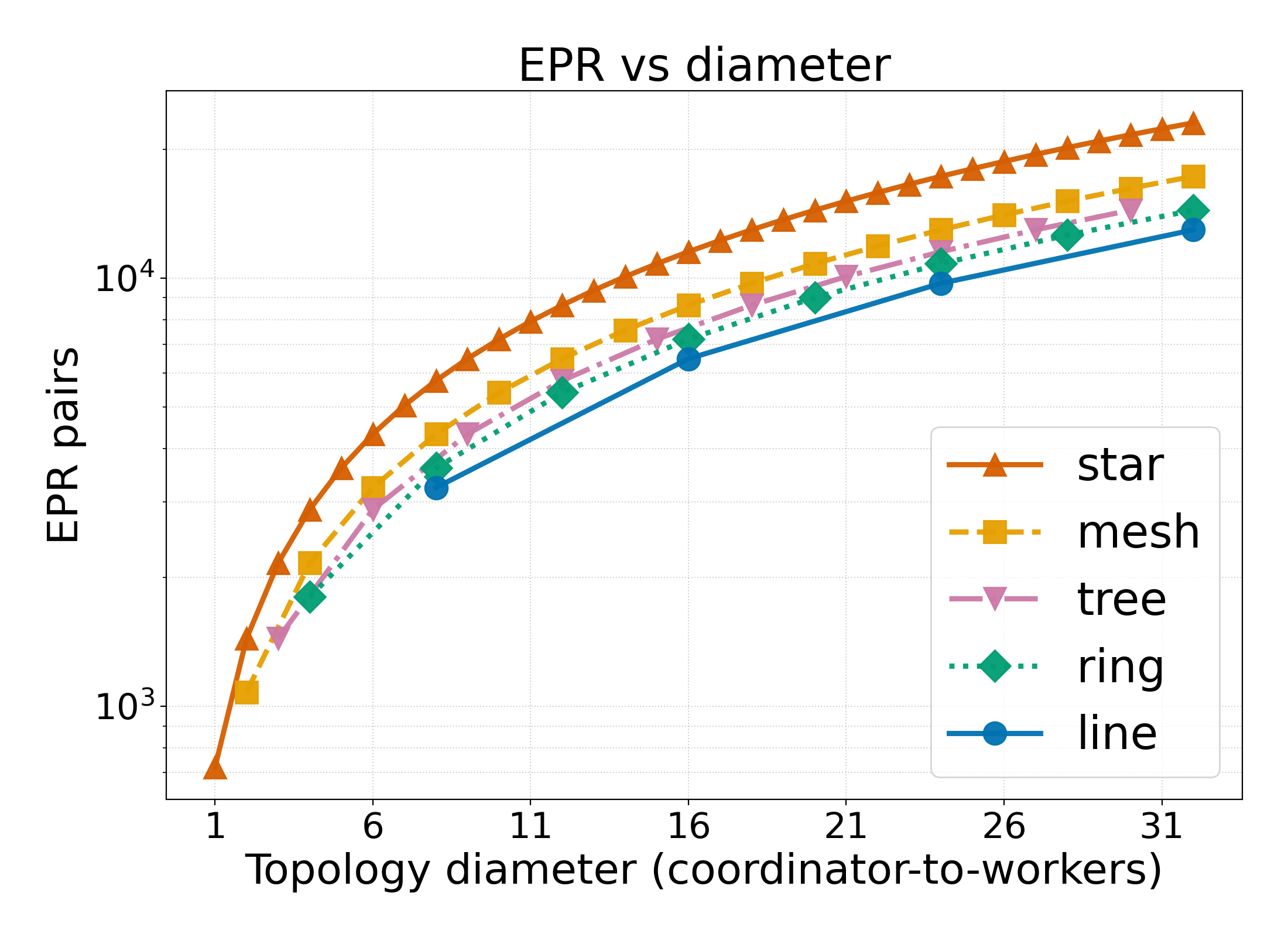}
\caption{Total number of elementary EPR pairs consumed versus the coordinator-to-workers diameter (hops) under topology stretching. The $y$-axis is logarithmic.}
\label{fig:epr_vs_diameter_gate_main}
\end{figure}

Figures~\ref{fig:queries_vs_n_gate_main}--\ref{fig:epr_vs_diameter_gate_main} provide an end-to-end gate-level
state vector validation of the distributed primitives analyzed in
Theorem~\ref{thm: overall_algorithm_performance} and
Proposition~\ref{prop:generalization_to_arbitrary_topologies}.
In order to illustrate the speed-up of our proposed algorithm, we report two full-execution counters:
(i) total number of leaf-oracle evaluations (\ie, oracle queries), and
(ii) total number of elementary EPR pairs consumed by teleportation-based coordination~\cite{Bennett1993}.

\textbf{Query scaling.}
As shown in Fig.~\ref{fig:queries_vs_n_gate_main},
the oracle-query count increases with the number of variable nodes $|\setV|$
and exhibits plateau--jump transitions.
In this sweep, the decomposition scales with problem size according to
$|\setV| = |\setVb| + N_G \cdot r$ (with fixed $N_G$), so both the number of boundary nodes $|\setVb|$ and the per-worker partition size $r$ increase with $|\setV|$ (equivalently, both $2^{|\setVb|}$ and $2^r$ grow).
The jumps occur when the integer phase-estimation register length increases by one
(a ceiling of a logarithmic bound, via the function $t$ in~\eqref{eqn: def of fun t}),
which adds a controlled layer and increases the oracle-query counter by a constant multiplicative factor, as shown in Theorem~\ref{thm: overall_algorithm_performance}.
For comparison, we also plot a classical-communication benchmark (\textit{cf.}~\cite{Avron2021}) in which workers run local quantum maximum-finding and coordinate only through classical messages.
Our distributed quantum algorithm empirically uses fewer worker-oracle queries across the sweep.
This behavior is consistent with Theorem~\ref{thm: overall_algorithm_performance}, where global coherence avoids repeated boundary-conditioned local searches.

\textbf{Topology-driven entanglement cost.}
In Fig.~\ref{fig:epr_vs_diameter_gate_main}, we plot EPR consumption as the coordinator-to-workers hop diameter increases under different topologies.
For all topologies, total EPR consumption grows with diameter.
The EPR counter aggregates routed hop distances in our coordinator--worker entanglement-preserving execution mode, so different topologies can yield different totals even at the same diameter.
(Note that Proposition~\ref{prop:generalization_to_arbitrary_topologies} provides a worst-case bound in terms of diameter, but does not characterize performance gaps among specific topologies.)

Appendix~\ref{sec:gate_level_validation} reports the complete simulation setup and additional results.

\subsection{Finding an Optimal Configuration}
\label{sec:finding_x_max}

The primary algorithm $A_{\mathrm{dist}}$ finds an approximation of the maximum value $\gmax$. However, many applications require the specific configuration $\vxmax$ that produces this value. We now introduce an extension, $A_{\mathrm{dist,config}}$, that finds a specific configuration $\hat{\vx}$ whose value $g(\hat{\vx})$ is provably close to the true maximum.

The algorithm employs iterative self-reduction, a standard technique in complexity theory~\cite{Arora2009}. It determines the optimal assignment for each variable $x_i \in \set{V}$ one by one. In order to decide the value for a variable $x_i$, the main algorithm $A_{\mathrm{dist}}$ runs twice: once with $x_i$ fixed to $0$ and once with $x_i$ fixed to $1$. By comparing the two resulting value-approximations, the algorithm infers the optimal assignment for $x_i$, fixes that value, and proceeds to the next variable.

\textbf{Setup}:
We use the primary algorithm $A_{\mathrm{dist}}(\graphG, \mathsf{N}, \delta, \Np)$ (from Algorithm~\ref{alg:dist-opt}) as a subroutine. This algorithm functions as an oracle that, with a success probability of at least $(1 - \delta)^{2\Np}$, outputs an approximation of $\gmax$ that satisfies the error bound in~\eqref{eqn:prop_approx_gmax_revised}.

\begin{theorem}[\textbf{Finding an Optimal Configuration}]
\label{prop: for an oracle solving the maximum how to find an oracle finds the location}
Given the Setup, there exists a distributed quantum algorithm $A_{\mathrm{dist,config}}(\graphG, \mathsf{N}, \delta, \Np)$ that finds a configuration $\hat{\vx} \in \set{X}$ and has the following performance guarantees:
\begin{enumerate}

    \item \textbf{Query Complexity.} The total query complexity to the fundamental local oracles is upper bounded by $(2|\set{V}| +1) \cdot C_{\mathrm{distr}}(\graphG, \mathsf{N},\delta,\Np) \in O(2^{|\set{V}|/2})$, where $C_{\mathrm{distr}}(\graphG, \mathsf{N},\delta,\Np)$ is given in~\eqref{eqn: overall_query_complexity_discussion_revised}.
    
    \item \textbf{Success Probability.} The algorithm succeeds with probability at least $(1 - \delta)^{4\Np|\set{V}| + 2\Np}$. In order to ensure \((1-\delta)^{4\Np|\set{V}| + 2\Np} \ge 1-\delta'\) for any $0<\delta'<1$,  Bernoulli's inequality implies that it suffices to choose
    $
      \delta \le \delta'/(4\Np|\set{V}| + 2\Np).
    $
    
    \item \textbf{Configuration Accuracy.} The output $\hat{\vx}$ satisfies
    \begin{align}
      |g(\hat{\vx}) - \gmax | \leq 4 \cdot |\set{V}| \cdot (\Nsfn + 1) \cdot 2^{-\Np}.
    \label{eqn:prop_approx_x_accuracy_revised}
    \end{align}
\end{enumerate}
\end{theorem}
\begin{proof}
See Appendix~\ref{apx: for an oracle solving the maximum how to find an oracle finds the location}.
\end{proof}

\begin{remark}
The performance of $A_{\mathrm{dist,config}}$ scales linearly with the number of variables, $|\set{V}|$, relative to the base algorithm $A_{\mathrm{dist}}$. This $O(|\set{V}|)$ factor arises because the self-reduction technique must iterate through the variables, invoking the $A_{\mathrm{dist}}$ subroutine $2|\set{V}|+1$ times to make its decisions. The error bound \eqref{eqn:prop_approx_x_accuracy_revised} scales linearly with $|\set{V}|$ as the small approximation error from each of the $|\set{V}|$ decision steps can accumulate. The success probability compounds exponentially for this same reason.
\eremark
\end{remark}

\section{Divide-and-Conquer: A Hierarchical Framework for Scalable Distributed Optimization}
\label{sec:hierarchical_algorithm}

The distributed algorithm $A_{\mathrm{dist}}$, presented in Section~\ref{sec: distributed_quantum_algorithm}, provides a robust method for solving optimization problems on factor graphs.
However, its applicability is bounded by a core assumption: each decomposed subproblem must fit within the resource limits of a single worker QPU.
As problem sizes grow, this assumption inevitably fails, presenting a hard barrier to scalability.
The qubit requirement for a sub-problem, established in Theorem~\ref{thm: overall_algorithm_performance}, scales with its size.
If this requirement exceeds the hardware's physical limit, $N_{\mathrm{QPU},\max}$, the sub-problem becomes intractable for one device.

In order to bridge this gap between large-scale problems and the practical constraints of quantum networks, we introduce a hierarchical framework based on the divide-and-conquer paradigm~\cite{Blahut2010}.
This framework acts as a recursive ``quantum compiler,'' systematically mapping a large logical problem onto a network of fixed-capacity quantum processors.
It treats any sub-problem exceeding hardware limits as a new, smaller optimization task to be solved by its own dedicated sub-network.
This strategy transforms our approach from a specialized tool into a universal paradigm, enabling it to tackle problems of arbitrary size by recursively applying the same distributed logic.

Recognizing the diverse capabilities of quantum hardware, we formalize two distinct execution modes for this framework.
These modes create a tunable trade-off, \ie, a ``performance dial'', between quantum advantage and hardware fidelity.
\begin{itemize}
    \item \textbf{Mode 1: Coherent Cascade.} A fully coherent execution that preserves the global Grover-like speedup across the entire problem. This represents a powerful theoretical model for future fault-tolerant networks.
    \item \textbf{Mode 2: Hybrid Approach.} A hybrid quantum-classical execution that introduces intermediate measurements to reduce circuit depth. This practical strategy mitigates noise in near-term (NISQ) devices, albeit at the cost of increased query complexity.
\end{itemize}

\subsection{The Hierarchical Strategy}
The hierarchical strategy transforms the flat coordinator-worker topology into a recursive, tree-like structure.
The process involves two main phases: a top-down hierarchical decomposition (a classical pre-processing step) and a corresponding bottom-up hierarchical execution (the distributed quantum computation).

\textbf{Top-Down Hierarchical Decomposition}.
In a classical pre-processing step, we recursively partition the problem until all computational tasks fit within the capacity of individual QPUs.
This process is formalized in Algorithms~\ref{alg:hierarchical_decomposition}, which establishes a consistent multi-index notation to track subgraphs across different levels of the hierarchy.

\begin{algorithm}[t!]
\caption{Top-Down Construction of the Hierarchical Decomposition}
\label{alg:hierarchical_decomposition}
\begin{algorithmic}[1]
\Procedure{Decompose}{$\graphG^{(\ell)}_{\vk}$, $\ell$, $N_{\mathrm{c},\max}$, $N_{\mathrm{w},\max}$}
    \State Using the qubit bounds in Theorem~\ref{thm: overall_algorithm_performance}, compute the coordinator and worker qubit requirements for $\graphG^{(\ell)}_{\vk}$, denoted by $N_{\mathrm{c},\mathrm{qubits}}(\vk)$ and $N_{\mathrm{w},\mathrm{qubits}}(\vk)$.
    \If{$N_{\mathrm{c},\mathrm{qubits}}(\vk) \le N_{\mathrm{c},\max}$ \textbf{and} $N_{\mathrm{w},\mathrm{qubits}}(\vk) \le N_{\mathrm{w},\max}$}
        \State Mark $\graphG^{(\ell)}_{\vk}$ as a leaf at level $\ell$.
        \State \Return
    \Else
        \State Choose an internal boundary-variable set $\set{V}^{(\ell+1)}_{\mathrm{B},\vk}$.
        \State Partition $\graphG^{(\ell)}_{\vk}$ into child subgraphs $\{\graphG^{(\ell+1)}_{\vk,k_{\ell+1}}\}_{k_{\ell+1}\in\mathcal{K}_{\ell+1}(\vk)}$.
        \For{each $k_{\ell+1}\in\mathcal{K}_{\ell+1}(\vk)$}
            \State \Call{Decompose}{$\graphG^{(\ell+1)}_{\vk,k_{\ell+1}}$, $\ell+1$, $N_{\mathrm{c},\max}$, $N_{\mathrm{w},\max}$}
        \EndFor
    \EndIf
\EndProcedure
\end{algorithmic}
\end{algorithm}

An $L$-level hierarchical decomposition of a factor graph $\graphG$ is the rooted tree generated by Algorithm~\ref{alg:hierarchical_decomposition}, initialized with
\[
    \Call{Decompose}{\graphG, 0, N_{\mathrm{c},\max}, N_{\mathrm{w},\max}}.
\]
For notational simplicity, we adopt the convention that all leaves are regarded as lying at level $L$; if a branch terminates earlier, it may be extended by trivial copies of the same terminal subgraph. The resulting structure is characterized as follows:
\begin{enumerate}
    \item \textbf{Level 0 (root):} The root of the tree is the original factor graph,
    \[
        \graphG^{(0)} \defeq \graphG.
    \]
    
    \item \textbf{Meta-subgraphs:}
    Fix $\ell\in\{1,\ldots,L\}$ and consider a non-leaf node $\graphG^{(\ell-1)}_{\vk'}$ at level $\ell-1$, indexed by
    \[
        \vk'=(k_1,\ldots,k_{\ell-1})\in [K_{1}]\times\cdots\times [K_{\ell-1}],
    \]
    where $K_{1},\ldots,K_{\ell-1}\in\sZpp$. Let
    $N_{\mathrm{c},\mathrm{qubits}}(\vk')$ and
    $N_{\mathrm{w},\mathrm{qubits}}(\vk')$
    denote the coordinator and worker qubit requirements associated with $\graphG^{(\ell-1)}_{\vk'}$.
    We call $\graphG^{(\ell-1)}_{\vk'}$ a \emph{meta-subgraph} if it violates the hardware constraints, that is,
    \[
        N_{\mathrm{c},\mathrm{qubits}}(\vk')>N_{\mathrm{c},\max}
        \quad\text{or}\quad
        N_{\mathrm{w},\mathrm{qubits}}(\vk')>N_{\mathrm{w},\max}.
    \]
    Such a node is partitioned by an internal boundary-variable set $\set{V}^{(\ell)}_{\mathrm{B},\vk'}$ into child subgraphs
    \[
        \bigl\{\graphG^{(\ell)}_{\vk',k_\ell}\bigr\}_{k_\ell\in\mathcal{K}_\ell(\vk')},
    \]
    where $\mathcal{K}_\ell(\vk')$ is a finite index set. We write
    $K_\ell(\vk')\defeq |\mathcal{K}_\ell(\vk')|$ for the number of children of node $\vk'$.
    In general, different meta-subgraphs may have different branching factors.

    
    \item \textbf{Leaf-subgraphs:}
    The recursion terminates at nodes whose resource requirements satisfy the hardware constraints. Under the uniform-depth convention above, such terminal nodes are denoted by $\graphG^{(L)}_{\vk}$ and satisfy
    \[
        N_{\mathrm{c},\mathrm{qubits}}(\vk)\leq N_{\mathrm{c},\max}
        \quad\text{and}\quad
        N_{\mathrm{w},\mathrm{qubits}}(\vk)\leq N_{\mathrm{w},\max}.
    \]
    These leaves are the atomic subproblems executed on individual QPUs.
\end{enumerate}

\textbf{Bottom-Up Hierarchical Execution}.
Once the decomposition tree has been constructed, computation proceeds from the leaves toward the root. Each node of the tree is assigned a corresponding sub-network, with one QPU designated as the local coordinator for that node and the QPUs assigned to its children acting as the workers for the next higher level. At an internal node, the coordinator solves the associated subproblem by running a level-specific instance of $A_{\mathrm{dist}}$ over its children. This requires a coherent implementation not only of the amplitude-amplification operator in~\eqref{eqn: def of UAANz}, but also of the controlled phase-testing primitive in~\eqref{eqn: controll phase estimation_main_alg}. As shown in Appendix~\ref{apx:implementationof_Upe}, the corresponding distributed protocol may require parent-child teleportation both for the child-dependent components of~\eqref{eqn: def of UAANz} and for the controlled-unitary mechanism inside~\eqref{eqn: controll phase estimation_main_alg}. The distinction between coherent and hybrid execution is therefore a distinction about how these inter-level quantum primitives are realized.

\begin{definition}[\textbf{Hierarchical Algorithm}]
\label{def:hierarchical_algorithm}
Let $A_{\mathrm{hier}}$ operate on the hierarchical decomposition produced by Algorithm~\ref{alg:hierarchical_decomposition}. Its execution is defined recursively as follows.
\begin{enumerate}
    \item \textbf{Base case (level $L$).}
    For a leaf subgraph $\graphG^{(L)}_{\vk}$, the assigned QPU runs the generic maximization routine $\matr{U}_{\max}$ from Appendix~\ref{apx:generic_maximization_subroutine}, conditioned on the boundary variables supplied by its parent. The output is an $\Np$-qubit register encoding an $\Np$-bit approximation of the corresponding local optimum.

    \item \textbf{Recursive step (level $\ell<L$).}
    Consider an internal node $\graphG^{(\ell)}_{\vk}$ with children
    $\{\graphG^{(\ell+1)}_{\vk,k}\}_{k\in\mathcal{K}_{\ell+1}(\vk)}$.
    Its local coordinator executes the single-level distributed routine $A_{\mathrm{dist}}$ of Algorithm~\ref{alg:dist-opt}, treating the child nodes as its workers. The key design choice is how the level-$\ell$ inter-level primitives are implemented from the child routines. In particular, one must realize coherently or classically the child-dependent operations that appear in the level-$\ell$ counterparts of~\eqref{eqn: def of UAANz} and~\eqref{eqn: controll phase estimation_main_alg}.
\end{enumerate}

\noindent\textit{Execution modes.} For any sub-coordinator at level $\ell<L$:
\begin{enumerate}
    \setcounter{enumi}{2}
    \item \textbf{Mode 1: Coherent cascade.}
    The level-$\ell$ coordinator invokes the child routines without intermediate measurement and maintains a coherent parent-child interface throughout the realization of the level-$\ell$ counterparts of~\eqref{eqn: def of UAANz} and~\eqref{eqn: controll phase estimation_main_alg}. Consequently, the child sub-network is used coherently not only to implement the level-$\ell$ state-preparation, marking, and reflection operations, but also to support the controlled-unitary steps that occur inside the phase-testing subroutine. This preserves coherence across levels and retains the Grover-like quadratic speedup, but it also requires repeated teleportation of data and control qubits across the parent-child interface, deep circuits, and sustained high-fidelity inter-node entanglement.

    \item \textbf{Mode 2: Hybrid with intermediate measurement.}
    Each child routine is measured after completion, and only the resulting classical bit strings are returned to the level-$\ell$ coordinator. The interface between levels therefore no longer supports coherent implementations of the child-dependent parts of~\eqref{eqn: def of UAANz} or of the distributed phase-testing primitive in~\eqref{eqn: controll phase estimation_main_alg}. Instead, the coordinator classically assembles the effective objective for its node and evaluates it by explicit iteration over the relevant boundary assignments. This removes the need for inter-level qubit teleportation at the measured interface and reduces circuit depth and inter-level entanglement, but it replaces quantum search at that level by classical enumeration.
\end{enumerate}
\end{definition}

\subsection{Performance Analysis}
We now analyze the performance of the hierarchical algorithm $A_{\mathrm{hier}}$, making explicit the trade-off between query complexity and circuit depth in the coherent and hybrid modes.

\textbf{Setup.}
Fix a precision parameter $\Np$ and a per-call failure parameter $\delta\in(0,1)$.
Let $\mathcal{T}$ be an $L$-level hierarchical decomposition tree produced by Algorithm~\ref{alg:hierarchical_decomposition}.
Each internal node $\graphG^{(\ell)}_{\vk}$ is decomposed by an internal boundary-variable set $\set{V}^{(\ell+1)}_{\mathrm{B},\vk}$ into children
$\{\graphG^{(\ell+1)}_{\vk,k}\}_{k\in\mathcal{K}_{\ell+1}(\vk)}$, where $K_{\ell+1}(\vk)\defeq |\mathcal{K}_{\ell+1}(\vk)|$.
When node $\graphG^{(\ell)}_{\vk}$ is executed in the coherent mode, its sub-coordinator runs the single-level routine $A_{\mathrm{dist}}$ (Algorithm~\ref{alg:dist-opt}) and requires a lower bound on the overlap of its state-preparation oracle $\matr{U}^{(\ell)}_{\mathrm{ini}}$.
Under the success guarantee $(1-\delta)^{2\Np}$ for each child routine, Theorem~\ref{thm: overall_algorithm_properties} allows us to take
\begin{align}
    p_{\min,\vk}^{(\ell)} \defeq \frac{(1-\delta)^{2\Np\,K_{\ell+1}(\vk)}}{\lvert \setx_{\set{V}^{(\ell+1)}_{\mathrm{B},\vk}} \rvert}.
\label{eqn:pmin_hier_level_ell}
\end{align}
Each leaf $\graphG^{(L)}_{\vk}$ runs $\matr{U}_{\max}$ (Appendix~\ref{apx:generic_maximization_subroutine}) with a local overlap parameter $p_{\min,\mathrm{leaf},\vk}$ determined by its local oracles.
For the baseline Hadamard state preparation (uniform superposition over internal configurations), we have
$p_{\min,\mathrm{leaf},\vk}=\lvert \setx_{\set{V}^{(L)}_{\vk}} \rvert^{-1}$.

\begin{theorem}[\textbf{Overall Performance of the Hierarchical Algorithm}]
\label{thm:hierarchical_performance}
Let $A_{\mathrm{hier}}$ be the $L$-level hierarchical algorithm (Definition~\ref{def:hierarchical_algorithm}) operating on a recursive decomposition of $\graphG$ with the Setup described above. The following bounds hold.
\begin{enumerate}
    \item \label{itm:hier_qubits_per_proc} \textbf{Qubits per processor.}
    \begin{align*}
        &N_{\mathrm{qubits}}(\text{sub-coord for } \graphG^{(\ell)}_{\vk})
        \in O\Bigl( \lvert \set{V}^{(\ell+1)}_{\mathrm{B},\vk} \rvert
        + \mathrm{Poly}\bigl(K_{\ell+1}(\vk),\Np,\log(1/\delta)\bigr)\Bigr), \\[5pt]
        &N_{\mathrm{qubits}}(\text{leaf }\vk)
        \in O\Bigl(\lvert \set{V}^{(L)}_{\vk} \rvert
        + \mathrm{Poly}\bigl(\Np,\log(1/\delta)\bigr)\Bigr).
    \end{align*}
    Here, for a child $\graphG^{(\ell+1)}_{\vk,k}$ with $k\in\mathcal{K}_{\ell+1}(\vk)$, we write
    $\set{V}^{(\ell+1)}_{\mathrm{B},\vk,k}$ for the subset of boundary variables that appears in that child subgraph.

    \item \label{itm:hier_query_complexity} \textbf{Query complexity}
    \begin{align*}
        C_{\mathrm{hier}}
        \in
        O\left(
        \sum_{\graphG^{(L)}_{\vk}\in\mathcal{T}}
        N_{\mathrm{inv}}\!\left(\graphG^{(L)}_{\vk}\right)
        \cdot
        C_{\mathrm{leaf}}\!\left(\graphG^{(L)}_{\vk}\right)
        \right),
    \end{align*}
    where $\set{L}_{\mathrm{M}}\subseteq\{0,\ldots,L-1\}$ is the set of levels executed in Mode~2 (intermediate measurement),
    $C_{\mathrm{leaf}}(\graphG^{(L)}_{\vk})$ is the number of queries to the leaf's local oracle in one execution of $\matr{U}_{\max}$:
    \begin{align}
      C_{\mathrm{leaf}}\!\left(\graphG^{(L)}_{\vk}\right)
      \defeq
      \Np\cdot\Bigl(2^{t(p_{\min,\mathrm{leaf},\vk},\delta)+2}-2\Bigr)
      \qquad\text{(Proposition~\ref{prop:generic_maximization_performance})},
      \label{eqn:hier_leaf_query_cost_def}
    \end{align}
    and $N_{\mathrm{inv}}(\cdot)$ is the (policy-dependent) invocation count induced by the hierarchy, defined recursively by
    $N_{\mathrm{inv}}(\graphG^{(0)})\defeq 1$ and
    \begin{align}
      N_{\mathrm{inv}}\!\left(\graphG^{(\ell+1)}_{\vk,k}\right)
      \defeq
      N_{\mathrm{inv}}\!\left(\graphG^{(\ell)}_{\vk}\right)\cdot C_{\mathrm{eff}}\!\left(\graphG^{(\ell)}_{\vk}\right),
      \qquad
      k\in\mathcal{K}_{\ell+1}(\vk),
      \label{eqn:hier_invocation_count_def}
    \end{align}
    with the effective per-node child-invocation factor
    \begin{align}
      C_{\mathrm{eff}}\!\left(\graphG^{(\ell)}_{\vk}\right)\in
      \begin{cases}
      \Np\cdot\Bigl(2^{t(p_{\min,\vk}^{(\ell)},\delta)+2}-2\Bigr)
      & \ell\notin \set{L}_{\mathrm M} \text{ (coherent)}\\[4pt]
      \bigl|\setx_{\set{V}^{(\ell+1)}_{\mathrm{B},\vk}}\bigr|
      & \ell\in \set{L}_{\mathrm M} \text{ (measured, classical enumeration)}
      \end{cases},
      \label{eqn:hier_Ceff_def}
    \end{align}
    where, in the coherent case, Item~\ref{prop: overall_query_complexity: final} of Theorem~\ref{thm: overall_algorithm_properties} and~\eqref{eqn:pmin_hier_level_ell} imply
    \begin{align}
      C_{\mathrm{eff}}\!\left(\graphG^{(\ell)}_{\vk}\right)
      \in
      O\Biggl(\Np\cdot (1-\delta)^{-\Np\,K_{\ell+1}(\vk)}\cdot \sqrt{|\setx_{\set{V}^{(\ell+1)}_{\mathrm{B},\vk}}|}\Biggr).
      \label{eqn:hier_Ceff_coherent_bigO}
    \end{align}

    \item \label{itm:hier_epr} \textbf{EPR consumption (coherent mode).}
    \begin{align*}
      N_{\mathrm{EPR,hier}}
      \in
      O\left(
      \sum_{\substack{\graphG^{(\ell)}_{\vk}\in\mathcal{T}:\\ 0\le \ell \le L-1}}
      N_{\mathrm{inv}}\!\left(\graphG^{(\ell)}_{\vk}\right)
      \cdot
      N_{\mathrm{EPR}}\!\left(\graphG^{(\ell)}_{\vk}\right)
      \right).
    \end{align*}
    Here $N_{\mathrm{EPR}}(\graphG^{(\ell)}_{\vk})$ denotes the number of EPR pairs consumed by \emph{one execution} of the level-$\ell$ distributed primitive (the single-level routine $A_{\mathrm{dist}}$ run on the sub-network induced by $\graphG^{(\ell)}_{\vk}$ and its children), counting only coherent inter-node communication at that level, and we set $N_{\mathrm{EPR}}(\graphG^{(\ell)}_{\vk})\defeq 0$ for $\ell\in\set{L}_{\mathrm M}$.
    In the single-hop communication model, expression~\eqref{eqn:total_epr_consumption} yields
    \begin{align*}
      N_{\mathrm{EPR}}\!\left(\graphG^{(\ell)}_{\vk}\right)
      &=
      \Np \cdot \Bigl( 2^{\,t(p_{\min,\vk}^{(\ell)},\delta)+2} - 2 \Bigr)
      \cdot 
        \sum_{k\in\mathcal{K}_{\ell+1}(\vk)} \bigl( 2\,|\set{V}^{(\ell+1)}_{\mathrm{B},\vk,k}| + \Np \bigr)
      \\&\quad
      + \Np \cdot \Bigl( 2^{\,t(p_{\min,\vk}^{(\ell)},\delta)+2} - 4 \Bigr)\cdot K_{\ell+1}(\vk)
      +
      \Np \cdot t(p_{\min,\vk}^{(\ell)}, \delta) \cdot \bigl(2\,K_{\ell+1}(\vk)\bigr),
    \end{align*}
    for $\ell\notin\set{L}_{\mathrm M}$.

	    \item \label{itm:hier_success_precision} \textbf{Success and precision.}
	    \begin{align*}
	      \Pr(\text{success}) \ge (1-\delta)^{2\Np}
	        \qquad\text{in the coherent cascade }(\set{L}_{\mathrm M}=\emptyset).
		    \end{align*}
		    In hybrid execution, we define the set of coherently executed subroutines whose outputs are read out (measured) and used as classical data as
		    \begin{align*}
		      \mathcal{T}_{\mathrm{out}}
		      \defeq
		      \left\{
		        \graphG^{(\ell)}_{\vk}\in\mathcal{T} \ \middle| \
		        \ell\notin \set{L}_{\mathrm M}
		        \text{ and }
		        (\ell=0 \text{ or } \ell-1\in \set{L}_{\mathrm M})
		      \right\},
		    \end{align*}
		    and let
		    \begin{align*}
		      N_{\mathrm{exec}}
		      \defeq
		      \sum_{\graphG^{(\ell)}_{\vk}\in\mathcal{T}_{\mathrm{out}}}
		      N_{\mathrm{inv}}\!\left(\graphG^{(\ell)}_{\vk}\right).
		    \end{align*}
		    Then
		    \begin{align*}
		        \Pr(\text{success}) &\ge (1-\delta)^{2\Np\,N_{\mathrm{exec}}}.
		    \end{align*}
	    If we want the overall success probability satisfying \(\Pr(\text{success}) \ge 1-\delta'\) for any $0<\delta'<1$, 
	    Bernoulli's inequality implies that 
	    it suffices to choose
	    \begin{align*}
        \delta
        \leq
        \begin{cases}
        \frac{\delta'}{2\Np}, & \set{L}_{\mathrm M}=\emptyset \quad\text{(coherent cascade)}\\[2pt]
        \frac{\delta'}{2\Np\cdot N_{\mathrm{exec}}}, & \set{L}_{\mathrm M}\neq\emptyset \quad\text{(hybrid execution)}
        \end{cases}\quad.
      \end{align*}
      The corresponding $\Np$-bit approximation $\mathbf{z}_{\max,\rmp,\rmc}$ of the maximum $\gmax$ satisfies
      \begin{align*}
		        0 &\le \gmax - \zdec\bigl(\mathbf{z}_{\max,\rmp,\rmc}\bigr) \le |\mathcal{T}|\cdot 2^{-\Np}.
      \end{align*}
\end{enumerate}
\end{theorem}

\begin{proof}
See Appendix~\ref{apx:hierarchical_performance}.
\end{proof}

Theorem~\ref{thm:hierarchical_performance} separates the hierarchical design problem into hardware feasibility, end-to-end query complexity, inter-node entanglement cost, and reliability.
\begin{itemize}
    \item \textbf{Qubits (Item~\ref{itm:hier_qubits_per_proc}):} This item is the hardware-admissibility condition for the decomposition. A valid hierarchy must keep every internal sub-coordinator within the coordinator budget and every leaf within the worker budget. In particular, a useful partition is one that reduces the sizes of the leaf subproblems without creating boundary sets so large that the internal coordinators become infeasible.
    \item \textbf{Query cost (Item~\ref{itm:hier_query_complexity}):} The bound is organized by leaves through the one-shot cost in~\eqref{eqn:hier_leaf_query_cost_def} and the recursive invocation count in~\eqref{eqn:hier_invocation_count_def}. This makes the source of the total cost transparent: each leaf contributes its local oracle cost, multiplied by the number of times it is triggered by its ancestors. The effective per-node amplification factor is specified in~\eqref{eqn:hier_Ceff_def}, with the coherent scaling summarized in~\eqref{eqn:hier_Ceff_coherent_bigO}.

    \item \textbf{Meaning of hybridization (Item~\ref{itm:hier_query_complexity}):} Measuring at level $\ell$ changes the recursion locally. In~\eqref{eqn:hier_Ceff_def}, the coherent child-invocation factor is replaced by explicit enumeration over the corresponding boundary assignments. The theorem therefore identifies the exact trade-off introduced by hybrid execution: one reduces coherent depth and removes EPR consumption at that cut, but increases the number of lower-level subroutine invocations.

    \item \textbf{EPR cost (Item~\ref{itm:hier_epr}):} The entanglement bound mirrors the query recursion, but counts only coherent communication. The per-node cost is inherited from the single-level distributed implementation in~\eqref{eqn:total_epr_consumption}, and the factor $N_{\mathrm{inv}}$ shows how often that coherent primitive is reused. As a result, a measured level contributes no EPR pairs directly, yet it may still increase total entanglement consumption below that level by replicating coherent calls deeper in the tree.
    
    \item \textbf{Success / precision (Item~\ref{itm:hier_success_precision}):} The theorem separates reliability from approximation quality. The success probability depends on how many coherent subroutines are executed and then consumed classically, whereas the approximation error accumulates additively across the decomposition tree. This makes explicit how the hierarchy trades depth and communication against a controlled, decomposition-dependent accuracy certificate.
\end{itemize}

\subsection{Gate-level Validation and Large-scale Resource Sweep}

\begin{figure}[t]
\centering
\includegraphics[width=0.6\linewidth]{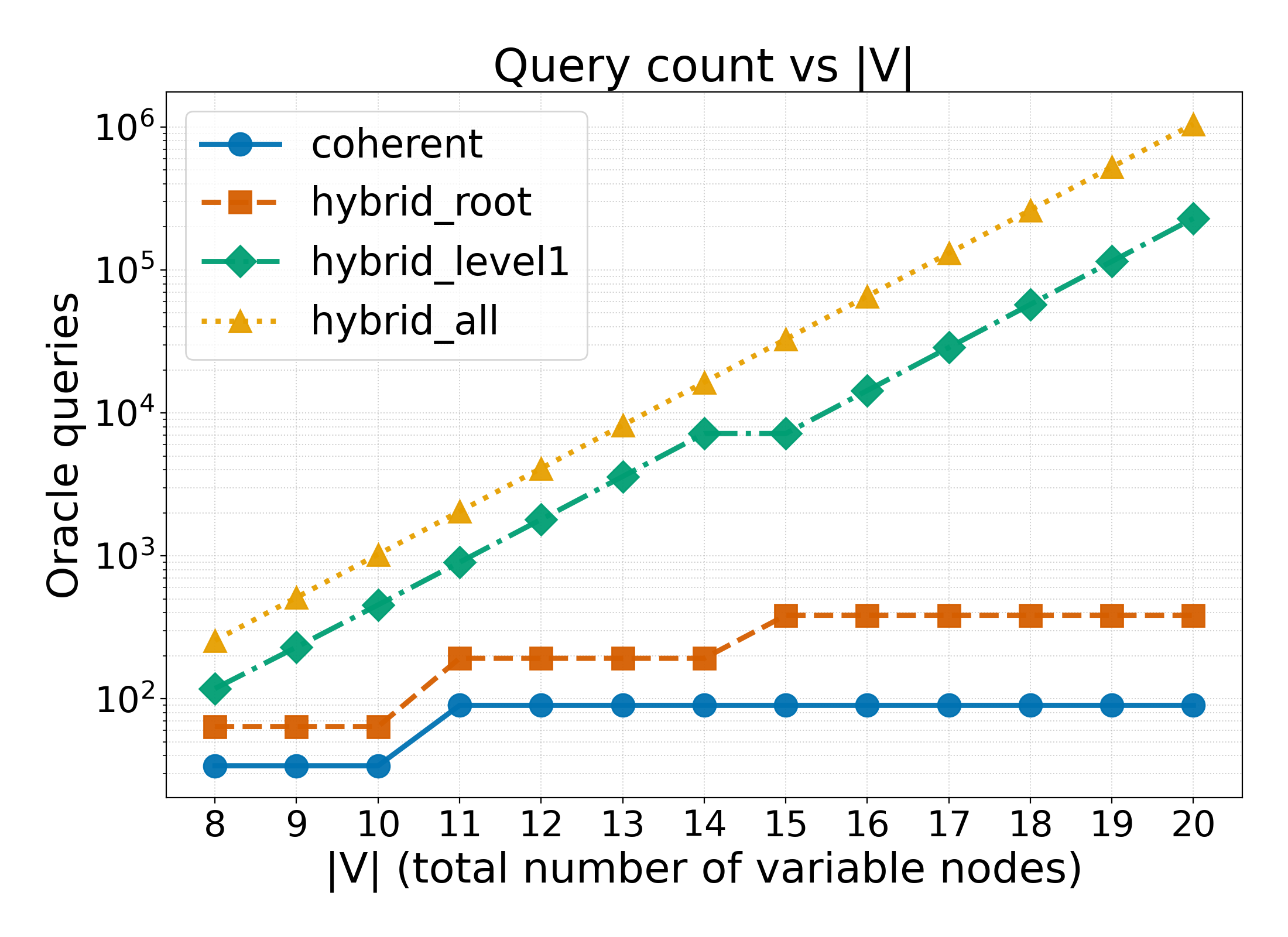}
\caption{Total oracle-query count versus $|\setV|$ (total number of variable nodes) on a two-level hierarchical benchmark, under four execution policies. The $y$-axis is logarithmic.}
\label{fig:hier_queries_vs_n_chain_main}
\centering
\includegraphics[width=0.6\linewidth]{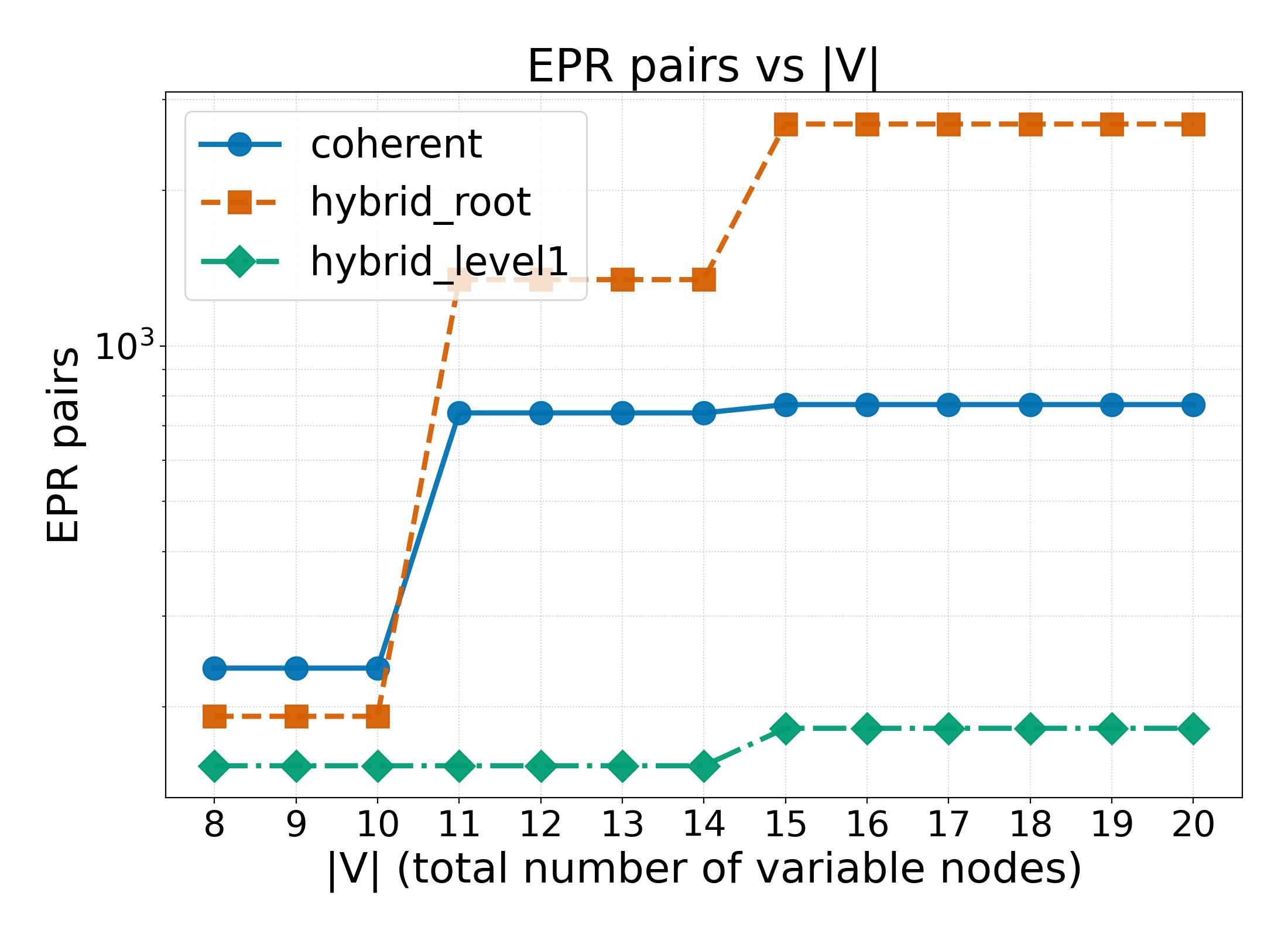}
\caption{Total EPR-pair count versus $|\setV|$ for the same benchmark. We plot only policies with nonzero coherent inter-node communication; in this accounting model, \texttt{hybrid\_all} is identically zero. The $y$-axis is logarithmic.}
\label{fig:hier_epr_vs_n_chain_main}
\end{figure}

Figures~\ref{fig:hier_queries_vs_n_chain_main}--\ref{fig:hier_epr_vs_n_chain_main} report exact state-vector simulations of the compiled gate-level circuits for a two-level hierarchy. We compare four execution policies: \texttt{coherent}, \texttt{hybrid\_root}, \texttt{hybrid\_level1}, and \texttt{hybrid\_all}. To obtain a dense but still statevector-feasible sweep, we evaluate thirteen benchmark instances along the monotone parameter path
\[
 (b_0,b_1,S_0,S_1,r)\in
\{(5,0,1,1,r)\}_{r=3}^{5}
\;\cup\;
\{(5,1,1,1,r)\}_{r=5}^{8}
\;\cup\;
\{(6,1,1,1,r)\}_{r=8}^{13},
\]
with fixed $(\Np,\delta,t_{\max})=(1,0.2,2)$.
Here $b_0$ and $b_1$ are boundary widths, $S_0$ and $S_1$ are branching factors at levels 0 and 1, and $r$ is the leaf-local bit width.
This construction yields the contiguous range $|\setV|=8,\ldots,20$ while keeping the state dimension tractable. All reported oracle-query and EPR totals are extracted directly from simulator counters, so the figures reflect the implemented circuits rather than a separate asymptotic model.

\textbf{Query trend.}
The ordering is
\[
\texttt{coherent} < \texttt{hybrid\_root} < \texttt{hybrid\_level1} < \texttt{hybrid\_all}.
\]
This ordering is consistent with Item~\ref{itm:hier_query_complexity}, and in particular with the invocation recursion in~\eqref{eqn:hier_invocation_count_def} and the policy-dependent effective factor in~\eqref{eqn:hier_Ceff_def}. The coherent policy is smallest because every internal level retains coherent search. By contrast, in this two-level hierarchy, \texttt{hybrid\_root}, \texttt{hybrid\_level1}, and \texttt{hybrid\_all} replace the coherent factor at the root level only, at level~1 only, and at both internal levels, respectively, by explicit enumeration of the corresponding boundary assignments. At the largest simulated point ($|\setV|=20$), the query multipliers relative to \texttt{coherent} are $4.27$, $2.55\times 10^{3}$, and $1.17\times 10^{4}$ for \texttt{hybrid\_root}, \texttt{hybrid\_level1}, and \texttt{hybrid\_all}, respectively.

\textbf{Entanglement trend.}
The EPR accounting follows Item~\ref{itm:hier_epr}, which counts only coherent inter-node communication. Accordingly, \texttt{hybrid\_all} is identically zero in this accounting model because both internal interfaces are measured. At $|\setV|=20$, the EPR multipliers relative to \texttt{coherent} are $3.49$ and $0.236$ for \texttt{hybrid\_root} and \texttt{hybrid\_level1}, respectively. This contrast is also explained by Item~\ref{itm:hier_epr}: measuring at the root removes the root-level coherent primitive but can replicate many coherent calls below it, whereas measuring at level~1 removes the deeper coherent interface itself and therefore reduces the total EPR count.

The visible jumps in the curves are structural rather than numerical artifacts. They occur when the sweep increases $b_1$ or $b_0$, thereby changing the boundary domains that enter~\eqref{eqn:hier_Ceff_def} and the coherent communication cost inherited from the single-level primitive. Between these jumps, increasing $r$ enlarges the leaf-local domain. As a result, policies with measured levels grow rapidly on the logarithmic scale because of explicit branching, whereas coherent policies change only when the relevant integer circuit parameters or boundary widths change. These simulation results therefore provide a direct circuit-level validation of the mechanisms identified in Theorem~\ref{thm:hierarchical_performance}. 

Appendix~\ref{sec:hier_chain_multilevel} provides the full simulation setup, additional observations, and detailed interpretation.

\begin{table}[t]
  \centering
  \caption{Comparison of the single-level and hierarchical variants of our framework.
  ``Coherent'' keeps global coherence across all levels; ``Hybrid'' inserts measurements at selected levels to cap
  circuit depth. }
  \label{table:comparison_of_dqc_paradigms}

  \scriptsize
  \begin{tabular}{@{}p{3.3cm}p{3.2cm}p{3.7cm}p{3.5cm}@{}}
    \toprule
    \textbf{Aspect} &
    \textbf{Single-level ($A_{\mathrm{dist}}$)} &
    \textbf{Hierarchical (coherent)} &
    \textbf{Hierarchical (hybrid)} \\
    \midrule

    \textbf{Organization} &
    One coordinator with one worker layer &
    \multicolumn{2}{p{7cm}}{Recursive coordinator tree matching the decomposition hierarchy} \\
    \midrule
    \addlinespace[2pt]

    \textbf{Scales when} &
    Each subproblem fits on one QPU &
    \multicolumn{2}{p{7cm}}{Oversized subproblems can be recursively decomposed until all leaves fit on a QPU} \\
    \midrule
    \addlinespace[2pt]

    \textbf{Global coherence} &
    Preserved end-to-end &
    Preserved across all levels &
    Broken at selected levels (by measurement) \\
    \midrule
    \addlinespace[2pt]

    \textbf{Grover-like scaling} &
    Yes &
    Yes (up to level-dependent multiplicative factors) &
    Only within each coherent block \\
    \midrule
    \addlinespace[2pt]

    \textbf{Dominant overhead} &
     Boundary search over $\setVb$ &
     Boundary-search factors compound across levels &
    Classical branching at measured levels \\
    \midrule
    \addlinespace[2pt]

    \textbf{Circuit depth} &
    High &
    Highest  (coherent cascade) &
     Lower (depth capped by measurements) \\
     \midrule
    \addlinespace[2pt]

    \textbf{Best-fit regime} &
    Mid-term fault-tolerant networks &
    Far-term fault-tolerant networks &
     NISQ / early fault-tolerant networks \\
    \bottomrule
  \end{tabular}








\end{table}

\subsection{Discussion: From a Bounded Method to a Universal Framework}
Definition~\ref{def:hierarchical_algorithm} lifts the single-level routine $A_{\mathrm{dist}}$ (Section~\ref{sec: distributed_quantum_algorithm}) to a recursion over a decomposition tree, yielding an end-to-end procedure for instances that exceed the capacity of any single processor.
The key shift is that $A_{\mathrm{dist}}$ requires every subproblem to fit on one QPU, whereas the hierarchical framework requires only that any oversized subproblem admit further decomposition until all leaves satisfy the hardware limits.
Table~\ref{table:comparison_of_dqc_paradigms} summarizes the qualitative differences, and Theorem~\ref{thm:hierarchical_performance} makes the resulting resource trade-offs explicit.

\textbf{Scalability via recursion.}
At each internal node, a sub-coordinator invokes $A_{\mathrm{dist}}$ on its children, treating the children's routines as effective oracles (Definition~\ref{def:hierarchical_algorithm}).
This can be viewed as a structure-aware compilation layer: given a factor graph and per-QPU budgets, it produces a distributed execution plan whose cost is governed by the boundary interfaces introduced by the cuts.

\textbf{A policy knob for hardware regimes.}
The designer may choose the set of measured levels $\set{L}_{\mathrm{M}}$ (Theorem~\ref{thm:hierarchical_performance}) to trade query / entanglement overhead for circuit depth.
The coherent cascade ($\set{L}_{\mathrm{M}}=\emptyset$) preserves global coherence across all levels and retains Grover-like scaling.
Hybrid execution ($\set{L}_{\mathrm{M}}\neq\emptyset$) inserts measurements to cap coherent depth and improve noise tolerance, but replaces the coherent boundary search at measured levels with classical branching.

\textbf{Design guidance.}
Theorem~\ref{thm:hierarchical_performance} shows that the dominant overhead is controlled by the boundary search spaces $\bigl|\set{X}_{\setVb^{(\ell+1)}}\bigr|$ across levels.
Good decompositions therefore keep boundary sets small throughout the hierarchy; otherwise recursion amplifies costs multiplicatively.
In hybrid mode, measurements should be introduced only where depth / noise constraints require them, because each measured level incurs the stronger $O\bigl(|\set{X}_{\setVb^{(\ell+1)}}|\bigr)$ enumeration penalty instead of the coherent $O\Bigl(\sqrt{|\set{X}_{\setVb^{(\ell+1)}}|}\Bigr)$ factor.


\section{Conclusion and Future Directions}
\label{sec:conclusion}

This paper has tackled a central obstacle in quantum optimization on networks: how to distribute a large search problem
across QPU nodes with limited qubits \emph{without} giving up the quadratic advantage that makes quantum search
attractive. We have proposed a \emph{structure-aware} framework that models the objective function via a factor graph and
decomposes it along a separator $\setVb$, so that communication and storage follow the problem topology.

At the single-decomposition level, we have shown that the fully coherent execution mode preserves Grover-like query
scaling, namely $O(\sqrt{N})$ oracle complexity up to processor- and separator-dependent factors, in contrast to
partitioning schemes that dilute the quadratic advantage~\cite{Avron2021}. We have also made the accompanying
communication cost explicit by deriving the required EPR consumption for the distributed implementation. At the systems
level, our analysis shows that the quality of the separator controls the central trade-off of the method: smaller and
sparser boundaries reduce inter-node communication and coordinator load while keeping the decomposition useful for
resource-constrained processors.

We have then extended the framework to a hierarchical setting for instances that exceed the capacity of any single QPU.
The hierarchical analysis clarifies how the overall cost is built from three components: per-node hardware feasibility,
recursive invocation multiplicity, and coherent inter-node communication. In particular, the hybrid mode does not aim to
preserve the full coherent speedup at measured levels; rather, it provides a controlled alternative that reduces coherent
circuit depth and entanglement requirements while preserving correctness and making the query-depth trade-off explicit.
The gate-level simulations support this interpretation: they corroborate the predicted scaling trends, identify the
dominant resource drivers, and confirm that measurement placement governs a concrete trade-off between oracle complexity
and entanglement consumption. Taken together, these results establish that structure-aware factor-graph decomposition is
a viable route to large-instance optimization on quantum networks: it reduces per-QPU qubit requirements, makes the
communication overhead analyzable, and retains Grover-like scaling in the fully coherent regime.

Future work will focus on three directions. We plan to develop \emph{hardware-aware} compilation that chooses $\setVb$
jointly with QPU placement, network topology, and link metrics (\eg, fidelity, rate, and latency) to reduce peak qubits,
EPR consumption, and end-to-end runtime. We plan to study \emph{noise-aware adaptivity} by modeling local-gate noise and
interconnect noise separately and using real-time calibration feedback to switch between coherent and hybrid execution
modes (and to choose the measurement levels). We also plan to release reference instances, oracle implementations, and
resource-accounting code to support reproducible evaluation on multi-node testbeds.

\ifsplitmainonly
\bibliographystyle{quantum}
\bibliography{biblio_used}
\fi

\fi
\ifbuildappendixpart
\appendix

\clearpage

\section{Basic Notations and Definitions}
\label{apx:basic_notation}
The sets $\mathbb{R}$, $\mathbb{R}_{\geq 0}$, and $\mathbb{R}_{>0}$ are
defined to be the field of real numbers, the set of nonnegative real numbers,
and the set of positive real numbers.
Unless we state otherwise, all variable
alphabets are assumed to be finite.
For any integer $ N \in \mathbb{Z}_{>0} $, the function $ [N] $ is defined to be the set $ [ N ] \defeq \{ 1,\ldots,N \} $ with cardinality $ N $.
For any length-$ N $ vector $ \mathbf{v} $, we define 
$ \mathbf{v} \defeq \bigl( v(1),\ldots,v(N) \bigr) $.
For any $ n_{1},n_{2} \in \sZ $ such that 
$ 1 \leq n_{1} \leq n_{2} \leq N $, we define
\begin{align*}
  \mathbf{v}(n_{1}:n_{2}) \defeq \bigl( v(n_{1}),\ldots,v(n_{2}) \bigr).
\end{align*}

Consider an arbitrary integer $ N \in \sZpp $.
For any $ (x_{1},\ldots,x_{N}) \in \{0,1\}^{N} $,
we define
\begin{align*}
  \ket{x_{1},\ldots,x_{N}}
  &\defeq \bigotimes_{n' = 1}^{N} \ket{x_{n'}}, 
\end{align*}
where
\begin{align*}
  \ket{0}
  &\defeq (1,\, 0)^{\Herm}, \qquad
  \ket{1}
  \defeq (0,\, 1)^{\Herm}.
\end{align*}
If there is no ambiguity, all the state vectors in this paper will be represented in the computational basis.
We use $ \imagunit $ for denoting the imaginary unit.
The Pauli matrices are defined as follows
\begin{align*}
  \matr{I} &\defeq \begin{pmatrix}
    1 & 0 \\ 0 & 1
  \end{pmatrix}, \qquad
  \matr{X} \defeq \begin{pmatrix}
    0 & 1 \\ 1 & 0
  \end{pmatrix}, \qquad
  \matr{Y} \defeq \begin{pmatrix}
    0 & - \imagunit \\ \imagunit & 0
  \end{pmatrix},\qquad
  \matr{Z} \defeq \begin{pmatrix}
    1 & 0 \\ 0 & -1
  \end{pmatrix}.
\end{align*}
The Hadamard gate is defined to be
\begin{align*}
  \matr{H} \defeq
  \frac{\sqrt{2}}{2}
  \begin{pmatrix}
    1 & 1 \\ 1 & -1
  \end{pmatrix}.
\end{align*}
For any complex-valued number $ c \in \sC $, we use $ \overline{c} $ denoting the complex conjugate of $ c $.
The controlled NOT (CNOT) gate is the unitary gate acting on two qubits:
\begin{align*}
    \mathrm{CNOT}
    \defeq 
    \begin{pmatrix}
        1 & 0 & 0 & 0 \\    
        0 & 1 & 0 & 0 \\
        0 & 0 & 0 & 1 \\
        0 & 0 & 1 & 0 
    \end{pmatrix},
\end{align*}
where the first qubit is the control 
qubit and the second qubit is the target qubit. 
If not specified, the measurement of a qubit is the measurement on the computational basis, \ie, the measurement
$ \{ \ket{0} \bra{0}, \ket{1} \bra{1} \} $.

\begin{definition}[\textbf{Marking Unitary Operator and Projection Matrix}]
 \label{def:marking_operators}
 Consider an $N$-qubit Hilbert space $\mathcal{H} \cong (\sC^2)^{\otimes N}$, where $N \in \sZpp{}$. A computational basis state is denoted by $\ket{\mathbf{z}} = \ket{z(1)\cdots z(N)}$ for a binary string $\mathbf{z} \in \{0,1\}^{N}$. For any subset of basis states $ \set{S} \subseteq \{0,1\}^{N} $, we define:
 \begin{enumerate}
  \item The \emph{marking oracle} for $ \set{S} $ is the unitary matrix $ \matr{U}_{\set{S}} \in \sC^{ 2^{N} \times 2^{N}} $, defined by its action on any computational basis state $\ket{\mathbf{z}}$:
  \begin{align*}
   \matr{U}_{ \set{S} } \ket{\mathbf{z}}
   \defeq \begin{cases}
   - \ket{\mathbf{z}} & \text{if } \mathbf{z} \in \set{S} \\
   \hphantom{-} \ket{\mathbf{z}} & \text{if } \mathbf{z} \notin \set{S}
   \end{cases}\quad .
  \end{align*}
  \item The \emph{projection matrix} onto the subspace spanned by states in $ \set{S} $ is $ \matr{P}_{ \set{S} } \in \sC^{ 2^{N} \times 2^{N}} $, which defined to be
  \begin{align*}
   \matr{P}_{ \set{S} }
   \defeq \sum_{\mathbf{z} \in \set{S}}
   \ket{\mathbf{z}}\bra{\mathbf{z}}.
  \end{align*}
 \end{enumerate}
\end{definition}

\begin{definition}
 \label{def: def of zdec}
 For an integer $ N \in \sZpp{} $ and a binary sequence $ \mathbf{z} = \bigl( z(1), z(2), \ldots, z(N) \bigr) \in \{0,1\}^{N} $, the function $\zdec(\mathbf{z})$ converts $\mathbf{z}$ to its fractional decimal representation as:
 \begin{align*}
  \zdec(\mathbf{z}) \defeq \sum_{n=1}^{N} 2^{-n} \cdot z(n).
 \end{align*}
 The value $\zdec(\mathbf{z})$ lies in the interval $[0,1-2^{-N}]$.
\end{definition}

\section{Factor Graph of an Illustrative Portfolio Investment Problem}
\label{ex:local_portfolio_investment_problem_factor_graph}
In order to provide a concrete illustration, we consider a portfolio optimization problem faced by an investment firm. The firm needs to allocate capital across nine different assets, which could represent individual stocks, bonds, or other financial instruments. The goal is to select a subset of these assets to maximize a utility function that balances expected returns against investment risk, subject to several regional budget constraints.

The problem is defined by the objective to maximize the mean-variance (Markowitz) utility function~\cite{Markowitz1952,Markowitz1959,Kolm2014}:
\begin{align*}
  U(\vx) &= \underbrace{\sum_{ \na =1 }^{9} x_{\na} \cdot\mu_{\na}}_{\text{expected return}} 
  - \lambda \cdot \underbrace{\left( \sum_{\na=1}^9 x_{\na}^{2} \cdot \sigma_{\na}^{2} + \!\!\!\sum_{(\na,\na') \in \mathcal{C}}\!\! 2 x_{\na} \cdot  x_{\na'} \cdot \mathrm{cov}_{\na,\na'} \right)}_{\text{portfolio variance (risk)}},
\end{align*}
where the decision variables are $\vx = (x_{1},\ldots,x_{9})$. For simplicity in this example, we assume $x_{\na} \in \{0,1\}$, representing a decision to either include ($x_{\na}=1$) or exclude ($x_{\na}=0$) a standard block of asset $\na$ in the portfolio. This 0--1 formulation is commonly used to model discrete investment decisions such as cardinality, minimum-lot, or indivisibility constraints~\cite{Mansini1999,Mansini2015,Kolm2014}.

Each term in the expression of $ U(\vx) $ has the following correspondence.
\begin{itemize}
    \item \textbf{Expected Return:} The term $\sum x_{\na} \cdot \mu_{\na}$ is the total expected financial return of the portfolio, calculated as the sum of expected returns $(\mu_{\na})_{\na \in [9]}$ from each included asset.
    \item \textbf{Portfolio Variance (Risk):} The second term represents the total risk. It includes the individual risk of each asset (its variance $\sigma_{\na}^{2}$) and the correlated risk between pairs of assets (their covariance $\mathrm{cov}_{\na,\na'}$). The set 
    \begin{align*}
      \mathcal{C} \defeq 
      \{ (1,2),(2,3),(3,4),(3,5),(5,7),(6,7),(7,8),(8,9) \}
    \end{align*}
    contains pairs of assets whose prices tend to move together.
    \item \textbf{Risk Aversion ($\lambda$):} The coefficient $\lambda \in \sRp$ models the investor's tolerance for risk. A higher $\lambda$ signifies a more conservative investor who heavily penalizes portfolio variance.
\end{itemize}

This maximization is subject to overlapping budget constraints, which might represent, for instance, regulatory limits or strategic allocation targets for different market sectors~\cite{Kolm2014,Mansini2015}:
\begin{align*}
    \text{(US Equities)} \quad & W_{\mathrm{l},1} \leq x_{1} + x_{2} + x_{3} \leq W_{\mathrm{u},1}, \\
    \text{(European Equities)} \quad & W_{\mathrm{l},2} \leq x_{3} + x_{4} + x_{5} \leq W_{\mathrm{u},2}, \\
    \text{(Tech Sector)} \quad & W_{\mathrm{l},3} \leq x_{5} + x_{6} + x_{7} \leq W_{\mathrm{u},3}, \\
    \text{(Emerging Markets)} \quad & W_{\mathrm{l},4} \leq x_{7} + x_{8} + x_{9} \leq W_{\mathrm{u},4}.
\end{align*}
Each inequality constrains the total capital allocated to a specific group of assets to be within a lower bound ($W_{\mathrm{l},k}$) and an upper bound ($W_{\mathrm{u},k}$). Notice that some assets, like asset 3 (in both US and European equities) and assets 5 and 7, belong to multiple groups, creating complex dependencies that make the optimization challenging.

\subsubsection*{Mapping to the factor graph formalism}

\begin{figure*}[t]
\centering
\scalebox{0.75}{
\begin{tikzpicture}[
vnode/.style={circle, draw, minimum size=0.7cm, inner sep=1pt},
fnode/.style={rectangle, draw, minimum size=0.7cm, inner sep=2pt, fill=white},
node distance=1.2cm and 1.5cm
]
\input{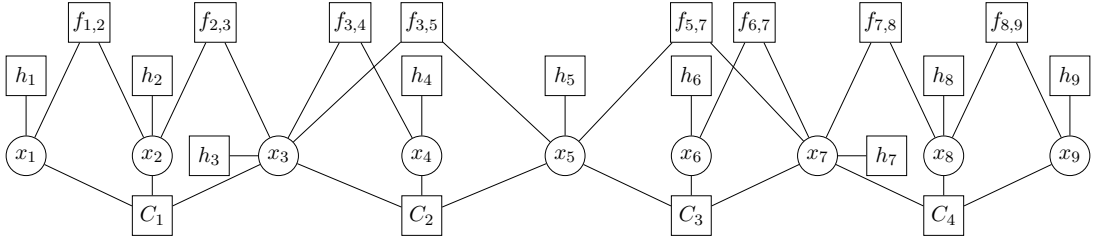}
\end{tikzpicture}
}
\caption{The factor graph $\graphG_{1}$ for the example portfolio investment problem in Example~\ref{ex:local_portfolio_investment_problem_factor_graph}, illustrating the relationships between asset decisions (circles) and financial objectives and constraints (squares).}
\label{fig: fg for a portfolio investment problem}
\end{figure*}

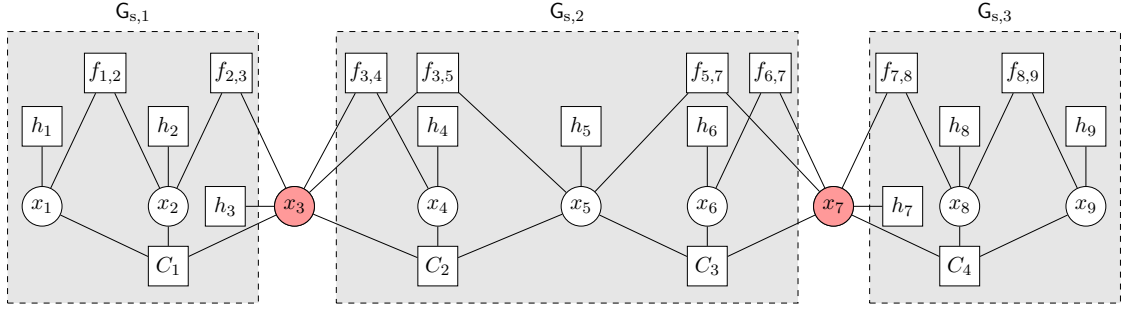
\begin{figure*}[t]
\centering
\scalebox{0.75}{
\begin{tikzpicture}[vnode/.style={circle, draw, minimum size=0.7cm, inner sep=1pt, fill=white},
 vnoder/.style={circle, draw, minimum size=0.7cm, inner sep=1pt, fill=red!40}, 
 fnode/.style={rectangle, draw, minimum size=0.7cm, inner sep=2pt, fill=white},
 node distance=1.2cm and 1.5cm,
 state_dash1/.style={shape=rectangle, draw, dashed, minimum width=4.4cm, minimum height=4.8cm, outer sep=-0.3pt, fill=black!10},
 state_dash2/.style={shape=rectangle, draw, dashed, minimum width=8.1cm, minimum height=4.8cm, outer sep=-0.3pt, fill=black!10},
 state_dash3/.style={shape=rectangle, draw, dashed, minimum width=1cm, minimum height=1.2cm, outer sep=-0.3pt, fill=black!10}]
 \begin{pgfonlayer}{above}
 \node[vnode] (x1) {$x_1$};
 \node[vnode] (x2) [right=of x1] {$x_2$};
 \node[vnode] (x3) [right=of x2] {$x_3$};
 \node[vnode] (x4) [right= 1.8 cm of x3] {$x_4$};
 \node[vnode] (x5) [right= 1.8 cm of x4] {$x_5$};
 \node[vnode] (x6) [right=of x5] {$x_6$};
 \node[vnode] (x7) [right=of x6] {$x_7$};
 \node[vnode] (x8) [right=of x7] {$x_8$};
 \node[vnode] (x9) [right=of x8] {$x_9$};
 \node[fnode] (h3) [left=0.5cm of x3] {$h_3$}; 
 \node[fnode] (h7) [right=0.5cm of x7] {$h_7$};

 \node[fnode] (h1) [above=0.7cm of x1] {$h_1$};
 \node[fnode] (h2) [above=0.7cm of x2] {$h_2$};
 \node[fnode] (h4) [above=0.7cm of x4] {$h_4$};
 \node[fnode] (h5) [above=0.7cm of x5] {$h_5$};
 \node[fnode] (h6) [above=0.7cm of x6] {$h_6$};
 \node[fnode] (h8) [above=0.7cm of x8] {$h_8$};
 \node[fnode] (h9) [above=0.7cm of x9] {$h_9$};

 \node[fnode] (f12) [above=2cm of $(x1)!0.5!(x2)$] {$f_{1,2}$};
 \node[fnode] (f23) [above=2cm of $(x2)!0.5!(x3)$] {$f_{2,3}$};
 \node[fnode] (f34) [above=2cm of $(x3)!0.5!(x4)$] {$f_{3,4}$};
 \node[fnode] (f35) [above=2cm of $(x3)!0.5!(x5)$] {$f_{3,5}$};
 \node[fnode] (f57) [above=2cm of $(x5)!0.5!(x7)$] {$f_{5,7}$};
 \node[fnode] (f67) [above=2cm of $(x6)!0.5!(x7)$] {$f_{6,7}$};
 \node[fnode] (f78) [above=2cm of $(x7)!0.5!(x8)$] {$f_{7,8}$};
 \node[fnode] (f89) [above=2cm of $(x8)!0.5!(x9)$] {$f_{8,9}$};

 \node[fnode] (I1) [below=0.7cm of $(x1)!0.5!(x3)$] {$C_{1}$};
 \node[fnode] (I2) [below=0.7cm of $(x3)!0.5!(x5)$] {$C_{2}$};
 \node[fnode] (I3) [below=0.7cm of $(x5)!0.5!(x7)$] {$C_{3}$};
 \node[fnode] (I4) [below=0.7cm of $(x7)!0.5!(x9)$] {$C_{4}$};
\end{pgfonlayer}

\begin{pgfonlayer}{main}
 \draw (x1) -- (h1); \draw (x2) -- (h2); \draw (x3) -- (h3); \draw (x4) -- (h4); \draw (x5) -- (h5);
 \draw (x6) -- (h6); \draw (x7) -- (h7); \draw (x8) -- (h8); \draw (x9) -- (h9);

 \draw (x1) -- (f12) -- (x2);
 \draw (x2) -- (f23) -- (x3);
 \draw (x3) -- (f34) -- (x4);
 \draw (x3) -- (f35) -- (x5);
 \draw (x5) -- (f57) -- (x7);
 \draw (x6) -- (f67) -- (x7);
 \draw (x7) -- (f78) -- (x8);
 \draw (x8) -- (f89) -- (x9);

 \draw (x1) -- (I1); \draw (x2) -- (I1); \draw (x3) -- (I1);
 \draw (x3) -- (I2); \draw (x4) -- (I2); \draw (x5) -- (I2);
 \draw (x5) -- (I3); \draw (x6) -- (I3); \draw (x7) -- (I3);
 \draw (x7) -- (I4); \draw (x8) -- (I4); \draw (x9) -- (I4);
\end{pgfonlayer}
 \begin{pgfonlayer}{above} 
  \node[vnoder] (x3_highlight) at (x3) {$x_3$}; 
  \node[vnoder] (x7_highlight) at (x7) {$x_7$}; 
 \end{pgfonlayer}
 \begin{pgfonlayer}{background}
 \node[state_dash1] (d1) [above=-1.7 cm of $(x1)!0.36!(x3)$, label=above: $\graphG_{\mathrm{s},1}$] {};
 \node[state_dash2] (d3) [above=-1.7 cm of $(x4)!0.48!(x6)$, label=above: $\graphG_{\mathrm{s},2}$] {};
 \node[state_dash1] (d5) [above=-1.7 cm of $(x9)!0.36!(x7)$, label=above: $\graphG_{\mathrm{s},3}$] {};
 \end{pgfonlayer}
\end{tikzpicture}
}
\caption{Decomposition of $\graphG_{1}$ (from Fig.~\ref{fig: fg for a portfolio investment problem}) using boundary variables $\setVb=\{x_3,x_7\}$ (red). Removing $\setVb$ splits $\graphG_{1}$ into $\Nsfn=3$ subgraphs $\graphG_{\mathrm{s},1}$, $\graphG_{\mathrm{s},2}$, and $\graphG_{\mathrm{s},3}$ (regions shown). The local boundary sets are $\setVb^{(1)}=\{x_3\}$, $\setVb^{(2)}=\{x_3,x_7\}$, and $\setVb^{(3)}=\{x_7\}$.}
\label{fig:decomposition_example_factor_graph_DQC_boundary_nodes_highlighted}
\end{figure*}

This optimization problem can be precisely mapped to the components of a factor graph $\graphG_{1}(\set{F}, \set{V}, \set{E}, \set{X})$ as defined in Definition~\ref{def: def of snfg}, as shown in Fig.~\ref{fig: fg for a portfolio investment problem}.

\begin{itemize}
    \item \textbf{Variable Nodes ($\set{V}$):} The set of variable nodes is $\set{V} = \{1, 2, \ldots, 9\}$. Each node $\na \in \set{V}$ corresponds to the decision variable $x_{\na}$ for the $\na$-th asset.

    \item \textbf{Function Nodes ($\set{F}$):} The function nodes represent the individual terms in the utility function and the constraints. The set $\set{F}$ is the union of three distinct subsets:
    \begin{enumerate}
        \item \textit{Individual Utility Nodes:} $\{h_{\na}\}_{\na \in \set{V}}$. Each node $h_{\na}$ corresponds to the local function $h_{\na}(x_{\na}) = x_{\na} \cdot \mu_{\na} - \lambda \cdot x_{\na}^{2} \cdot \sigma_{\na}^{2}$, which captures the individual return and risk contribution of asset $\na$.
        \item \textit{Covariance Risk Nodes:} $\{f_{\na,\na'}\}_{(\na,\na') \in \mathcal{C}}$. Each node $f_{\na,\na'}$ corresponds to the local function $f_{\na,\na'}(x_{\na}, x_{\na'}) = -2\lambda \cdot x_{\na} \cdot x_{\na'} \cdot \mathrm{cov}_{\na,\na'}$, representing the penalty from correlated risk between assets $\na$ and $\na'$.
        \item \textit{Constraint Nodes:} $\{C_k\}_{k \in [4]}$. Each node $C_k$ enforces a budget constraint. Its local function is a hard constraint: $C_k(\mathbf{x}_{\partial C_k}) = 0$ if the $k$-th budget constraint is met, and $C_k(\mathbf{x}_{\partial C_k}) = -c$ for a sufficiently large coefficient $c \in \sRpp$ otherwise. This penalizes any infeasible solution, effectively removing it from consideration.
    \end{enumerate}

    \item \textbf{Edges ($\set{E}$):} The edge set $\set{E}$ connects each variable node to the function nodes that depend on it. For example, variable node $x_3$ is connected to the function nodes $h_3$, $f_{2,3}$, $f_{3,4}$, $f_{3,5}$, $C_1$, and $C_2$.

    \item \textbf{Global Function ($g(\vx)$):} The global function to be maximized is the sum of all local functions:
    \begin{align*}
    g(\vx) &= \sum_{\na=1}^9 h_{\na}(x_{\na}) + \sum_{(\na,\na') \in \mathcal{C}} f_{\na,\na'}(x_{\na}, x_{\na'}) + \sum_{k=1}^4 C_k(\mathbf{x}_{\partial C_k}).
    \end{align*}
    By construction, maximizing $g(\vx)$ is equivalent to solving the original constrained portfolio optimization problem.
\end{itemize}
Solving such portfolio investment problems is in the complexity class NP-complete~\cite{Mansini1999}, which motivates the search for efficient quantum algorithms.


\subsection{Decomposition}

Based on Section~\ref{sec:decomposition_of_factor_graph}, a decomposition of $\graphG_{1}$ (from Fig.~\ref{fig: fg for a portfolio investment problem}) is shown in Fig.~\ref{fig:decomposition_example_factor_graph_DQC_boundary_nodes_highlighted}.

\section{Transformation to Binary and Normalized Factor Graphs}
\label{apx: binarization_normalization}

This section details a transformation from an arbitrary factor graph $\graphG$ (as per Definition~\ref{def: def of snfg}, assuming an additive global function) with general discrete alphabets and real-valued local functions, into an equivalent factor graph $\hat{\graphG}$. The transformed graph $\hat{\graphG}$ features binary variables and nonnegative-valued local functions $\hat{f}$ such that the corresponding global additive function $\hat{g}(\hat{\vx}) = \sum_{f \in \set{F}} \hat{f}(\hat{\vx}_{\partial \hat{f}})$ is normalized to the interval $[0,1]$. This transformation allows our quantum algorithms, which are designed for binary inputs and normalized functions, to address a broader class of optimization problems without loss of generality.

\begin{definition}
\label{def: appendix_binarized_normalized_factor_graph}
Consider an arbitrary factor graph $\graphG(\set{F}, \set{V}, \set{E}, \set{X})$ with local functions $f: \setx_{\setpf} \to \sR$ and an additive global function $g(\mathbf{x}) = \sum_{f \in \set{F}} f(\mathbf{x}_{\setpf})$. The transformation to an equivalent factor graph $\hat{\graphG}(\set{F}, \hat{\set{V}}, \hat{\set{E}}, \hat{\set{X}})$ proceeds in the following steps:

\begin{enumerate}
    \item \textbf{Binarization of Variables:} For each variable node $v \in \set{V}$ with a discrete alphabet $\setx_v$ of cardinality $N_v^{\mathrm{card}}$, we introduce $N_v^{\mathrm{bits}} \defeq \lceil \log_2 N_v^{\mathrm{card}} \rceil$ binary variables $\hat{x}_{v,1}, \ldots, \hat{x}_{v,N_v^{\mathrm{bits}}}$. Each value $x_v \in \setx_v$ is uniquely encoded by a binary vector 
    \begin{align*}
      \bigl( \hat{x}_{v,1}, \ldots, \hat{x}_{v,N_v^{\mathrm{bits}}} \bigr) \in \{0,1\}^{N_v^{\mathrm{bits}}}.
    \end{align*}
    If $N_v^{\mathrm{card}} < 2^{N_v^{\mathrm{bits}}}$, some binary vectors might be unused; the behavior of functions for these unused encodings must be consistently defined (\eg, by mapping them to a value that ensures they are not part of any optimal solution, or by restricting the domain of sums / optimizations to valid encodings). The set of new binary variable nodes is $\hat{\set{V}} \defeq \left\{ \hat{x}_{v,b}  \ \middle| \ v \in \set{V}, b \in [N_v^{\mathrm{bits}}] \right\}$, and their collective alphabet is $\hat{\set{X}} \defeq \{0,1\}^{|\hat{\set{V}}|}$.

    \item \textbf{Transformation of Local Functions to Binary Domain:} For each local function $f: \setx_{\setpf} \to \sR$ in $\graphG$, we define a corresponding function $f': \hat{\set{X}}_{\setpf'} \to \sR$ that operates on the binary encodings of the original variables. Specifically, $f'(\hat{\mathbf{x}}_{\setpf'}) \defeq f(\mathbf{x}_{\setpf})$, where $\hat{\mathbf{x}}_{\setpf'}$ is the collection of binary variables in $\hat{\set{V}}$ that encode the original variables $\mathbf{x}_{\setpf}$ connected to function node $f$.

    \item \textbf{Normalization of the local Function:} For each $ f \in \set{F} $, we compute the minimum and maximum values of $f'(\hat{\mathbf{x}}_{\setpf'})$ over its domain (considering only valid encodings): $f'_{\min} \defeq \min_{ \hat{\mathbf{x}}_{\setpf'} \in \hat{\set{X}}_{\mathrm{valid}, f'} } f'(\hat{\mathbf{x}}_{\setpf'})$ and $f'_{\max} \defeq \max_{ \hat{\mathbf{x}}_{\setpf'} \in \hat{\set{X}}_{\mathrm{valid}, f'} } f'(\hat{\mathbf{x}}_{\setpf'})$, where $ \hat{\set{X}}_{\mathrm{valid}, f'} $ denotes the subset of 
    $\hat{\set{X}}_{\setpf'} =\{0,1\}^{|\setpf'|}$ 
    corresponding to valid encodings of $\mathbf{x}_{\setpf} \in \set{X}_{\setpf}$. 
    Let $S_{\mathrm{range}} \defeq \sum_{f} (f'_{\max} - f'_{\min})$. The normalized local function 
    $\hat{f}: \hat{\set{X}} \to [0,1]$ is then defined to be
    \begin{align*}
      \hat{f}(\hat{\mathbf{x}}) \defeq 
      \begin{cases}
      \frac{f'(\hat{\mathbf{x}}) - f'_{\min}}{S_{\mathrm{range}}} & \text{if } 
      S_{\mathrm{range}} \neq 0 \text{ and } \hat{\mathbf{x}}_{\setpf'} \in \hat{\set{X}}_{\mathrm{valid}} \\
      0 & \text{if } S_{\mathrm{range}} = 0 \text{ and } \hat{\mathbf{x}}_{\setpf'} 
      \in \hat{\set{X}}_{\mathrm{valid}} \\
      \text{undefined} & \text{if } \hat{\mathbf{x}}_{\setpf'}
       \notin \hat{\set{X}}_{\mathrm{valid}}
      \end{cases} \quad.
    \end{align*}
    This ensures both the local function $\hat{f}(\hat{\mathbf{x}})$ and the global function 
    $ \hat{g}(\hat{\mathbf{x}}) \defeq \sum_{f \in \set{F}} 
    \hat{f}(\hat{\mathbf{x}}_{\setpf'}) = \sum_{f}(f'(\hat{\mathbf{x}}) - f'_{\min}) / S_{\mathrm{range}} $ lie within $[0,1]$ for all valid configurations.

    \item \textbf{The Transformed Factor Graph $\hat{\graphG}$:} The transformed factor graph is $\hat{\graphG}(\set{F}, \hat{\set{V}}, \hat{\set{E}}, \hat{\set{X}})$. It uses the same set of function nodes $\set{F}$ and a structurally equivalent set of edges $\hat{\set{E}}$ connecting these function nodes to the new binary variable nodes in $\hat{\set{V}}$. Each function node $f \in \set{F}$ in $\hat{\graphG}$ is now associated with the transformed local function $\hat{f}$.
\end{enumerate}
\end{definition}

The significance of this transformation lies in its preservation of the optimization problem's core structure while adapting it for algorithms that require binary inputs and normalized outputs.

\begin{proposition}
\label{prop: appendix_prop_simplification}
  Let $\graphG$ and its transformation $\hat{\graphG}$ be defined as in Definition~\ref{def: appendix_binarized_normalized_factor_graph}.
  If $S_{\mathrm{range}} \neq 0$, then a configuration $\mathbf{x}^{*} \in \set{X}$ maximizes $g(\mathbf{x})$ if and only if its corresponding valid binary configuration $\hat{\mathbf{x}}^{*} \in \hat{\set{X}}_{\mathrm{valid}}$ (derived from $\mathbf{x}^{*}$ through the binarization encoding) maximizes $\hat{g}(\hat{\mathbf{x}})$. Furthermore, $\hat{g}(\hat{\mathbf{x}}) \in [0,1]$ for all valid $\hat{\mathbf{x}} \in \hat{\set{X}}_{\mathrm{valid}}$. If $S_{\mathrm{range}} = 0$, then $g(\mathbf{x})$ is constant, any configuration $\mathbf{x} \in \set{X}$ is optimal, and $\hat{g}(\hat{\mathbf{x}}) = 0$ for all valid $\hat{\mathbf{x}} \in \hat{\set{X}}_{\mathrm{valid}}$.
\end{proposition}

\begin{proof}
  The binarization step establishes a bijective mapping between configurations $\mathbf{x} \in \set{X}$ and valid binary configurations $\hat{\mathbf{x}} \in \hat{\set{X}}_{\mathrm{valid}}$ such that $g(\mathbf{x}) = \sum_{f} f'(\hat{\mathbf{x}})$. The normalization step, $\hat{g}(\hat{\mathbf{x}}) = \sum_{f} (f'(\hat{\mathbf{x}}) - f'_{\min}) / S_{\mathrm{range}}$, is a positive affine transformation for $S_{\mathrm{range}} > 0$. Such transformations preserve the set of arguments that achieve the maximum value of the function. Since $f'(\hat{\mathbf{x}}) $ ranges from $f'_{\min}$ to $f'_{\max}$ for valid encodings, it follows directly that $\hat{g}(\hat{\mathbf{x}})$ ranges from $0$ to $1$. If $S_{\mathrm{range}} = 0$, $g(\mathbf{x})$ is constant, making all configurations optimal, and $\hat{g}(\hat{\mathbf{x}})$ is identically zero for valid encodings.
\end{proof}

This transformation ensures that optimization problems defined on general factor graphs can be converted into an equivalent form suitable for quantum algorithms designed for binary variables and normalized objective functions, thereby broadening the applicability of the methods presented in this paper.
The ability to perform such a conversion is crucial, as many quantum algorithmic primitives, such as those based on phase estimation or amplitude amplification, naturally operate on binary inputs (qubits) and expect function outputs to be mapped to a bounded range (\eg, phases or probabilities). By demonstrating this equivalence, the scope of the proposed distributed quantum computing framework is extended beyond problems that are natively binary and normalized.


\section{Register Layout for the distributed algorithm}
\label{apx:registers}
 Consider the factor graph decomposition strategy in Definition~\ref{def: setup of the distributed quantum} and arbitrary positive integers $ \Np, \Nsfn \in \sZpp $.
The quantum network employs the following collections of qubits, all initialized to the state $ \svzero $.
The initial physical location of each register collection within the quantum network (Section~\ref{def: the quantum network setup}) is specified.
For simplicity, we use $ ( \cdot )_{\np} $ and $ ( \cdot )_{\nsfn} $ instead of $ ( \cdot )_{\np \in [\Np]} $ and 
$ ( \cdot )_{\nsfn \in [\Nsfn]} $, respectively, if there is no ambiguity.

\begin{enumerate}
  \item \textbf{Boundary Variable Qubits ($ \vqB $):} Residing at the coordinator processor $\QPUc$, 
  this register $ \vqB \defeq \bigl( \qB(v) \bigr)_{\! v \in \setVb} $ comprises one qubit $\qB(v)$ for each boundary variable node $ x_{v} $ (where $v \in \setVb$) in the factor graph $ \graphG{} $. This register stores the quantum state of boundary variables, enabling coherent superposition over their configurations.

\item \textbf{Coordinator Processing Qubits ($ \vqcen $):} Located at $\QPUc$, the collection 
$ \vqcen \defeq \bigl( \vqstn{\mathrm{c}- }, \vqpn{\mathrm{c}} \bigr) $ includes:
    \begin{itemize}
     \item $\vqstn{\mathrm{c}}$: A quantum register contains $t(p_{\min,\mathrm{c}},\delta)$ qubits and is used by $\QPUc$ for phase estimation.
The function $ t(c, \delta) $, defined in expression~\eqref{eqn: def of fun t}, determines the number of qubits required for phase estimation to achieve a specified precision with a success probability $1-\delta$, where the parameter $c$ is related to the phase to be estimated (details of the phase estimation algorithm are provided in Appendix~\ref{apx:reversible_phase_estimation}).
     \item $\vqpn{\mathrm{c}}$: An $\Np$-qubit register designated for storing the $\Np$-bit approximation of the global maximum value $ \gmax $. Measurement outcomes of $\vqpn{\mathrm{c}}$ in the computational basis yield classical bit strings denoted by $ \mathbf{z}_{\rmc} \in \{0,1\}^{\Np} $.
    \end{itemize}

\item \textbf{Local Processing Qubits ($ \vqloc $):} This distributed collection 
$ \vqloc \defeq ( \vqlocn )_{\nsfn} $ consists of $ \Nsfn $ sub-collections, one for each worker processor.
Specifically, for each $\nsfn \in [\Nsfn]$, the local collection 
$ \vqlocn \defeq \bigl( \vqstn{\nsfn},\vqsrn{\nsfn},\vqpn{\nsfn} \bigr) $ resides at worker processor $\QPU_{\nsfn}$.
It is used to determine an $\Np$-bit approximation of the subgraph's local maximum value, $ g_{\nsfn,\mathbf{x}_{\setVbnsfn}}^{\max} $ (conditioned on the boundary variables $\mathbf{x}_{\setVbnsfn}$). This collection comprises:
\begin{itemize}
  \item $\vqstn{\nsfn}$: A register contains $t(p_{\min,\nsfn}, \delta)$ qubits and is used by $ \QPU_{\nsfn} $ for its local phase estimation.
  
  \item $ \vqsrn{\nsfn} $: A register consists of $|\set{V}_{\nsfn}|$ qubits. These qubits are used by $\QPU_{\nsfn}$ for implementing local oracles that provides information about the function
      $g_{\nsfn,\mathbf{x}_{\setVbnsfn}}$, 
      encoding the internal variables $\vx_{\set{V}_{\nsfn}}$ of the subgraph.

  \item $\vqpn{\nsfn}$: An $\Np$-qubit quantum register used to store the $\Np$-bit approximation 
  of $ g_{\nsfn,\mathbf{x}_{\setVbnsfn}}^{\max} $.
\end{itemize}
Measurement outcomes of the qubit register $ \vqpn{\nsfn} $ (for a single subgraph $ \nsfn $) in the computational basis yield classical bit strings $ \mathbf{z}_{\mathrm{p},\nsfn} \in \{0,1\}^{\Np} $.
The collection of all such outcomes is denoted by $ \mathbf{z}_{\mathrm{p},[\Nsfn]} \defeq ( \mathbf{z}_{\mathrm{p},\nsfn} )_{\nsfn \in [\Nsfn]} \in \prod_{\nsfn \in [\Nsfn]} \{0,1\}^{\Np} $.

\item \textbf{Auxiliary Qubits ($ \vqaux $):} The set $\vqaux$ comprises all auxiliary qubits necessary for intermediate computations throughout the algorithm. These are distributed across the network as follows:
    \begin{itemize}
    \item At the coordinator processor $\QPUc$:
        \begin{itemize}
        \item $ (\vqBauxn)_{\nsfn} $: A set of $ \Nsfn $ quantum registers. Each register $\vqBauxn$ contains $|\setVbnsfn|$ qubits. These are used by $\QPUc$ to temporarily hold or manipulate quantum information related to the boundary variables $\mathbf{x}_{\setVbnsfn}$ relevant to subgraph $\nsfn$, particularly during distribution to or retrieval from worker QPUs.
        \item $ (\vqstn{\mathrm{c},\nsfn})_{\nsfn} $: $\Nsfn$ quantum registers, each with $ t(p_{\min,\mathrm{c}}, \delta) $ qubits. These serve as auxiliary registers at $\QPUc$, potentially used to facilitate distributed controlled operations or to temporarily store the control information for phase estimation related to different subgraphs.
        
        \item $q_{\mathrm{ancilla}}$: a auxiliary qubit used by $\QPUc$ as part of a distributed implementation of the reflection operator $ \matr{U}_{\set{Z}_{\ge\mathbf{z}}} $ 
        (see Appendix~\ref{apx:implementationof_Upe}).

        \item $ q_{\mathrm{CZ},\rmc} $ and $ q_{\mathrm{global\_zero}} $: Two auxiliary qubits used by $\QPUc$ as part of a distributed implementation of the reflection operator $ 2\svzero\bra{\mathbf{0}} - \matr{I} $ 
        (see Appendix~\ref{apx:implementationof_Upe}).

        \end{itemize}
    \item At each worker processor $\QPU_{\nsfn}$:
        \begin{itemize}

          \item $ \vqoran{\nsfn} $: A register of $N_{\mathrm{ora},\nsfn}$ auxiliary qubits at $\QPU_{\nsfn}$, serving as ancillas for the local oracles.

          \item $ q_{\mathrm{CZ},\nsfn} $: A single auxiliary qubit at $\QPU_{\nsfn}$ used as part of a distributed implementation of the reflection operator $ 2\svzero\bra{\mathbf{0}} - \matr{I} $.

        \end{itemize}

      \item Define a collection of auxiliary qubits to be 
      \begin{align*}
        \hspace{-0.5cm}\vqauxc & 
        \defeq\bigl( (\vqBauxn,\, q_{\mathrm{CZ},\nsfn}, \, \vqstn{\mathrm{c},\nsfn} )_{\nsfn \in [\Nsfn]}, \,
        q_{\mathrm{CZ},\rmc}, \, q_{\mathrm{global\_zero}}
        \bigr),
      \end{align*}
      where the index ``c'' indicates that these qubits serve as ``clean ancillas'' after the implementation of the distributed quantum algorithm, \ie, the state of $\vqauxc$ remains $\svzero$. 
      For completeness, we define $\vqauxd \defeq (\vqoran{\nsfn})_{\nsfn}$, 
      where the index ``d'' indicates that these qubits may be entangled with other qubits after the implementation of the distributed quantum algorithm, \ie, these qubits are ``dirty ancillas''.
      It holds that
      \begin{align*}
        \vqaux = \bigl( \vqauxc,\, \vqauxd,\, q_{\mathrm{ancilla}} \bigr).
      \end{align*}
    \end{itemize}
\end{enumerate}
\begin{table}[t!]
  \centering
  \scriptsize
  \caption{Qubit counts in the distributed algorithm.}
  \label{table: number of qubits}
  \setlength{\tabcolsep}{7pt}
  \renewcommand{\arraystretch}{1.25}
  \begin{tabular}{@{}p{6.8cm}p{7.2cm}@{}}
    \toprule
    \textbf{Register collection (location)} & \textbf{Number of qubits} \\
    \midrule
    $\vqB$ (at $\QPUc$) &
    $|\setVb|$ \\
    \midrule
    \addlinespace[2pt]

    $\vqcen$ (at $\QPUc$) &
    $\Np + t\bigl(p_{\min,\mathrm{c}},\delta\bigr)$ \\
    \midrule
    \addlinespace[2pt]

    $\vqloc$ (total across all workers $\QPU_{\nsfn}$) &
    $\sum\limits_{\nsfn \in [\Nsfn]}\Bigl(\Np + t\bigl(p_{\min,\nsfn},\delta\bigr) + |\set{V}_{\nsfn}|\Bigr)$ \\
    \midrule
    \addlinespace[2pt]

    $\vqaux$ (total across the network) &
    $\begin{aligned}[t]
      3 + \sum_{\nsfn \in [\Nsfn]}\Bigl(&|\setVbnsfn| + N_{\mathrm{ora},\nsfn} 
      + t\bigl(p_{\min,\mathrm{c}},\delta\bigr) + 1\Bigr)
    \end{aligned}$ \\
    \bottomrule
  \end{tabular}
\end{table}
  Table~\ref{table: number of qubits} summarizes the number of qubits in these primary collections.
  The parameters $\delta \in (0,1]$ (overall error tolerance), $p_{\min,\nsfn}, p_{\min,\mathrm{c}} \in (0,1]$ (minimum success probabilities for local and coordinator operations, respectively), and $N_{\mathrm{ora},\nsfn} \in \sZpp$ (number of oracle ancillas for subgraph $\nsfn$) will be defined in subsequent sections and appendices where they are used.

\section{An Approximately Reversible Phase-Testing Subroutine}
\label{apx:reversible_phase_estimation}

This appendix formally defines a unitary subroutine based on the quantum phase estimation (QPE) algorithm, proves its key properties, and show a distributed implementation of this subroutine. This subroutine is a fundamental building block of the distributed quantum algorithm presented in Section~\ref{sec: distributed_quantum_algorithm}. Its purpose is to test for a non-zero phase associated with a given unitary oracle and to record a binary outcome on an qubit. The ability to perform this test in an approximately reversible manner, \ie, to uncompute intermediate results and restore the primary registers to their initial state with large amplitude, is a cornerstone of efficient quantum algorithm design, as it allows for the reuse of qubit registers and the preservation of quantum coherence across iterative steps. We begin by establishing the necessary definitions for the standard QPE algorithm before defining the full subroutine and proving its performance guarantees.

\begin{assumption}
Throughout this paper, the arcsine function, $\arcsin(y)$, is restricted to the range $[0, \pi/2]$ for any 
$y \in [0,1]$. Within this domain, $\arcsin$ is a bijective mapping.
\end{assumption}

\subsection{The Standard Quantum Phase Estimation Algorithm}

We first review the standard QPE algorithm~\cite{Nielsen2010}, which forms the core of our subroutine.

\begin{definition}
\label{def:setup_of_phase_estimation_apx}
The phase estimation algorithm operates under the following conditions:
\begin{enumerate}
    \item\label{initial_condition_Qdist} Let $\matr{U}_{\mathrm{pe,0}}$ be an initialization unitary oracle acting on an $\Npe$-qubit \textbf{target register}. Suppose that the initial state of the target register is 
    $ \ket{\mathbf{0}}_{\mathrm{tagt}} $ and we evolve it as 
    \begin{align}
      \ket{\psi_{\mathrm{pe,0}}}_{\mathrm{tagt}} = \matr{U}_{\mathrm{pe,0}}
      \cdot \ket{\mathbf{0}}_{\mathrm{tagt}}. \label{eqn:state_vector_oracle}
    \end{align}

    \item \label{state of the sec reg is the eigenvector of the oracle} 
    In the algorithm, we suppose that there is another unitary operator $\matr{U}_{\mathrm{o}}$ such that
    the state vector of the target register $\ket{\psi_{\mathrm{pe,0}}}_{\mathrm{tagt}} $ is a superposition of two orthonormal right-eigenvectors of $\matr{U}_{\mathrm{o}}$, \ie,
    \begin{align*}
        \ket{\psi_{\mathrm{pe,0}}}_{\mathrm{tagt}} = \frac{1}{\sqrt{2}} \cdot \left( \e^{ \imagunit \phi_p/2 } \cdot \ket{v_{\mathrm{e}}} 
        + \e^{ -\imagunit \phi_p/2 } \cdot \overline{\ket{v_{\mathrm{e}}}} \right),
    \end{align*}
	    where $\matr{U}_{\mathrm{o}} \cdot \ket{v_{\mathrm{e}}} = \e^{\imagunit \phi_p} \cdot \ket{v_{\mathrm{e}}}$ and $\matr{U}_{\mathrm{o}} \cdot \overline{\ket{v_{\mathrm{e}}}} = \e^{-\imagunit \phi_p} \cdot \overline{\ket{v_{\mathrm{e}}}}$. The phase to be estimated is $\phi_p = 2\arcsin(\sqrt{p})$ for some scalar $p \in [0,1]$. This specific superposition is characteristic of amplitude amplification subroutines (see, \eg,~\cite[Section 2]{brassard2000}). For the distributed algorithm studied here, the corresponding spectral decomposition is established in Lemma~\ref{lem:spectral_properties_U_AA}.

    \item \label{lower bound on the phase estimation precision} We assume a known positive lower bound $p_{\min}$. We want to show that if $p \geq p_{\min} > 0$, our proposed subroutine has a performance guarantee.

    \item Let $\delta \in (0,1]$ be a small positive scalar defining the error probability tolerance for a single phase-testing operation.
    \item An \textbf{estimation register} with $t(p_{\min}, \delta)$ qubits is used, where the function $t(\cdot,\cdot)$ is defined in~\eqref{eqn: def of fun t}.
\end{enumerate}

\end{definition}

\begin{definition}[QPE Procedure]
\label{def:qpe_procedure_apx}
The standard quantum phase estimation procedure is a unitary operation, denoted $\mUphamd{\matr{U}_{\mathrm{o}}}{p_{\min}, \delta}$, that acts on the estimation and target registers. It is implemented as follows:
\begin{enumerate}
    \item Initialize the estimation register to $ \svzero_{\mathrm{est}} $ and the target register to $\ket{\psi_{\mathrm{pe,0}}}_{\mathrm{tagt}}$.

    \item Apply a Hadamard gate to each qubit in the estimation register.

    \item For each $k \in \{0, \dots, t(p_{\min}, \delta)-1\}$, apply a controlled-$(\matr{U}_{\mathrm{o}})^{2^k}$ operation to the target register, where the $(k+1)$-th qubit of the estimation register serves as the control.

    \item Apply an inverse quantum Fourier transform (QFT) to the estimation register.
\end{enumerate}
The action of this procedure on the initial state is given by
\begin{align*}
    \ket{\psi_{\mathrm{pe,final}}}_{\mathrm{est,tagt}} &\defeq 
    \mUphamd{\matr{U}_{\mathrm{pe,0}},\matr{U}_{\mathrm{o}}}{p_{\min}, \delta} \cdot \bigl( \svzero_{\mathrm{est}} \otimes \ket{\psi_{\mathrm{pe,0}}}_{\mathrm{tagt}} \bigr),
\end{align*}
where $\ket{\psi_{\mathrm{pe,final}}}_{\mathrm{est,tagt}}$ is the final entangled state of the two registers.

\end{definition}

\subsection{The Approximately Reversible Subroutine and Its Properties}

We now define the full subroutine, which incorporates the QPE procedure to test for a non-zero phase in an approximately reversible manner.

\begin{definition}
\label{def:reversible_subroutine_apx}
Consider the setup from Definition~\ref{def:setup_of_phase_estimation_apx}. The approximately reversible phase-testing subroutine, denoted by the unitary operator $\mUpha{\matr{U}_{\mathrm{o}}}{p_{\min}, \delta}$, acts on the combined system of the estimation register, the target register, a single qubit $q_{\mathrm{pe,c}}$ for storing the phase information about whether $p>0$ or not, and an auxiliary qubit $q_{\mathrm{pe,aux}}$,
where the index $\mathrm{pe-sub}$ means that it is a phase estimation algorithm-based subroutine. 
The system is initialized to the state $\svzero_{\mathrm{est},\mathrm{tagt},q_{\mathrm{pe,c}},q_{\mathrm{pe,aux}}}$. The subroutine is the composition of the following six operations:
\begin{enumerate}

    \item\label{QPE_subroutine:Initialization} \textbf{Initialization} Apply the unitary oracle $ \matr{U}_{\mathrm{pe,0}} $ to the target register. 
    The state becomes $ \svzero_{\mathrm{est}} \ket{\psi_{\mathrm{pe,0}}}_{\mathrm{tagt}} \ket{0}_{q_{\mathrm{pe,c}}} \ket{0}_{q_{\mathrm{pe,aux}}} $

    \item\label{QPE_subroutine:Forward_QPE} \textbf{Forward QPE:} Apply the standard QPE procedure, \ie, applying
    \begin{align*}
      \mUphamd{\matr{U}_{\mathrm{pe,0}},\matr{U}_{\mathrm{o}}}{p_{\min}, \delta} \otimes \matr{I}_{q_{\mathrm{pe,c}},q_{\mathrm{pe,aux}}}.
    \end{align*}

    \item\label{QPE_subroutine:Conditional_Flip} \textbf{Conditional Flip:} Apply a multi-controlled NOT gate that flips the qubit $q_{\mathrm{pe,c}}$ if and only if the estimation register is in any state \emph{other than} the all-zero state $\svzero_{\mathrm{est}}$.

    \item\label{QPE_subroutine:Inverse_QPE} \textbf{Inverse QPE} Apply the inverse QPE procedure, $\mUphamd{\matr{U}_{\mathrm{o}}}{p_{\min}, \delta}^{\Herm} \otimes \matr{I}_{q_{\mathrm{pe,c}},q_{\mathrm{pe,aux}}}$.

    \item\label{QPE_subroutine:Inverse_Initialization} \textbf{Inverse Initialization}: Apply the inverse unitary oracle $ \matr{U}_{\mathrm{pe,0}}^{\Herm} $ to the target register.

    \item\label{QPE_subroutine:Error_Mitigation} \textbf{Error Mitigation}:
    Apply a controlled bit flip operation $\matr{U}_{\mathrm{error\_mitigation}}$ to the involved qubit registers and the auxiliary qubit $q_{\mathrm{pe,aux}}$ initialized to $\ket{0}$, which is defined via 
      \begin{align*}
        &\matr{U}_{\mathrm{error\_mitigation}}
         \ket{\vz_{\mathrm{est}}}
        \ket{\vz_{\mathrm{tagt}}}
        \ket{z}_{q_{\mathrm{pe,c}}}
        \ket{0}_{q_{\mathrm{pe,aux}}}
        \nonumber\\&=
        \ket{\vz_{\mathrm{est}}}
        \ket{\vz_{\mathrm{tagt}}}
        \otimes
        \begin{cases}
          \ket{0}_{q_{\mathrm{pe,c}}}\ket{1}_{q_{\mathrm{pe,aux}}} & \text{if $(\vz_{\mathrm{est}},\vz_{\mathrm{tagt}}) \neq \vect{0}$ and $z = 1$} \\
          \ket{z}_{q_{\mathrm{pe,c}}}
        \ket{0}_{q_{\mathrm{pe,aux}}} & \text{otherwise} \\
        \end{cases}\quad .
      \end{align*}
      
      
\end{enumerate}
The final state of the system is defined as
\begin{align*}
    &\ket{\psi_{\mathrm{pe,final}}}_{\mathrm{est,tagt},q_{\mathrm{pe,c}},q_{\mathrm{pe,aux}}} 
    \defeq \mUpha{\matr{U}_{\mathrm{pe,0}},\matr{U}_{\mathrm{o}}}{p_{\min}, \delta} 
    \cdot \svzero_{\mathrm{est},\mathrm{tagt},q_{\mathrm{pe,c}},q_{\mathrm{pe,aux}}}.
\end{align*}
\end{definition}


The key properties of this subroutine, which are essential for its use in our main algorithm, are summarized in the following theorem.

\begin{theorem}
  \label{thm:property_of_phase_estimation_apx}
  Consider the setup in Definition~\ref{def:setup_of_phase_estimation_apx}.
  The approximately reversible phase-testing subroutine from Definition~\ref{def:reversible_subroutine_apx} has the following properties:
  \begin{enumerate}
      \item\label{thm:property_of_phase_estimation_apx:query_complexity} \textbf{Query Complexity:} The total number of applications of the operator $\matr{U}_{\mathrm{o}}$ and its inverse version is $2^{t(p_{\min}, \delta)+1}-2$. The number of applications of the operator $\matr{U}_{\mathrm{pe,0}}$ and its inverse is $ 2 $.
      
      \item\label{thm:property_of_phase_estimation_apx:pefect_reverse} \textbf{Perfect Reversibility (Zero-Phase Case):} If $p = 0$ (implying $\phi_p = 0$), the subroutine is perfectly reversed. The final state of the system is identical to its initial state, 
      $\svzero_{\mathrm{est},\mathrm{tagt},q_{\mathrm{pe,c}},q_{\mathrm{pe,aux}}}$.
      
      \item\label{thm:property_of_phase_estimation_apx:error_confinement}
      \textbf{Error Confinement:} If $p > 0$, the final state $\ket{\psi_{\mathrm{pe,final}}}_{\mathrm{est,tagt},q_{\mathrm{pe,c}},q_{\mathrm{pe,aux}}}$ has no component in which the qubit $q_{\mathrm{pe,c}}$ is in the state $\sone$, and the estimation and target registers  and the auxiliary qubit are in any non-zero state. Formally, for any non-zero basis state $\ket{\mathbf{x}}_{\mathrm{est},\mathrm{tagt},q_{\mathrm{pe,aux}}} \neq \svzero_{\mathrm{est},\mathrm{tagt},q_{\mathrm{pe,aux}}}$, it holds that
      \begin{align*}
        \left( \bra{\mathbf{x}}_{\mathrm{est},\mathrm{tagt},q_{\mathrm{pe,aux}}}\bra{1}_{q_{\mathrm{pe,c}}} \right) \ket{\psi_{\mathrm{pe,final}}}_{\mathrm{est,tagt},q_{\mathrm{pe,c}},q_{\mathrm{pe,aux}}} = 0.
      \end{align*}

     \item\label{thm:property_of_phase_estimation_apx:appro_reverse} \textbf{Approximate Reversibility (Non-Zero Phase Case):} If $p > 0$ and we measure the final state in the computational basis, the probability of finding the ancilla flipped to $\sone$ \emph{and} the primary registers returned to their initial state 
      $\svzero_{\mathrm{est},\mathrm{tagt}}$ is at least $\bigl( 1 - \delta \bigr)^{2}$. Formally,
      \begin{align*}
          \left| \left( \bra{\mathbf{0}}_{\mathrm{est},\mathrm{tagt}} 
          \bra{1}_{q_{\mathrm{pe,c}}} \bra{0}_{q_{\mathrm{pe,aux}}} \right) \ket{\psi_{\mathrm{pe,final}}}_{\mathrm{est,tagt},q_{\mathrm{pe,c}},q_{\mathrm{pe,aux}}} \right|^2 \geq \left( 1 - \delta \right)^{2}.
      \end{align*}

\end{enumerate}
\end{theorem}
\begin{proof}
  See Section~\ref{subsection: proof of thm:property_of_phase_estimation_apx}.
\end{proof}

\begin{corollary}
The final state $\ket{\psi_{\mathrm{pe,final}}}_{\mathrm{est,tagt},q_{\mathrm{pe,c}},q_{\mathrm{pe,aux}}}$ can be written as
\begin{align*}
    \ket{\psi_{\mathrm{pe,final}}} =
    \begin{cases}
        \svzero_{\mathrm{est},\mathrm{tagt},q_{\mathrm{pe,c}},q_{\mathrm{pe,aux}}} & \text{if } p = 0 \\
        \begin{array}{l} 
          c_{1} \cdot \svzero_{\mathrm{est},\mathrm{tagt},q_{\mathrm{pe,aux}}}  
          \sone_{q_{\mathrm{pe,c}}} 
          + c_{0} \cdot \ket{\psi_{\mathrm{junk}}}_{\mathrm{est,tagt},q_{\mathrm{pe,aux}}} \szero_{q_{\mathrm{pe,c}}} 
        \end{array}
        & \text{if } p > 0
    \end{cases} \quad,
\end{align*}
where $c_0, c_1 \in \sC$ are complex coefficients satisfying $|c_1|^2 \geq (1-\delta)^2$ and $|c_0|^2 + |c_1|^2 = 1$, and where $\ket{\psi_{\mathrm{junk}}}_{\mathrm{est},\mathrm{tagt},q_{\mathrm{pe,aux}}}$ is some normalized state of the estimation register, the target register, and the auxiliary qubit $q_{\mathrm{pe,aux}}$. This structure ensures that the correct-outcome branch, \ie, $\szero_{q_{\mathrm{pe,c}}}$ for $p=0$ and $\sone_{q_{\mathrm{pe,c}}}$ for $p>0$, is always entangled with the clean state $\svzero_{\mathrm{est},\mathrm{tagt},q_{\mathrm{pe,aux}}}$. 

\end{corollary}
\begin{proof}
  This follows straightforwardly from Theorem~\ref{thm:property_of_phase_estimation_apx}.
\end{proof}


Therefore, the branch encoding the correct decision always returns the estimation register, the target register, and the auxiliary qubit $q_{\mathrm{pe,aux}}$ exactly to the state $\svzero_{\mathrm{est},\mathrm{tagt},q_{\mathrm{pe,aux}}}$. In particular, when $p \geq p_{\min} >0 $, this clean branch occurs with probability at least $(1-\delta)^2$. Consequently, conditioned on the correct outcome branch stored in $q_{\mathrm{pe,c}}$, these qubits can be reused as zero-initialized work qubits in subsequent operations.


\subsection{Proof of Theorem~\ref{thm:property_of_phase_estimation_apx}}
\label{subsection: proof of thm:property_of_phase_estimation_apx}

We prove each claim of Theorem~\ref{thm:property_of_phase_estimation_apx} in order. In the proof, we will use $ \mUphamd{\cdot}{\cdot} $ instead of $ \mUphamd{\matr{U}_{\mathrm{o}}}{p_{\min}, \delta} $ for simplicity, if there is no ambiguity.

\subsubsection{Proof of Claim~\ref{thm:property_of_phase_estimation_apx:query_complexity} (Query Complexity)}
The subroutine consists of six steps. The initialization (Step~\ref{QPE_subroutine:Initialization}) and uncomputation (Step~\ref{QPE_subroutine:Inverse_Initialization}) 
each require one application of $ \matr{U}_{\mathrm{pe,0}} $ or its inverse. The total query complexity of $ \matr{U}_{\mathrm{pe,0}} $ and its inverse is thus $ 2 $.

The forward QPE procedure (Step~\ref{QPE_subroutine:Forward_QPE}) with a $t$-qubit estimation register, where $t \defeq t(p_{\min}, \delta)$, requires $\sum_{k=0}^{t-1} 2^k = 2^t - 1$ applications of controlled-$\matr{U}_{\mathrm{o}}$ operations. The inverse QPE procedure (Step~\ref{QPE_subroutine:Inverse_QPE}) requires the same number of queries. The conditional flip (Step~\ref{QPE_subroutine:Conditional_Flip}) and the error mitigation (Step~\ref{QPE_subroutine:Error_Mitigation}) do not involve queries to $\matr{U}_{\mathrm{o}}$. Therefore, the total query complexity of $ \matr{U}_{\mathrm{o}} $ and its inverse 
is $(2^t - 1) + (2^t - 1) = 2^{t(p_{\min}, \delta)+1} - 2$.

\subsubsection{Proof of Claim~\ref{thm:property_of_phase_estimation_apx:pefect_reverse} (Perfect Reversibility)}
If $p=0$, then the phase is $\phi_p = 0$. The oracle $\matr{U}_{\mathrm{o}}$ acts as the identity on the eigenvectors $\ket{v_{\mathrm{e}}}$ and $\overline{\ket{v_{\mathrm{e}}}}$, and thus on their superposition $\ket{\psi_{\mathrm{pe,0}}}_{\mathrm{tagt}}$. Consequently, all controlled-$(\matr{U}_{\mathrm{o}})^{2^k}$ operations in the QPE procedure are identity operations. The forward QPE procedure (Step~\ref{QPE_subroutine:Forward_QPE}) simplifies to applying Hadamards to the estimation register, followed by an inverse QFT. This sequence deterministically results in the estimation register being in the state $\svzero_{\mathrm{est}}$. The state after Step~\ref{QPE_subroutine:Forward_QPE} is thus $\svzero_{\mathrm{est}} \ket{\psi_{\mathrm{pe,0}}}_{\mathrm{tagt}} \ket{0}_{q_{\mathrm{pe,c}}} \ket{0}_{q_{\mathrm{pe,aux}}}$. The conditional flip (Step~\ref{QPE_subroutine:Conditional_Flip}) is not triggered, as the estimation register is in the $\svzero_{\mathrm{est}}$ state. The qubits $q_{\mathrm{pe,c}}$ and $ q_{\mathrm{pe,aux}} $ remain in $\svzero$. The inverse QPE (Step~\ref{QPE_subroutine:Inverse_QPE}) perfectly cancels the forward QPE, and the uncomputation (Step~\ref{QPE_subroutine:Inverse_Initialization}) perfectly cancels the initialization. Then the state becomes $ \svzero_{\mathrm{est},\mathrm{tagt},q_{\mathrm{pe,c}},q_{\mathrm{pe,aux}}} $ and thus the error mitigation (Step~\ref{QPE_subroutine:Error_Mitigation}) is not triggered. Therefore, the final state is $\svzero_{\mathrm{est},\mathrm{tagt},q_{\mathrm{pe,c}},q_{\mathrm{pe,aux}}}$, which demonstrates perfect reversibility.

\subsubsection{Proof of Claims~\ref{thm:property_of_phase_estimation_apx:error_confinement} and~\ref{thm:property_of_phase_estimation_apx:appro_reverse}
(Error Confinement and Approximate Reversibility (Non-Zero Phase Case))}
The proof proceeds in two parts. First, we bound the probability that the forward QPE procedure fails to detect a non-zero phase. Second, we analyze the final state of the full subroutine to prove approximate reversibility. 

\textbf{(a) Bounding the QPE Failure Probability.}
If $p = 0$, \ie, $\phi_p = 0$, by Claim~\ref{thm:property_of_phase_estimation_apx:pefect_reverse}, we know that 
the QPE runs successfully with output $ \ket{\psi_{\mathrm{pe,final}}} =  \ket{\mathbf{0}} $.
Thus the QPE failure corresponds to obtaining the all-zero state $\svzero$ in the estimation register after the forward QPE procedure and then performing measurement on the estimation register under the condition $\phi_p > 0$. The probability of this event, $P(\text{outcome}=\mathbf{0})$, is the squared magnitude of the amplitude of the $\svzero_{\mathrm{est}}$ component in the state after Step~\ref{QPE_subroutine:Forward_QPE}.  Following the similar idea as in the analysis of the standard phase estimation algorithm in~\cite[Chapter 5.2]{Nielsen2010}, this amplitude can be written as
\begin{align*}
    \alpha_{\mathbf{0}} 
    &= \frac{1}{2^t} \cdot \sum_{k=0}^{2^t-1} \frac{1}{\sqrt{2}} 
    \cdot \Bigl( \e^{\imagunit \phi_p/2} \cdot \e^{\imagunit k \phi_p} 
    + \e^{-\imagunit \phi_p/2} \cdot \e^{-\imagunit k \phi_p} \Bigr) 
    \\
    &= \frac{\sqrt{2}}{2^t} \cdot \sum_{k=0}^{2^t-1} \cos\Biggl( \biggl( k + \frac{1}{2} \biggr) \cdot \phi_p \Biggr)
    \\&
    \overset{(a)}{=}
    \frac{\sqrt{2}}{2^t} \cdot
    \frac{\cos(2^{t-1}\phi_p) \cdot \sin(2^{t-1}\phi_p)}{\sin(\phi_p/2)}
    \\&= \frac{\sqrt{2}}{2^{t+1}} \cdot \frac{\sin(2^t\phi_p)}{\sin(\phi_p/2)},
\end{align*}
where $t \defeq t(p_{\min}, \delta)$, and where step $(a)$ follows from the trigonometric identity for a sum of cosines
\begin{align*}
  \sum_{k=0}^{n-1} \cos(a+kd) = \frac{\cos\bigl( a+(n-1) \cdot d/2 \bigr) \cdot \sin(nd/2)}{\sin(d/2)}
\end{align*}
with $n=2^t$, $a=\phi_p/2$, and $d=\phi_p$.
The probability is then bounded as follows:
\begin{align}
    P(\text{outcome}=\mathbf{0}) 
    &= |\alpha_{\mathbf{0}}|^2 
    \nonumber\\
    &= \frac{\sin^2(2^t\phi_p)}{2^{2t+1} \cdot \sin^2(\phi_p/2)} \nonumber \\
    &\overset{(a)}{\leq} \frac{1}{2^{2t+1} \cdot \sin^2(\phi_p/2)} \nonumber \\
    &\overset{(b)}{=} \frac{1}{2^{2t+1} \cdot p} \nonumber \\
    &\overset{(c)}{\leq} \frac{1}{2^{2t+1} \cdot p_{\min}} \nonumber \\
    &\overset{(d)}{\leq} \frac{1}{2^{-\log_{2}\left( \delta \cdot p_{\min} \right)} \cdot p_{\min}} \nonumber \\
    &= \delta, 
    \label{eqn:lower_bound_prob_phase_estimation}
\end{align}
where step $(a)$ follows from $|\sin(x)| \leq 1$, where step $(b)$ follows from  
\begin{align*}
  \sin(\phi_p/2) = \sin\bigl( \arcsin(\sqrt{p}) \bigr) = \sqrt{p},
\end{align*}
where step $(c)$ follows from $p \geq p_{\min} > 0$, and where step $(d)$ follows from the definition of $t(\cdot,\cdot)$ in~\eqref{eqn: def of fun t}.

\textbf{(b) Analysis of the Final State and Approximate Reversibility.}
Let  
\begin{align*}
  \ket{\psi_{\mathrm{before\_QPE}}} = \svzero_{\mathrm{est}} \ket{\psi_{\mathrm{pe,0}}}_{\mathrm{tagt}}.
\end{align*}
The state after the forward QPE (Step~\ref{QPE_subroutine:Forward_QPE}) is  
\begin{align*}
  \ket{\psi_{\mathrm{pe,final}}} 
  = \mUphamd{\cdot}{\cdot} \cdot \ket{\psi_{\mathrm{before\_QPE}}}.
\end{align*}
We can decompose this state into two orthogonal components:
\begin{align}
    \ket{\psi_{\mathrm{pe,final}}} = \alpha_0 \cdot \svzero_{\mathrm{est}} \ket{\phi_0}_{\mathrm{tagt}} + \alpha_1 \cdot \ket{\psi_{\perp}}_{\mathrm{est,tagt}},
    \label{eqn:property_of_state_psi_pe_org}
\end{align}
where $\ket{\phi_0}_{\mathrm{tagt}}$ is some suitably defined state vector for the target register associated with the state $\svzero_{\mathrm{est}}$, and where $\ket{\psi_{\perp}}$ is a normalized state in the subspace orthogonal to $\svzero_{\mathrm{est}}$, and $|\alpha_0|^2 = P(\text{outcome}=\mathbf{0}) \le \delta$ by~\eqref{eqn:lower_bound_prob_phase_estimation}.
The conditional flip (Step~\ref{QPE_subroutine:Conditional_Flip}) acts on this state, yielding
\begin{align*}
    \ket{\psi_{\mathrm{con\_flip}}} &= \alpha_0 \cdot \svzero_{\mathrm{est}} \ket{\phi_0}_{\mathrm{tagt}}  \szero_{q_{\mathrm{pe,c}}} + \alpha_1 \cdot \ket{\psi_{\perp}}_{\mathrm{est,tagt}} \sone_{q_{\mathrm{pe,c}}}.
\end{align*}
The inverse QPE (Step~\ref{QPE_subroutine:Inverse_QPE}) is then applied. The state after applying the inverse unitary oracle $ \matr{U}_{\mathrm{pe,0}}^{\Herm} $ (Step~\ref{QPE_subroutine:Inverse_Initialization}) is 
\begin{align*}
  \ket{\psi_{\mathrm{uncom}}} 
  &= (\matr{I}_{\mathrm{est}} \otimes \matr{U}_{\mathrm{pe,0}}^{\Herm} \otimes \matr{I}_{q_{\mathrm{pe,c}}}) 
  \cdot \bigl(\mUphamd{\cdot}{\cdot}^{\Herm} \otimes \matr{I}_{q_{\mathrm{pe,c}}} \bigr) \ket{\psi_{\mathrm{con\_flip}}}
  \\
  &= 
  \alpha_0 \cdot (\matr{I}_{\mathrm{est}} \otimes \matr{U}_{\mathrm{pe,0}}^{\Herm} \otimes \matr{I}_{q_{\mathrm{pe,c}}}) 
  \cdot \bigl( \mUphamd{\cdot}{\cdot}^{\Herm} \otimes \matr{I}_{q_{\mathrm{pe,c}}}\bigr) 
  \cdot \svzero_{\mathrm{est}} \ket{\phi_0}_{\mathrm{tagt}}\szero_{q_{\mathrm{pe,c}}} 
  \\&\quad+ \alpha_1 
  \cdot (\matr{I}_{\mathrm{est}} \otimes \matr{U}_{\mathrm{pe,0}}^{\Herm} \otimes \matr{I}_{q_{\mathrm{pe,c}}})
  \cdot \bigl( \mUphamd{\cdot}{\cdot}^{\Herm} \otimes \matr{I}_{q_{\mathrm{pe,c}}}\bigr) 
  \cdot \ket{\psi_{\perp}}_{\mathrm{est,tagt}} \sone_{q_{\mathrm{pe,c}}}.
\end{align*}
Note that the second term is decomposed as
\begin{align*}
  \begin{split}
  &\alpha_1 \cdot (\matr{I}_{\mathrm{est}} \otimes \matr{U}_{\mathrm{pe,0}}^{\Herm} \otimes \matr{I}_{q_{\mathrm{pe,c}}})
  \cdot \bigl( \mUphamd{\cdot}{\cdot}^{\Herm} \otimes \matr{I}_{q_{\mathrm{pe,c}}}\bigr) 
  \cdot\ket{\psi_{\perp}}_{\mathrm{est,tagt}} \sone_{q_{\mathrm{pe,c}}}
  \\&\qquad= \alpha_1 \cdot \beta_{0}
  \cdot
  \svzero_{\mathrm{est},\mathrm{tagt}}\sone_{q_{\mathrm{pe,c}}}
  + \alpha_1 \cdot \beta_{1}
   \cdot \ket{\varphi_{1}}_{\mathrm{est},\mathrm{tagt}} \sone_{q_{\mathrm{pe,c}}},
  \end{split}
\end{align*}
where $ \beta_{0},\beta_{1} \in \sC $ are suitably defined scalars such that $ |\beta_{0}|^{2} + |\beta_{1}|^{2}  =1 $ and $ \ket{\varphi_{1}}_{\mathrm{est},\mathrm{tagt}} $
consists of all non-zero components. We know that $ \matr{U}_{\mathrm{error\_mitigation}} $ is triggered only for 
the component $ \ket{\varphi_{1}}_{\mathrm{est},\mathrm{tagt}}\sone_{q_{\mathrm{pe,c}}} \ket{0}_{q_{\mathrm{pe,aux}}} $ and transform it into 
$ \ket{\varphi_{1}}_{\mathrm{est},\mathrm{tagt}} \szero_{q_{\mathrm{pe,c}}} \ket{1}_{q_{\mathrm{pe,aux}}} $.
Therefore, after the error mitigation (Step~\ref{QPE_subroutine:Error_Mitigation}), the final state becomes
\begin{align*}
  &\ket{\psi_{\mathrm{final}}}  
  \\&=
  \matr{U}_{\mathrm{error\_mitigation}} \cdot (\ket{\psi_{\mathrm{uncom}}} 
  \ket{0}_{q_{\mathrm{pe,aux}}})
  \\&= 
  \alpha_0 \cdot (\matr{I}_{\mathrm{est}} \otimes \matr{U}_{\mathrm{pe,0}}^{\Herm} \otimes \matr{I}_{q_{\mathrm{pe,c}},q_{\mathrm{pe,aux}}}) 
  \cdot \bigl( \mUphamd{\cdot}{\cdot}^{\Herm} \otimes 
  \matr{I}_{q_{\mathrm{pe,c}},q_{\mathrm{pe,aux}}}\bigr) 
  \cdot (\svzero_{\mathrm{est}}  \ket{\phi_0}_{\mathrm{tagt}} 
  \szero_{q_{\mathrm{pe,c}},q_{\mathrm{pe,aux}}} )
  \\&\quad+ \alpha_1 \cdot \beta_{1}
  \cdot \ket{\varphi_{1}}_{\mathrm{est},\mathrm{tagt}}\szero_{q_{\mathrm{pe,c}}}
  \ket{1}_{q_{\mathrm{pe,aux}}}
  + 
  \alpha_1 \cdot \beta_{0}
  \cdot
  \svzero_{\mathrm{est},\mathrm{tagt}}\sone_{q_{\mathrm{pe,c}}}
  \ket{0}_{q_{\mathrm{pe,aux}}},
\end{align*}
where for the branch associated with $\sone_{q_{\mathrm{pe,c}}}$, only the state vector $\svzero_{\mathrm{est},\mathrm{tagt}} \ket{0}_{q_{\mathrm{pe,aux}}}$ appears, which proves Claim~\ref{thm:property_of_phase_estimation_apx:error_confinement}.

Now we prove Claim~\ref{thm:property_of_phase_estimation_apx:appro_reverse}. We are interested in the amplitude associated with the target state $\svzero_{\mathrm{est},\mathrm{tagt}} \sone_{q_{\mathrm{pe,c}}} \ket{0}_{q_{\mathrm{pe,aux}}}$. This component can only arise from the second term in $\ket{\psi_{\mathrm{con\_flip}}}$. The amplitude is
\begin{align*}
    &\hspace{-0.25 cm}
    \bigl( \bra{\mathbf{0}}_{\mathrm{est},\mathrm{tagt}} \bra{1}_{q_{\mathrm{pe,c}}}
    \bra{0}_{q_{\mathrm{pe,aux}}}
    \bigr) \ket{\psi_{\mathrm{final}}} 
    \nonumber\\
    &= \alpha_1 \cdot \bigl( \bra{\mathbf{0}}_{\mathrm{est},\mathrm{tagt}} 
    \bra{1}_{q_{\mathrm{pe,c}}} \bra{0}_{q_{\mathrm{pe,aux}}} \bigr)
    \quad\cdot \matr{U}_{\mathrm{error\_mitigation}} \cdot
    \bigl( \matr{I}_{\mathrm{est}} \otimes \matr{U}_{\mathrm{pe,0}}^{\Herm} 
    \otimes \matr{I}_{q_{\mathrm{pe,c}},q_{\mathrm{pe,aux}}} \bigr) 
    \nonumber\\&\quad\cdot \bigl( \mUphamd{\cdot}{\cdot}^{\Herm} \otimes 
    \matr{I}_{q_{\mathrm{pe,c}},q_{\mathrm{pe,aux}}} \bigr) \cdot \bigl(\ket{\psi_{\perp}}_{\mathrm{est,tagt}} 
     \ket{1}_{q_{\mathrm{pe,c}}} \ket{0}_{q_{\mathrm{pe,aux}}} \bigr) \nonumber \\
    &\overset{(a)}{=} \alpha_1 \cdot \bigl( \bra{\mathbf{0}}_{\mathrm{est},\mathrm{tagt}} 
    \bra{1}_{q_{\mathrm{pe,c}}} \bra{0}_{q_{\mathrm{pe,aux}}} \bigr)
    \cdot
    \bigl(\matr{I}_{\mathrm{est}} \otimes \matr{U}_{\mathrm{pe,0}}^{\Herm} \otimes \matr{I}_{q_{\mathrm{pe,c}},q_{\mathrm{pe,aux}}} \bigr) 
    \nonumber\\&\qquad \quad\!\cdot \bigl( \mUphamd{\cdot}{\cdot}^{\Herm} \otimes \matr{I}_{q_{\mathrm{pe,c}},q_{\mathrm{pe,aux}}} \bigr)
    \cdot (\ket{\psi_{\perp}}_{\mathrm{est,tagt}} 
     \ket{1}_{q_{\mathrm{pe,c}}} \ket{0}_{q_{\mathrm{pe,aux}}} ) \nonumber \\
    &= \alpha_1 \cdot \bra{\mathbf{0}}_{\mathrm{est},\mathrm{tagt}}
    \bigl( \matr{I}_{\mathrm{est}} \otimes \matr{U}_{\mathrm{pe,0}}^{\Herm} \bigr) \cdot \mUphamd{\cdot}{\cdot}^{\Herm} \cdot \ket{\psi_{\perp}}_{\mathrm{est,tagt}} 
    \nonumber \\
    &\overset{(b)}{=} \alpha_1 \cdot \bra{\psi_{\mathrm{before\_QPE}}}  \mUphamd{\cdot}{\cdot}^{\Herm} \cdot \ket{\psi_{\perp}}_{\mathrm{est,tagt}}
    \nonumber\\&\overset{(c)}{=} 
    \underbrace{\braket{\psi_{\mathrm{before\_QPE}}}{\psi_{\mathrm{before\_QPE}}}}_{=1} 
     - 
    \alpha_0 \cdot \underbrace{\bra{\psi_{\mathrm{before\_QPE}}} 
    \mUphamd{\cdot}{\cdot}^{\Herm}\cdot (\svzero_{\mathrm{est}} \ket{\phi_0}_{\mathrm{tagt}})
    }_{\overset{(d)}{=}\overline{\alpha_0} }
    \nonumber\\&= 1-|\alpha_0|^2,
\end{align*}
where step $(a)$ follow the definition of $ \matr{U}_{\mathrm{error\_mitigation}} $, \ie, this conditional flip operation acts as an identity matrix for the state vector 
$ \bra{\mathbf{0}}_{\mathrm{est},\mathrm{tagt}} \bra{1}_{q_{\mathrm{pe,c}}} \bra{0}_{q_{\mathrm{pe,aux}}} $,
where step $(b)$ follows from~\eqref{eqn:state_vector_oracle}, \ie, 
\begin{align*}
  \ket{\psi_{\mathrm{before\_QPE}}} 
  = \svzero_{\mathrm{est}} \ket{\psi_{\mathrm{pe,0}}}_{\mathrm{tagt}}
  =
  (\matr{I}_{\mathrm{est}} \otimes 
  \matr{U}_{\mathrm{pe,0}})
  \cdot \svzero_{\mathrm{est},\mathrm{tagt}},
\end{align*}
where step $(c)$ follows from~\eqref{eqn:property_of_state_psi_pe_org}, \ie,
\begin{align}  
  \ket{\psi_{\mathrm{before\_QPE}}}
  &= 
  \mUphamd{\cdot}{\cdot}^{\Herm} \cdot \ket{\psi_{\mathrm{pe,final}}} 
  \nonumber\\&= 
  \alpha_0 \cdot \mUphamd{\cdot}{\cdot}^{\Herm} \cdot (\svzero_{\mathrm{est}} \ket{\phi_0}_{\mathrm{tagt}})
  + \alpha_1 \cdot \mUphamd{\cdot}{\cdot}^{\Herm} \cdot \ket{\psi_{\perp}}_{\mathrm{est},\mathrm{tagt}},
  \label{eqn:decomposition_psi_before_QPE}
\end{align}
which implies
\begin{align*}
  &\alpha_1 \cdot \mUphamd{\cdot}{\cdot}^{\Herm} \cdot \ket{\psi_{\perp}}_{\mathrm{est},\mathrm{tagt}}
  =\ket{\psi_{\mathrm{before\_QPE}}} - 
  \alpha_0 \cdot \mUphamd{\cdot}{\cdot}^{\Herm} \cdot (\svzero_{\mathrm{est}} \ket{\phi_0}_{\mathrm{tagt}}),
\end{align*}
and where step $(d)$ follows from the decomposition in~\eqref{eqn:decomposition_psi_before_QPE}
and the fact that $\ket{\psi_{\perp}}$ is a normalized state in the subspace orthogonal to $\svzero_{\mathrm{est}}$.
The probability of the outcome $ \ket{\mathbf{0}}_{\mathrm{est},\mathrm{tagt}} \ket{1}_{q_{\mathrm{pe,c}}} \ket{0}_{q_{\mathrm{pe,aux}}} $ is
\begin{align*}
   \left| \left( \bra{\mathbf{0}}_{\mathrm{est},\mathrm{tagt}} \bra{1}_{q_{\mathrm{pe,c}}} 
   \bra{0}_{q_{\mathrm{pe,aux}}}\right) \ket{\psi_{\mathrm{final}}} \right|^{2} 
   =  (1-|\alpha_0|^2)^2
   \overset{(a)}{\geq} (1-\delta)^2,
\end{align*}
where step $(a)$ follows from the bound $|\alpha_0|^2 \le \delta$ proven in~\eqref{eqn:lower_bound_prob_phase_estimation}.
This proves Claim~\ref{thm:property_of_phase_estimation_apx:appro_reverse}.

This concludes the proof.

\section{Distributed Implementation of the Approximately Reversible Subroutine \texorpdfstring{in~\eqref{eqn: controll phase estimation_main_alg}}{}}
\label{apx:implementationof_Upe}

This section presents a rigorous, distributed implementation of the controlled, approximately reversible subroutine shown in~\eqref{eqn: controll phase estimation_main_alg}. This operation is the core computational primitive of the iterative algorithm in Algorithm~\ref{alg:dist-opt}. The implementation is designed for the quantum network model specified in Section~\ref{def: the quantum network setup}, which comprises a coordinator processor, $\QPUc$, and $\Nsfn$ worker processors, $(\QPU_{\nsfn})_{\nsfn \in [\Nsfn]}$. We assume the availability of LOCC, pre-shared EPR pairs for quantum teleportation, and classical channels for synchronization between processors.

Our exposition is structured hierarchically to ensure clarity. First, we detail the distributed protocols for the fundamental unitary operators that constitute the amplitude amplification operator $\matr{U}_{\mathrm{AA},\mathsf{N}, (\mathbf{z}_{\mathrm{prefix}}, 1)}$. Second, we formalize a general protocol for applying a controlled-unitary operation in a distributed manner. Finally, we assemble these components to construct the complete, approximately reversible subroutine, providing a full accounting of the required quantum and classical resources.

\subsection{Distributed Implementation of Core Unitary Operations}
\label{sec:distributed_implementation}

The subroutine relies on the repeated application of the amplitude amplification operator $\matr{U}_{\mathrm{AA},\mathsf{N}, (\mathbf{z}_{\mathrm{prefix}}, 1)}$, which, as defined in~\eqref{eqn: def of UAANz}, is a composition of three distinct unitary operations: the state preparation oracle $\matr{U}_{\mathrm{ini}}$, the marking oracle $\matr{U}_{\set{Z}_{\geq\mathbf{z}}}$, and the reflection operator $(2\svzero\bra{\mathbf{0}}-\matr{I})_{\mathrm{all}}$. The distributed implementation of each is detailed below.

\subsubsection{State Preparation Oracle (\texorpdfstring{$\matr{U}_{\mathrm{ini}}$}{U\_ini}) and its Inverse}
The state preparation oracle $\matr{U}_{\mathrm{ini}}$ and its inverse are implemented according to the distributed procedure in Algorithm~\ref{alg:u_ini_protocol}. This protocol, coordinated by $\QPUc$, initializes the system by creating a coherent superposition of states, each corresponding to a different boundary variable configuration. In summary, the protocol proceeds in four main stages:
\begin{enumerate}[label=1.\arabic*)]
    \item \textbf{Superposition and Entanglement at $\QPUc$:} The coordinator applies Hadamard gates to its boundary register $\vqB$ and uses CNOT gates to entangle this state with auxiliary registers $(\vqBauxn)_{\nsfn \in [\Nsfn]}$.
    \item \textbf{Distribution to Workers:} $\QPUc$ teleports each entangled auxiliary register $\vqBauxn$ to the corresponding worker processor $\QPU_{\nsfn}$.
    \item \textbf{Parallel Local Maximization:} All worker processors execute the quantum maximization algorithm $\matr{U}_{\max}$ from Appendix~\ref{apx:generic_maximization_subroutine} in parallel, controlled by the received boundary information.
    \item \textbf{Return of Results and Uncomputation:} Each worker teleports its result register $\vqpn{\nsfn}$ and the auxiliary register $\vqBauxn$ back to $\QPUc$, and then $\QPUc$ uncomputes the initial entanglement between $\vqB$ and auxiliary registers $(\vqBauxn)_{\nsfn \in [\Nsfn]}$.
\end{enumerate}
The inverse operation, $\matr{U}_{\mathrm{ini}}^{\Herm}$, is constructed by executing these stages in reverse order and conjugating each unitary operation.

\subsubsection{Marking Unitary Operator (\texorpdfstring{$\matr{U}_{\set{Z}_{\ge\mathbf{z}}}$}{U\_Z})}
The marking oracle $\matr{U}_{\set{Z}_{\ge\mathbf{z}}}$ applies a phase of $-1$ to computational basis states corresponding to the set $\set{Z}_{\ge\mathbf{z}}$, which is defined by the arithmetic condition $\sum_{\nsfn \in [\Nsfn]} \zdec(\mathbf{z}_{\mathrm{p},\nsfn}) \ge \zdec(\mathbf{z})$. This operation acts on the collection of local result registers $(\vqpn{\nsfn})_{\nsfn \in [\Nsfn]}$. In order to implement this multi-register comparison as a single coherent operation, the registers must be physically co-located. Therefore, the protocol is executed locally at the coordinator processor, $\QPUc$, after gathering the necessary registers from the workers.

\begin{enumerate}[label=2.\arabic*)]
    \item \textbf{Gather Result Registers:} For each worker $\nsfn \in [\Nsfn]$, the processor $\QPU_{\nsfn}$ teleports its $\Np$-qubit result register $\vqpn{\nsfn}$ to $\QPUc$. This step consumes $\Nsfn \cdot \Np$ EPR pairs and requires classical communication for teleportation corrections. 
    \item\label{step:Perform Quantum Arithmetic} \textbf{Perform Quantum Arithmetic:} $\QPUc$ executes a quantum arithmetic circuit that performs the following actions:
    \begin{enumerate}
        \item It computes the sum $\sum_{\nsfn} \zdec(\mathbf{z}_{\mathrm{p},\nsfn})$ from the registers $(\vqpn{\nsfn})_{\nsfn}$, storing the result in a local auxiliary register $q_{\mathrm{CZ},\rmc}$.

        \item It compares this sum with the classical value $\zdec(\mathbf{z})$. The binary result of this comparison (true or false) is stored in a single ancilla qubit, $q_{\mathrm{ancilla}}$.
    \end{enumerate}
    \item \textbf{Apply Conditional Phase:} $\QPUc$ applies a controlled-Z gate to the system, controlled by the ancilla qubit $q_{\mathrm{ancilla}}$. This imparts the required phase of $-1$ to the appropriate states.
    \item \textbf{Uncompute and Return:} $\QPUc$ reverses the quantum arithmetic circuit from Step~\ref{step:Perform Quantum Arithmetic} to return the ancillas $q_{\mathrm{CZ},\rmc}$ and $q_{\mathrm{ancilla}}$ to the $\svzero_{q_{\mathrm{CZ},\rmc},q_{\mathrm{ancilla}}}$ state. Subsequently, it teleports each register $\vqpn{\nsfn}$ back to its respective worker processor $\QPU_{\nsfn}$. This step consumes an additional $\Nsfn \cdot \Np$ EPR pairs.
\end{enumerate}

\subsubsection{Reflection Operator (\texorpdfstring{$(2\svzero\bra{\mathbf{0}}-\matr{I})_{\mathrm{all}}$}{2|0><0|-I})}
\label{sec:condition_flip}
The reflection operator $(2\svzero\bra{\mathbf{0}}-\matr{I})_{\mathrm{all}}$ acts as the identity on the all-zero state $\svzero$ of the combined registers $(\vqB,\vqloc,\vqaux)$ and applies a phase of $-1$ to every orthogonal state. This is a distributed operation requiring network-wide coordination. The protocol corrects a logical flaw in the original description and proceeds as follows.

\begin{enumerate}[label=3.\arabic*)]
    \item\label{control_Z_flip_teleport:1} \textbf{Local Zero-State Check:} Each processor ($\QPUc$ and each worker $\QPU_{\nsfn}$) checks if all its local data qubits participating in the reflection are in the $\ket{0}$ state. This is achieved by:
    \begin{enumerate}
        \item Applying Pauli-$\matr{X}$ gates to all participating local data qubits.
        \item Applying a multi-controlled Toffoli gate where the controls are all the aforementioned data qubits and the target is a local ancilla qubit ($q_{\mathrm{CZ},\rmc}$ at $\QPUc$, $q_{\mathrm{CZ},\nsfn}$ at $\QPU_{\nsfn}$), initially in the $\ket{0}$ state.
        \item Re-applying Pauli-$\matr{X}$ gates to the data qubits to return them to their original state.
    \end{enumerate}
    After this step, the local ancilla is flipped to $\ket{1}$ if and only if all local data qubits were in the $\ket{0}$ state.
    \item\label{step:dist_phase_flip:gather_coordinate} \textbf{Gather Ancillas:} Each worker $\QPU_{\nsfn}$ teleports its ancilla qubit $q_{\mathrm{CZ},\nsfn}$ to $\QPUc$.
    \item \textbf{Global Zero-State Check:} $\QPUc$ performs a multi-controlled Toffoli gate where the control qubits are its own ancilla $q_{\mathrm{CZ},\rmc}$ and all received ancillas from the workers. The target is the global ancilla qubit $q_{\mathrm{global\_zero}}$. This target qubit is flipped to $\ket{1}$ if and only if all local ancillas were in the $\ket{1}$ state, indicating that the entire distributed system is in the $\svzero$ state.

    \item\label{Apply Global Phase} \textbf{Apply Global Phase:} $\QPUc$ applies a $-\matr{Z}$ gate to the global ancilla $q_{\mathrm{global\_zero}}$. 
    This imparts phases of $+1$ and $-1$ to the state components with $q_{\mathrm{global\_zero}}=\ket{1}$ and $q_{\mathrm{global\_zero}}=\ket{0}$, corresponding to 
    the entire distributed system being in the $\svzero$ state and in an arbitrary non-zero state, respectively.
    \item\label{step:dist_phase_flip:uncomputation}\label{control_Z_flip_teleport:2} \textbf{Uncomputation:} The protocol is reversed to disentangle all ancillas and return them to the $\svzero$ state. This involves:
    \begin{enumerate}
        \item Reversing the global Toffoli gate from Step 3.3.
        \item Teleporting the worker ancillas from $\QPUc$ back to their respective $\QPU_{\nsfn}$.
        \item Reversing the local zero-state check from Step 3.1 at each processor.
    \end{enumerate}
\end{enumerate}

\subsection{Protocol for Distributed Controlled-Unitary Application}
\label{sec:control_unitary}

The QPE algorithm requires the application of controlled powers of a unitary operator. In order to structure this clearly and avoid repetition, we define a general protocol for applying a distributed unitary $\matr{U}$ for $k$ times, conditioned on a control qubit $q_c$ (from the register $\vqstn{\mathrm{c}}$) at $\QPUc$.

\subsubsection*{Protocol 1: Distributed Controlled-$U^k$ Application}
\begin{itemize}
    \item \textbf{Objective:} To implement the operation $C\text{-}\matr{U}^k$, where $\matr{U}$ is a distributed unitary, and the control qubit $q_c$ (from the register $\vqstn{\mathrm{c}}$) resides at $\QPUc$.
    \item \textbf{Procedure:}
    \begin{enumerate}
        \item \textbf{Control Information Distribution:} For each worker $\nsfn \in [\Nsfn]$:
        \begin{enumerate}
            \item At $\QPUc$, apply a CNOT gate with $q_c$ (from the register $\vqstn{\mathrm{c}}$) as the control and a local auxiliary qubit $q_{\mathrm{aux},\nsfn}$ (from the register $\vqstn{\mathrm{c},\nsfn}$) as the target.
            \item $\QPUc$ sends the qubit $q_{\mathrm{aux},\nsfn}$ to $\QPU_{\nsfn}$ via quantum teleportation. This qubit now serves as the local control at $\QPU_{\nsfn}$.
            \item After all workers receive their control qubits, $\QPUc$ broadcasts a classical synchronization signal, `CONTROL\_DISTRIBUTED'.
        \end{enumerate}
        \item \textbf{Synchronized Execution:} Upon receiving the signal, all processors execute the following loop for $j=1$ to $k$:
        \begin{enumerate}
            \item All processors (coordinator and workers) apply their respective components of the distributed unitary $\matr{U}$. Every local quantum gate within this distributed operation is made conditional on the corresponding local control qubit ($q_c$ at $\QPUc$, $q_{\mathrm{aux},\nsfn}$ at each $\QPU_{\nsfn}$).
            \item After completing one application of $C\text{-}\matr{U}$, all processors wait for a classical synchronization signal, `ITERATION\_COMPLETE', from $\QPUc$ before proceeding to the next iteration (if $j < k$).
        \end{enumerate}
        \item \textbf{Control Information Uncomputation:} After $k$ iterations, $\QPUc$ broadcasts a classical signal, `UNCOMPUTE\_CONTROL'. The operations of Step 1 are then executed in reverse to disentangle and reset all auxiliary control qubits to the $\svzero$ state.
    \end{enumerate}
\end{itemize}

\subsection{Distributed Implementation of the Approximately Reversible Subroutine}
\label{sec:Implementation_Subroutine}
We now assemble the preceding components to realize the full distributed implementation of the approximately reversible subroutine from~\eqref{eqn: controll phase estimation_main_alg}. The protocol coherently applies a phase estimation procedure conditional on the state of the prefix register $\vqpc(1:\np-1)$ at $\QPUc$.

\subsubsection*{Protocol 2: Distributed Implementation of \texorpdfstring{$\mUpha{\cdot}{\cdot}$}{U\_pe-sub}}
\begin{itemize}
    \item \textbf{Objective:} To implement the controlled unitary operation in~\eqref{eqn: controll phase estimation_main_alg} in a distributed and approximately reversible manner.
    \item \textbf{Procedure:}
    \begin{enumerate}
        \item \textbf{Forward Quantum Phase Estimation (QPE):}
        \begin{enumerate}
            \item \textbf{Initialize Estimation Register (at $\QPUc$):} The coordinator applies a Hadamard transform, $\matr{H}^{\otimes t(p_{\min,\rmc},\delta)}$, to its local estimation register $\vqstn{\mathrm{c}}$.
            \item \textbf{Apply Controlled Unitaries:} For each qubit index $t'$ from $1$ to $t(p_{\min,\rmc},\delta)$:
            \begin{itemize}
                \item The network executes \textbf{Protocol 1} with the $t'$-th qubit of $\vqstn{\mathrm{c}}$ as the control qubit $q_c$
                and the $t'$-th qubit of $ \vqstn{\mathrm{c},\nsfn} $ as the local auxiliary qubit.
                \item The target unitary $\matr{U}$ is $\matr{U}_{\mathrm{AA},\mathsf{N}, (\mathbf{z}_{\mathrm{prefix}}, 1)}$. Its application is additionally controlled by the state of the prefix register $\vqpc(1:\np-1)$. 
                Note that only the implementation of $\matr{U}_{\set{Z}_{\ge\mathbf{z}}}$ needs this control register and $\matr{U}_{\set{Z}_{\ge\mathbf{z}}}$ is implemented by $\QPUc$, \ie, $\QPUc$ does not need to teleport any auxiliary qubits related to 
                $\vqpc(1:\np-1)$ to any worker QPU.            
                
                \item The number of repetitions is $k=2^{t'-1}$.
            \end{itemize}
        \end{enumerate}
        \item \textbf{Local Phase Readout (at $\QPUc$):}
        \begin{enumerate}
            \item \textbf{Inverse Quantum Fourier Transform:} $\QPUc$ applies the inverse QFT to the estimation register $\vqstn{\mathrm{c}}$.
            \item \textbf{Conditional Flip:} $\QPUc$ applies a multi-controlled NOT gate that flips the target result qubit $\qpc(\np)$ if and only if the state of the estimation register $\vqstn{\mathrm{c}}$ is not $\svzero$. This is implemented by applying Pauli-$\matr{X}$ gates to all qubits in $\vqstn{\mathrm{c}}$, followed by a standard multi-controlled Toffoli gate targeting $\qpc(\np)$, and finally re-applying the Pauli-$\matr{X}$ gates to $\vqstn{\mathrm{c}}$.
        \end{enumerate}
        \item \textbf{Uncomputation for Reversibility:} In order to ensure the subroutine is approximately reversible, the operations of the forward QPE are undone.
        \begin{enumerate}
            \item \textbf{Reverse Phase Readout:} $\QPUc$ applies the forward QFT to $\vqstn{\mathrm{c}}$ (reversing Step 2(a)).
            \item \textbf{Reverse QPE:} For each qubit index $t'$ from $t(p_{\min,\rmc},\delta)$ down to $1$ (in reverse order), the network executes the inverse of \textbf{Protocol 1}.

            \item \textbf{Reverse Initialization (at $\QPUc$):} The coordinator applies a Hadamard transform, $\matr{H}^{\otimes t(p_{\min,\rmc},\delta)}$, to its local estimation register $\vqstn{\mathrm{c}}$.
        \end{enumerate}
        
                \item \textbf{Error mitigation}: For the first invocation, \ie, for $\np=1$, the ancilla qubit $q_{\mathrm{ancilla}}$ is initialized to $\ket{0}$ and is returned to $\ket{0}$ after every application of $\matr{U}_{\set{Z}_{\ge\mathbf{z}}}$ by Step~2.4 of the marking-oracle implementation. By Theorem~\ref{thm:property_of_phase_estimation_apx}, for $\np \geq 2$, the same reset property is guaranteed on the clean component inherited from the previous invocations, namely the component on which the prefix register $\vqpc(1:\np-1)$ stores the correct prefix and all workspace qubits have been restored to their zero states after implementing $\mUpha{\cdot}{\cdot}$. On this clean component, $q_{\mathrm{ancilla}}$ is again available in the state $\ket{0}$ and can therefore be reused as the auxiliary qubit required in the present step. The quantum network then executes the distributed zero-state-check procedure from Section~\ref{sec:condition_flip} on all qubits involved in Protocol~2 except $\qpc(\np)$ and $q_{\mathrm{ancilla}}$. The procedure is identical to that in Section~\ref{sec:condition_flip}, except that Step~\ref{Apply Global Phase} is replaced by the following operation:
        \begin{enumerate}
            \item[3.4)] Conditioned on $q_{\mathrm{global\_zero}}$ being in the state $\ket{0}$, the coordinator $\QPUc$ applies a SWAP gate to the pair $\bigl(\qpc(\np), q_{\mathrm{ancilla}}\bigr)$. On the clean component entering this step, where $q_{\mathrm{ancilla}}$ is in the state $\ket{0}$, this maps $\ket{1}_{\qpc(\np)}\ket{0}_{q_{\mathrm{ancilla}}}$ to $\ket{0}_{\qpc(\np)}\ket{1}_{q_{\mathrm{ancilla}}}$, while leaving $\ket{0}_{\qpc(\np)}\ket{0}_{q_{\mathrm{ancilla}}}$ unchanged. Equivalently, this step is triggered precisely when the joint state of all qubits participating in the zero-state test is not the all-zero state.
        \end{enumerate}
    \end{enumerate}
\end{itemize}
Combining Item~\ref{thm:property_of_phase_estimation_apx:pefect_reverse} and Item~\ref{thm:property_of_phase_estimation_apx:appro_reverse} of Theorem~\ref{thm:property_of_phase_estimation_apx}, each execution of the phase-testing primitive $\mUpha{\cdot}{\cdot}$ admits a clean component in which the decision qubit stores the correct outcome and all workspace qubits, including $q_{\mathrm{ancilla}}$, are returned to their zero states. For a single execution, this clean component has norm at least $1-\delta$, and hence probability at least $(1-\delta)^2$.

Therefore, $q_{\mathrm{ancilla}}$ can be reused across successive executions of~\eqref{eqn: controll phase estimation_main_alg} on the clean component. More precisely, after the first $\np$ invocations, the component on which the prefix register $\vqpc(1:\np)$ stores the correct $\np$-bit prefix and $q_{\mathrm{ancilla}}$ is restored to $\ket{0}$ after each invocation has norm at least $(1-\delta)^{\np}$ and therefore probability at least $(1-\delta)^{2\np}$. The same argument applies to the other workspace registers that are uncomputed on the clean component.

This sequence of protocols provides a complete and feasible distributed implementation of the core subroutine. The explicit enumeration of steps, including classical synchronization signals and the handling of distributed control, demonstrates that the proposed algorithm is well-defined within the established model of distributed quantum computation.



\section{Proof of Theorem~\ref{thm: overall_algorithm_properties}}
\label{apx: proof_of_thm_overall_algorithm_properties}

This appendix provides a rigorous proof for the three performance guarantees of the distributed quantum algorithm asserted in Theorem~\ref{thm: overall_algorithm_properties}. We will establish each of the three claims of the theorem in turn: the query complexity, the success probability and outcome correctness, and the approximation precision.

The core of the algorithm is an iterative application of a phase-testing subroutine to a sequence of amplitude amplification operators. This structure is used to perform a binary search for the optimal value $\gmax$. The proof's foundation therefore rests on analyzing the spectral properties of the amplitude amplification operators and leveraging the performance guarantees of the underlying phase-testing subroutine, as detailed in Appendix~\ref{apx:reversible_phase_estimation}.

\subsection{Proof of Item~\ref{prop: overall_query_complexity: final} (Query Complexity to \texorpdfstring{$\matr{U}_{\mathrm{ini}}$}{U\_ini})}

The query complexity is determined by the total number of applications of the state preparation oracle $\matr{U}_{\mathrm{ini}}$ and its inverse $\matr{U}_{\mathrm{ini}}^{\Herm}$. The algorithm proceeds through an iterative loop that executes for $\Np$ steps, one for each bit of the final approximation.

In each iteration $\np \in [\Np]$, Step~\ref{step:threshold-test} of Algorithm~\ref{alg:dist-opt} executes the approximately reversible phase-testing subroutine $\mUpha{\cdot}{\cdot, \cdot}$ from Definition~\ref{def:reversible_subroutine_apx}. The core unitary operator for this subroutine is $\matr{U}_{\mathrm{o}} \defeq \matr{U}_{\mathrm{AA},\mathsf{N}, ( \mathbf{z}_{\mathrm{prefix}}, 1) }$, and the initial state preparation is performed by $\matr{U}_{\mathrm{pe,0}} = \matr{U}_{\mathrm{ini}}$.

From Definition~\ref{def: the set of operators UAA}, a single application of the operator $\matr{U}_{\mathrm{o}}$ requires one application of $\matr{U}_{\mathrm{ini}}$ and one application of its inverse $\matr{U}_{\mathrm{ini}}^{\Herm}$. According to Theorem~\ref{thm:property_of_phase_estimation_apx} (Item~\ref{thm:property_of_phase_estimation_apx:query_complexity}), a single execution of the subroutine $\mUpha{\matr{U}_{\mathrm{pe,0}}, \matr{U}_{\mathrm{o}}}{p_{\min,\mathrm{c}}, \delta}$ requires $2^{t(p_{\min,\mathrm{c}}, \delta)+1} - 2$ controlled applications of $\matr{U}_{\mathrm{o}}$ and one application of $ \matr{U}_{\mathrm{pe,0}} $ and its inverse.

Therefore, the total number of applications of the $\matr{U}_{\mathrm{ini}}$ and its inverse per iteration index by $\np \in [\Np]$ is given by
\begin{align*}
    2 \cdot \left( 2^{t(p_{\min,\mathrm{c}}, \delta)+1} - 2 \right) + 2 = 2^{t(p_{\min,\mathrm{c}}, \delta)+2} - 2.
\end{align*}
Summing this cost over all $\Np$ iterations gives the total query complexity to the $\matr{U}_{\mathrm{ini}}$ oracle and its inverse as
\begin{align*}
    \Np \cdot \left( 2^{t(p_{\min,\mathrm{c}}, \delta)+2} - 2 \right).
\end{align*}
This completes the proof of the query complexity.

\subsection{Proof of Item~\ref{prop: overall_success_probability: final} (Success Probability and Outcome)}
\label{apx:proof_success_probability}

The proof of the success probability and outcome correctness relies on the properties of the phase-testing subroutine from Theorem~\ref{thm:property_of_phase_estimation_apx}. In order to apply this theorem, we first establish two key lemmas regarding the spectral properties of the amplitude amplification operator $\matr{U}_{\mathrm{AA},\mathsf{N}, \vz }$.

\begin{lemma}[Eigen-decomposition for the amplitude-amplification operator]
\label{lem:spectral_properties_U_AA}
Fix $\np \in [\Np]$ and a prefix threshold $\vz \in \{0,1\}^{\np}$. Let
$\matr{U}_{\mathrm{o}} = \matr{U}_{\mathrm{AA},\mathsf{N}, \vz}$
be the amplitude-amplification operator from Definition~\ref{def: the set of operators UAA}. We define the projector onto the marked subspace by
\begin{align*}
  \matr{\Pi}_{\vz}
  \defeq
  \matr{I}_{\mathrm{rem}} \otimes \matr{P}_{\set{Z}_{\ge \vz}},
\end{align*}
and define the corresponding weight of the prepared state by
\begin{align*}
  p_{\vz}
  \defeq
  \left\|
  \matr{\Pi}_{\vz} \cdot \ket{\psi_{\rmini}}
  \right\|_2^2
  \in [0,1].
\end{align*}
We assume that $0 < p_{\vz} < 1$, and define
\[
  \phi_{p_{\vz}} \defeq 2\arcsin\bigl(\sqrt{p_{\vz}}\bigr) \in (0,\pi).
\]
Then $\matr{U}_{\mathrm{o}}$ has two orthonormal eigenvectors
$\ket{\phi_{\rmini}}_{\vz}$ and
$\overline{\ket{\phi_{\rmini}}_{\vz}}$
with eigenvalues $\e^{\pm \imagunit \phi_{p_{\vz}}}$, respectively, and the prepared state satisfies
\begin{align*}
  \ket{\psi_{\rmini}}
  =
  \frac{1}{\sqrt{2}}
  \left(
    \e^{ \imagunit \phi_{p_{\vz}}/2 } \cdot \ket{\phi_{\rmini}}_{\vz}
    +
    \e^{ -\imagunit \phi_{p_{\vz}}/2 } \cdot \overline{\ket{\phi_{\rmini}}_{\vz}}
  \right).
\end{align*}
\end{lemma}
\begin{proof}
By Definition~\ref{def: the set of operators UAA} and~\eqref{eqn: def of UAANz}, it holds that
\begin{align*}
  \matr{U}_{\mathrm{o}}
  &=
  \matr U_{\mathrm{ini}}
  \cdot\bigl(2\ket{\mathbf 0}\bra{\mathbf 0}-\matr I\bigr)_{\mathrm{all}}
  \cdot \matr U_{\mathrm{ini}}^{\Herm}
  \cdot\bigl(\matr I_{\mathrm{rem}}\otimes \matr U_{\set{Z}_{\ge \vz}}\bigr).
\end{align*}
We define
\begin{align*}
  \matr{R}_{\psi}
  &\defeq
  \matr U_{\mathrm{ini}}
  \cdot\bigl(2\ket{\mathbf 0}\bra{\mathbf 0}-\matr I\bigr)_{\mathrm{all}}
  \cdot \matr U_{\mathrm{ini}}^{\Herm},
  \\
  \matr{R}_{\vz}
  &\defeq
  \matr I_{\mathrm{rem}}\otimes \matr U_{\set{Z}_{\ge \vz}}.
\end{align*}
Then $\matr{U}_{\mathrm{o}}=\matr{R}_{\psi}\matr{R}_{\vz}$, where
\begin{align*}
  \matr{R}_{\psi}
  &=
  \matr U_{\mathrm{ini}}
  \cdot\Bigl(2\ket{\mathbf 0}\bra{\mathbf 0}\Bigr)_{\mathrm{all}}
  \cdot \matr U_{\mathrm{ini}}^{\Herm}
  - \matr I
  \overset{(a)}{=}
  2\ket{\psi_{\rmini}}\bra{\psi_{\rmini}} - \matr I,
\end{align*}
and step $(a)$ follows from the state-preparation specification~\eqref{eqn: generated state vector} (absorbing the always-zero clean ancillas into $\ket{\psi_{\rmini}}$) and the unitarity of $\matr U_{\mathrm{ini}}$. Moreover, since $\matr U_{\set{Z}_{\ge \vz}}$ applies a phase of $-1$ exactly to the marked computational-basis states, it holds that
\begin{align}
  \matr{R}_{\vz}
  =
  \matr I - 2\matr{\Pi}_{\vz}.
  \label{eqn: express Rz in terms of Pi z}
\end{align}

Since $0 < p_{\vz} < 1$, we define
\begin{align*}
  \ket{\psi_{\vz}^{\mathrm{good}}}
  &\defeq
  \frac{\matr{\Pi}_{\vz} \cdot \ket{\psi_{\rmini}}}{\sqrt{p_{\vz}}},
  \qquad
  \ket{\psi_{\vz}^{\mathrm{bad}}}
  \defeq
  \frac{(\matr I-\matr{\Pi}_{\vz}) \cdot \ket{\psi_{\rmini}}}{\sqrt{1-p_{\vz}}}.
\end{align*}
These two vectors are normalized and orthogonal because $\matr{\Pi}_{\vz}$ is a projector and $\matr{\Pi}_{\vz}(\matr I-\matr{\Pi}_{\vz})=\matr 0$. Furthermore, it holds that
\begin{align}
  \ket{\psi_{\rmini}}
  &=
  \matr{\Pi}_{\vz} \cdot \ket{\psi_{\rmini}} + (\matr I-\matr{\Pi}_{\vz}) \cdot \ket{\psi_{\rmini}}
  \nonumber\\&=
  \sqrt{p_{\vz}} \cdot \ket{\psi_{\vz}^{\mathrm{good}}} + \sqrt{1-p_{\vz}}\cdot \ket{\psi_{\vz}^{\mathrm{bad}}}
  \nonumber\\&=
  \sin\bigl(\phi_{p_{\vz}}/2\bigr) \cdot \ket{\psi_{\vz}^{\mathrm{good}}}
  +
  \cos\bigl(\phi_{p_{\vz}}/2\bigr) \cdot \ket{\psi_{\vz}^{\mathrm{bad}}},
  \label{eqn:psi_ini_good_bad_decomp}
\end{align}
where the last equality uses $\phi_{p_{\vz}} = 2\arcsin(\sqrt{p_{\vz}})$.

We now analyze the action of $\matr{U}_{\mathrm{o}}$ on the two-dimensional subspace
\[
  \mathcal{K}_{\vz}
  \defeq
  \mathrm{span}\Bigl\{
    \ket{\psi_{\vz}^{\mathrm{bad}}},
    \ket{\psi_{\vz}^{\mathrm{good}}}
  \Bigr\}.
\]
By~\eqref{eqn: express Rz in terms of Pi z}, it holds that
\begin{align*}
  \matr{R}_{\vz} \cdot \ket{\psi_{\vz}^{\mathrm{bad}}}
  &= \ket{\psi_{\vz}^{\mathrm{bad}}},
  \qquad
  \matr{R}_{\vz} \cdot \ket{\psi_{\vz}^{\mathrm{good}}}
  = -\ket{\psi_{\vz}^{\mathrm{good}}}.
\end{align*}
Using $\matr{R}_{\psi}=2\ket{\psi_{\rmini}}\bra{\psi_{\rmini}}-\matr I$ together with~\eqref{eqn:psi_ini_good_bad_decomp}, we obtain
\begin{align*}
  \matr{R}_{\psi} \cdot \ket{\psi_{\vz}^{\mathrm{bad}}}
  &\overset{(a)}{=}
  \Bigl(2\ket{\psi_{\rmini}}\bra{\psi_{\rmini}}-\matr I\Bigr)
  \cdot\ket{\psi_{\vz}^{\mathrm{bad}}}
  \\&\overset{(b)}{=}
  2\ket{\psi_{\rmini}}\braket{\psi_{\rmini}}{\psi_{\vz}^{\mathrm{bad}}} - \ket{\psi_{\vz}^{\mathrm{bad}}}
  \\&\overset{(c)}{=}
  2 \cos(\phi_{p_{\vz}}/2) \cdot \ket{\psi_{\rmini}} - \ket{\psi_{\vz}^{\mathrm{bad}}}
  \\&\overset{(d)}{=}
  \Bigl(2\cos^2(\phi_{p_{\vz}}/2)-1\Bigr) \cdot \ket{\psi_{\vz}^{\mathrm{bad}}}
  +2\sin(\phi_{p_{\vz}}/2)  \cdot \cos(\phi_{p_{\vz}}/2) \cdot \ket{\psi_{\vz}^{\mathrm{good}}}
  \\&\overset{(e)}{=}
  \cos\bigl(\phi_{p_{\vz}}\bigr)\cdot \ket{\psi_{\vz}^{\mathrm{bad}}}
  +\sin\bigl(\phi_{p_{\vz}}\bigr) \cdot \ket{\psi_{\vz}^{\mathrm{good}}},
\end{align*}
where step $(a)$ follows from the definition of $\matr{R}_{\psi}$, where step $(b)$ follows from linearity, where step $(c)$ follows from~\eqref{eqn:psi_ini_good_bad_decomp} and orthonormality of $\ket{\psi_{\vz}^{\mathrm{bad}}}$ and $\ket{\psi_{\vz}^{\mathrm{good}}}$, where step $(d)$ follows from substituting~\eqref{eqn:psi_ini_good_bad_decomp}, and where step $(e)$ follows from the trigonometric identities $2\cos^2(\theta)-1=\cos(2\theta)$ and $2\sin(\theta)\cos(\theta)=\sin(2\theta)$. Similarly, we have
\begin{align*}
  \matr{R}_{\psi}\ket{\psi_{\vz}^{\mathrm{good}}}
  &\overset{(a)}{=}
  2\ket{\psi_{\rmini}}\braket{\psi_{\rmini}}{\psi_{\vz}^{\mathrm{good}}} - \ket{\psi_{\vz}^{\mathrm{good}}}
  \\&\overset{(b)}{=}
  2 \sin(\phi_{p_{\vz}}/2) \cdot \ket{\psi_{\rmini}} - \ket{\psi_{\vz}^{\mathrm{good}}}
  \\&\overset{(c)}{=}
  \sin\bigl(\phi_{p_{\vz}}\bigr) \cdot \ket{\psi_{\vz}^{\mathrm{bad}}}
  -\cos\bigl(\phi_{p_{\vz}}\bigr) \cdot \ket{\psi_{\vz}^{\mathrm{good}}},
\end{align*}
where step $(a)$ follows from the same expansion, where step $(b)$ follows from~\eqref{eqn:psi_ini_good_bad_decomp}, and where step $(c)$ follows from substitution and the same trigonometric identities. Therefore, it holds that
\begin{align}
  \matr{U}_{\mathrm{o}}\ket{\psi_{\vz}^{\mathrm{bad}}}
  &= \cos\bigl(\phi_{p_{\vz}}\bigr) \cdot \ket{\psi_{\vz}^{\mathrm{bad}}}
  +\sin\bigl(\phi_{p_{\vz}}\bigr) \cdot \ket{\psi_{\vz}^{\mathrm{good}}},
  \label{eqn:Uo_action_bad}
  \\
  \matr{U}_{\mathrm{o}} \cdot \ket{\psi_{\vz}^{\mathrm{good}}}
  &= -\sin\bigl(\phi_{p_{\vz}}\bigr) \cdot \ket{\psi_{\vz}^{\mathrm{bad}}}
  +\cos\bigl(\phi_{p_{\vz}}\bigr) \cdot \ket{\psi_{\vz}^{\mathrm{good}}}.
  \label{eqn:Uo_action_good}
\end{align}
Hence, $\mathcal{K}_{\vz}$ is invariant under $\matr{U}_{\mathrm{o}}$, and in the ordered orthonormal basis
\[
  \Bigl(
    \ket{\psi_{\vz}^{\mathrm{bad}}},
    \ket{\psi_{\vz}^{\mathrm{good}}}
  \Bigr),
\]
the restriction of $\matr{U}_{\mathrm{o}}$ to $\mathcal{K}_{\vz}$ is the rotation matrix
\[
  \begin{pmatrix}
    \cos\bigl(\phi_{p_{\vz}}\bigr) & -\sin\bigl(\phi_{p_{\vz}}\bigr) \\
    \sin\bigl(\phi_{p_{\vz}}\bigr) & \cos\bigl(\phi_{p_{\vz}}\bigr)
  \end{pmatrix}.
\]

We define
\begin{align*}
  \ket{\phi_{\rmini}}_{\vz}
  &\defeq
  \frac{1}{\sqrt{2}} \cdot \Bigl(\ket{\psi_{\vz}^{\mathrm{bad}}} - \imagunit \cdot \ket{\psi_{\vz}^{\mathrm{good}}}\Bigr),
  \\
  \overline{\ket{\phi_{\rmini}}_{\vz}}
  &\defeq
  \frac{1}{\sqrt{2}} \cdot \Bigl(\ket{\psi_{\vz}^{\mathrm{bad}}} + \imagunit \cdot \ket{\psi_{\vz}^{\mathrm{good}}}\Bigr).
\end{align*}
Since $\ket{\psi_{\vz}^{\mathrm{bad}}}$ and $\ket{\psi_{\vz}^{\mathrm{good}}}$ are orthonormal, the vectors $\ket{\phi_{\rmini}}_{\vz}$ and $\overline{\ket{\phi_{\rmini}}_{\vz}}$ are also orthonormal. Using~\eqref{eqn:Uo_action_bad}--\eqref{eqn:Uo_action_good}, we compute
\begin{align*}
  \matr{U}_{\mathrm{o}}\ket{\phi_{\rmini}}_{\vz}
  &\overset{(a)}{=}
  \frac{1}{\sqrt{2}} \cdot \Bigl(\matr{U}_{\mathrm{o}}\ket{\psi_{\vz}^{\mathrm{bad}}} - \imagunit \cdot \matr{U}_{\mathrm{o}}\ket{\psi_{\vz}^{\mathrm{good}}}\Bigr)
  \\&\overset{(b)}{=}
  \frac{1}{\sqrt{2}} \cdot \Bigl(
  \bigl(\cos(\phi_{p_{\vz}})+\imagunit \cdot \sin(\phi_{p_{\vz}})
  \bigr) \cdot \ket{\psi_{\vz}^{\mathrm{bad}}}
  - \imagunit \cdot \bigl(\cos(\phi_{p_{\vz}})+\imagunit\sin(\phi_{p_{\vz}})\bigr) \cdot \ket{\psi_{\vz}^{\mathrm{good}}}
  \Bigr)
  \\&\overset{(c)}{=}
  \e^{\imagunit \phi_{p_{\vz}}}
  \cdot
  \frac{1}{\sqrt{2}} \cdot \Bigl(\ket{\psi_{\vz}^{\mathrm{bad}}} - \imagunit \ket{\psi_{\vz}^{\mathrm{good}}}\Bigr)
  \\&=
  \e^{\imagunit \phi_{p_{\vz}}} \cdot \ket{\phi_{\rmini}}_{\vz},
\end{align*}
where step $(a)$ follows from linearity, where step $(b)$ follows from substituting~\eqref{eqn:Uo_action_bad}--\eqref{eqn:Uo_action_good}, and where step $(c)$ follows from $\cos\alpha+\imagunit\sin\alpha=\e^{\imagunit\alpha}$ with $\alpha=\phi_{p_{\vz}}$. Similarly, it holds that
\begin{align*}
  \matr{U}_{\mathrm{o}}\overline{\ket{\phi_{\rmini}}_{\vz}}
  &\overset{(a)}{=}
  \frac{1}{\sqrt{2}} \cdot \Bigl(\matr{U}_{\mathrm{o}}\ket{\psi_{\vz}^{\mathrm{bad}}} + \imagunit \cdot \matr{U}_{\mathrm{o}}\ket{\psi_{\vz}^{\mathrm{good}}}\Bigr)
  \\&\overset{(b)}{=}
  \frac{1}{\sqrt{2}} \cdot \Bigl(
  \bigl(\cos(\phi_{p_{\vz}})-\imagunit \cdot \sin(\phi_{p_{\vz}})
  \bigr) \cdot \ket{\psi_{\vz}^{\mathrm{bad}}}
  + \imagunit \cdot \bigl(\cos(\phi_{p_{\vz}})-\imagunit\sin(\phi_{p_{\vz}})\bigr) \cdot \ket{\psi_{\vz}^{\mathrm{good}}}
  \Bigr)
  \\&\overset{(c)}{=}
  \e^{-\imagunit \phi_{p_{\vz}}}
  \cdot
  \frac{1}{\sqrt{2}} \cdot \Bigl(\ket{\psi_{\vz}^{\mathrm{bad}}} + \imagunit \ket{\psi_{\vz}^{\mathrm{good}}}\Bigr)
  \\&=
  \e^{-\imagunit \phi_{p_{\vz}}} \cdot \overline{\ket{\phi_{\rmini}}_{\vz}},
\end{align*}
where step $(a)$ follows from linearity, where step $(b)$ follows from substituting~\eqref{eqn:Uo_action_bad}--\eqref{eqn:Uo_action_good}, and where step $(c)$ follows from $\cos\alpha-\imagunit\sin\alpha=\e^{-\imagunit\alpha}$ with $\alpha=\phi_{p_{\vz}}$.
Finally, it holds that
\begin{align*}
  \frac{1}{\sqrt{2}}&
  \left(
    \e^{\imagunit\phi_{p_{\vz}}/2} \cdot \ket{\phi_{\rmini}}_{\vz}
    +
    \e^{-\imagunit\phi_{p_{\vz}}/2} \cdot \overline{\ket{\phi_{\rmini}}_{\vz}}
  \right)
  \\&\overset{(a)}{=}
  \frac{1}{2}
  \Biggl(
    \e^{\imagunit\phi_{p_{\vz}}/2}\cdot \Bigl(\ket{\psi_{\vz}^{\mathrm{bad}}}-\imagunit\ket{\psi_{\vz}^{\mathrm{good}}}\Bigr)
    +
    \e^{-\imagunit\phi_{p_{\vz}}/2} \cdot \Bigl(\ket{\psi_{\vz}^{\mathrm{bad}}}+\imagunit\ket{\psi_{\vz}^{\mathrm{good}}}\Bigr)
  \Biggr)
  \\&\overset{(b)}{=}
  \cos(\phi_{p_{\vz}}/2) \cdot \ket{\psi_{\vz}^{\mathrm{bad}}}
  +\sin(\phi_{p_{\vz}}/2) \cdot \ket{\psi_{\vz}^{\mathrm{good}}}
  \\&\overset{(c)}{=}
  \ket{\psi_{\rmini}},
\end{align*}
where step $(a)$ follows from the definitions of $\ket{\phi_{\rmini}}_{\vz}$ and $\overline{\ket{\phi_{\rmini}}_{\vz}}$, where step $(b)$ uses $\e^{\imagunit\theta}+\e^{-\imagunit\theta}=2\cos\theta$ and $\e^{-\imagunit\theta}-\e^{\imagunit\theta}=-2\imagunit\sin\theta$, and where step $(c)$ follows from~\eqref{eqn:psi_ini_good_bad_decomp}. This completes the proof.
\end{proof}

Lemma~\ref{lem:spectral_properties_U_AA} shows that, for each prefix threshold $\vz \in \{0,1\}^{\np}$ with $0 < p_{\vz} < 1$, the phase-testing instance arising in Theorem~\ref{thm: overall_algorithm_properties} fits the abstract setup of Definition~\ref{def:setup_of_phase_estimation_apx} under the identifications
\begin{align*}
    \matr{U}_{\mathrm{pe,0}} = \matr{U}_{\mathrm{ini}},\qquad
    \matr{U}_{\mathrm{o}} = \matr{U}_{\mathrm{AA},\mathsf{N}, \vz},\qquad
    p_{\min} = p_{\min,\rmc}.
\end{align*}
Moreover, Lemma~\ref{lem:lower_bound_p_z} below shows that whenever $p_{\vz} > 0$, we have $p_{\vz} \geq p_{\min,\rmc}$, so the lower-bound requirement in Definition~\ref{def:setup_of_phase_estimation_apx} is satisfied in the nondegenerate case. Therefore, Theorem~\ref{thm:property_of_phase_estimation_apx} can be applied to each phase-testing call used in the proof of Theorem~\ref{thm: overall_algorithm_properties}. The case $p_{\vz}=0$ is covered by Item~\ref{thm:property_of_phase_estimation_apx:pefect_reverse} of Theorem~\ref{thm:property_of_phase_estimation_apx}.

\begin{lemma}
    \label{lem:lower_bound_p_z}
    If $p_{\vz} > 0$ for some $\vz \in \{0,1\}^{\np}$, then $p_{\vz} \ge p_{\min,\rmc}$.
\end{lemma}
\begin{proof}
    The condition $p_{\vz} > 0$ implies that the state $\ket{\psi_{\rmini}}$ has a non-zero projection onto the subspace marked by $\matr{P}_{\set{Z}_{\ge\vz}}$. This means there exists at least one configuration $\vz_{\rmp,[\Nsfn]}' \in \set{Z}_{\ge\vz}$ such that the probability of measuring this configuration in the registers $(\vqpn{\nsfn})_{\nsfn}$ is strictly positive. Let $\mathbf{z}_{\max,\rmp,[\Nsfn]}$ be the configuration defined as
    \begin{align}
      \begin{aligned}
       \mathbf{z}_{\max,\rmp,[\Nsfn]}
       \in &\arg \max\limits_{ \mathbf{z}_{\rmp,[\Nsfn]} } \sum_{\nsfn \in [\Nsfn]}
       \zdec( \mathbf{z}_{\rmp,\nsfn} ) \\
       &\qquad  \  \mathrm{s.t.} \quad
       \bra{\psi_{\rmini}}
        \Bigl(
         \matr{I}_{\mathrm{rem}}
         \otimes
         \ket{ \mathbf{z}_{\rmp,[\Nsfn]} }
         \bra{ \mathbf{z}_{\rmp,[\Nsfn]} }
        \Bigr)
       \ket{\psi_{\rmini}}
       > 0 \\
       &\qquad \qquad \qquad \mathbf{z}_{\rmp,[\Nsfn]}= (\mathbf{z}_{\rmp,\nsfn})_{\nsfn} \in \prod_{\nsfn \in [\Nsfn]} \{0,1\}^{\Np}
      \end{aligned}\quad .\label{eqn: def of z p nG xB}
    \end{align}
    By definition, $\mathbf{z}_{\max,\rmp,[\Nsfn]}$ is the configuration with the highest value of $\sum_{\nsfn} \zdec(\mathbf{z}_{\rmp,\nsfn})$ among all configurations that have a non-zero measurement probability. If $p_{\vz} > 0$, it must be that $\sum_{\nsfn} \zdec(\mathbf{z}_{\max,\rmp,[\Nsfn]}) \ge \sum_{\nsfn} \zdec(\vz_{\rmp,[\Nsfn]}')$ for some $\vz_{\rmp,[\Nsfn]}' \in \set{Z}_{\ge\vz}$. Since $\vz_{\rmp,[\Nsfn]}' \in \set{Z}_{\ge\vz}$, we have $\sum_{\nsfn} \zdec(\vz_{\rmp,[\Nsfn]}') \ge \zdec(\vz)$. Thus, we have $\sum_{\nsfn} \zdec(\mathbf{z}_{\max,\rmp,[\Nsfn]}) \ge \zdec(\vz)$, which implies that $\mathbf{z}_{\max,\rmp,[\Nsfn]} \in \set{Z}_{\ge\vz}$. From~\eqref{eqn: property of psi vect with positive magnitude: final}, the probability of measuring $\mathbf{z}_{\max,\rmp,[\Nsfn]}$ is at least $p_{\min,\rmc}$. Since $p_{\vz}$ is the total probability of measuring any state in $\set{Z}_{\ge\vz}$, it must be at least as large as the probability of measuring the single state $\mathbf{z}_{\max,\rmp,[\Nsfn]}$. Therefore, we have $p_{\vz} \ge p_{\min,\rmc}$.
\end{proof}

We now prove Item~\ref{prop: overall_success_probability: final} of Theorem~\ref{thm: overall_algorithm_properties}. The proof proceeds by induction on the number of iterations, $\np$, to demonstrate that the iterative procedure results in a quantum state that has a large amplitude associated with an ideal state encoding the correct solution. This ideal state after $k$ iterations is defined as
\begin{align}
    \ket{\psi_{\mathrm{ideal},k}} \defeq \ket{\bm{0}}_{\vqstn{\mathrm{c}}, \vqB, \vqloc, \vqaux} \ket{ \mathbf{z}_{\max,\rmp,\rmc}(1:k) }_{\vqpn{\rmc}(1:k)},
    \label{eqn:ideal_state_revised}
\end{align}
where $k \in [\Np]$.
Here, the state $\ket{\bm{0}}$ signifies that all registers other than the result register $\vqpn{\rmc}$ are returned to their initial all-zero state, and the state $\ket{ \mathbf{z}_{\max,\rmp,\rmc}(1:k) }$ represents the first $k$ bits of the best $\Np$-bit lower approximation of the target value, as defined in expression~\eqref{eqn: def_z_max_p_c_thm_main_overall}.

\textbf{Inductive Hypothesis $P(\np)$}
Let $P(\np)$ be the inductive hypothesis for an iteration $\np \in [\Np]$. After $\np$ iterations of the algorithm, the system is in a state $\ket{\Psi^{(\np)}}$. The hypothesis $P(\np)$ asserts that the squared magnitude of the projection of this state onto the ideal state from expression~\eqref{eqn:ideal_state_revised} is lower-bounded as follows:
\begin{align*}
    \left| \braket{\psi_{\mathrm{ideal},\np}}{\Psi^{(\np)}} \right|^2 \geq (1-\delta)^{2\np}.
\end{align*}
This inequality states that the probability of a measurement on the relevant registers yielding the outcome corresponding to the ideal state $\ket{\psi_{\mathrm{ideal},\np}}$ is at least $(1-\delta)^{2\np}$.

\textbf{Base Case ($\np=1$)}
For the base case, we establish that the hypothesis $P(1)$ holds. This step corresponds to the algorithm's first iteration, which determines the most significant bit of the solution, $\mathbf{z}_{\max,\rmp,\rmc}(1)$. This is equivalent to testing whether the true value $\sum_{\nsfn} \zdec(\vz_{\max,\rmp,\nsfn})$ is greater than or equal to $0.5$. The analysis is divided into two exhaustive cases.

\begin{itemize}
    \item \textbf{Case 1:} Suppose that $\sum_{\nsfn} \zdec(\vz_{\max,\rmp,\nsfn}) \ge 0.5$. This condition implies that the best 1-bit lower approximation is $\mathbf{z}_{\max,\rmp,\rmc}(1)=1$, which corresponds to the set of outcomes $\vz_{\max,\rmp,[\Nsfn]} \in \set{Z}_{\ge (1) }$. By the definition of $\vz_{\max,\rmp,\nsfn}$ in~\eqref{eqn: def of z p nG xB}, the initial state $\ket{\psi_{\rmini}}$ must have a non-zero projection onto this set. Therefore, the probability $p_{(1)}$ is strictly positive,
    \begin{align*}
        p_{(1)} = \left\|
        (\matr{I}_{\mathrm{rem}} \otimes \matr{P}_{\set{Z}_{\ge (1)}}) \ket{ \psi_{\rmini} } \right\|_{2}^{2} > 0.
    \end{align*}
    From Lemma~\ref{lem:lower_bound_p_z}, it follows that this probability is lower-bounded by $p_{(1)} \ge p_{\min,\rmc} > 0$. According to Theorem~\ref{thm:property_of_phase_estimation_apx} and Lemma~\ref{lem:spectral_properties_U_AA}, the resulting system state $\ket{\Psi^{(1)}}$ satisfies
    \begin{align*}
        &\Bigl| \underbrace{\bigl( \bra{\bm{0}}_{\vqstn{\mathrm{c}}, \vqB, \vqloc, \vqaux} 
        \bra{1}_{\vqpn{\rmc}(1)} \bigr)}_{ = \bra{\psi_{\mathrm{ideal},1}}  } \ket{\Psi^{(1)}} \Bigr|^2 = 
        \left| \braket{\psi_{\mathrm{ideal},1}}{\Psi^{(1)}} \right|^2 \geq (1-\delta)^2.
    \end{align*}
    This satisfies the condition for $P(1)$, since $(1-\delta)^2 \ge (1-\delta)^{2 \cdot 1}$.

    \item \textbf{Case 2:} Suppose that $\sum_{\nsfn} \zdec(\vz_{\max,\rmp,\nsfn}) < 0.5$. In this case, the best 1-bit lower approximation is $\mathbf{z}_{\max,\rmp,\rmc}(1)=0$. This condition implies that no component of the state $\ket{\psi_{\rmini}}$ corresponds to an outcome in $\set{Z}_{\ge (1) }$. Therefore, the probability $p_{(1)}$ is exactly zero,
    \begin{align*}
        p_{(1)} = \left\|
        (\matr{I}_{\mathrm{rem}} \otimes \matr{P}_{\set{Z}_{\ge (1)}}) \ket{ \psi_{\rmini} } \right\|_{2}^{2} = 0.
    \end{align*}
    According to Theorem~\ref{thm:property_of_phase_estimation_apx} and Lemma~\ref{lem:spectral_properties_U_AA}, when the phase is zero, the system state $\ket{\Psi^{(1)}}$ satisfies
    \begin{align*}
        &\Bigl| \underbrace{\bigl( \bra{\bm{0}}_{\vqstn{\mathrm{c}}, \vqB, \vqloc, \vqaux} 
        \bra{0}_{\vqpn{\rmc}(1)} \bigr)}_{ = \bra{\psi_{\mathrm{ideal},1}}  } \ket{\Psi^{(1)}} \Bigr|^2 =
        \left| \braket{\psi_{\mathrm{ideal},1}}{\Psi^{(1)}} \right|^2 = 1.
    \end{align*}
    This also satisfies the condition for $P(1)$, since $1 \ge (1-\delta)^2$.
\end{itemize}
In both cases, the hypothesis $P(1)$ holds.

\textbf{Inductive Step}
Assume that the inductive hypothesis $P(k)$ holds for an arbitrary integer $k$ such that $1 \le k < \Np$. Our objective is to prove that $P(k+1)$ also holds.

By the inductive hypothesis $P(k)$,
\[
    \left| \braket{\psi_{\mathrm{ideal},k}}{\Psi^{(k)}} \right|
    \geq
    (1-\delta)^k.
\]
Equivalently, the state $\ket{\Psi^{(k)}}$ contains a clean component of norm at least $(1-\delta)^k$ in which all registers other than $\vqpn{\rmc}(1:k)$ are returned to their all-zero states and the register $\vqpn{\rmc}(1:k)$ stores the correct prefix $\ket{\mathbf{z}_{\max,\rmp,\rmc}(1:k)}$. In the $(k+1)$-th iteration, the algorithm applies the same phase-testing subroutine as in the base case, now conditioned on this $k$-bit prefix. Accordingly, the relevant threshold string is obtained by appending the bit $1$ to the correct prefix, namely
\[
    \bigl(\mathbf{z}_{\max,\rmp,\rmc}(1:k),1\bigr).
\]
Therefore, the analysis of the $(k+1)$-th step is the same as in the base case, with the test threshold replaced by $\zdec\bigl(\mathbf{z}_{\max,\rmp,\rmc}(1:k),1\bigr)$.

Let $\mathbf{z}_{\mathrm{prefix}} \defeq \mathbf{z}_{\max,\rmp,\rmc}(1:k)$. The algorithm tests whether the true value is greater than or equal to $\zdec\bigl((\mathbf{z}_{\mathrm{prefix}},1)\bigr)$. This is determined by the value of the probability $p_{(\mathbf{z}_{\mathrm{prefix}},1)}$.
\begin{itemize}
    \item \textbf{Case 1:} Suppose that $\sum_{\nsfn} \zdec(\vz_{\max,\rmp,\nsfn}) \ge \zdec\bigl( (\mathbf{z}_{\mathrm{prefix}},1) \bigr)$. This condition implies that the $(k+1)$-th bit of the solution is $\mathbf{z}_{\max,\rmp,\rmc}(k+1)=1$. Following the same logic as in the base case, this leads to $p_{(\mathbf{z}_{\mathrm{prefix}},1)} > 0$
    and according to Theorem~\ref{thm:property_of_phase_estimation_apx} and Lemma~\ref{lem:spectral_properties_U_AA}, the resulting system state $\ket{\Psi^{(k+1)}}$ satisfies $ \left| \braket{\psi_{\mathrm{ideal},k+1}}{\Psi^{(k+1)}} \right|^2 \geq (1-\delta)^{2(k+1)}. $

    \item \textbf{Case 2:} Suppose that $\sum_{\nsfn} \zdec(\vz_{\max,\rmp,\nsfn}) < \zdec\bigl( (\mathbf{z}_{\mathrm{prefix}},1) \bigr)$. This condition implies that the $(k+1)$-th bit of the solution is $\mathbf{z}_{\max,\rmp,\rmc}(k+1)=0$. Following the same logic as in the base case, this leads to $p_{(\mathbf{z}_{\mathrm{prefix}},1)} = 0$
    and according to Theorem~\ref{thm:property_of_phase_estimation_apx} and Lemma~\ref{lem:spectral_properties_U_AA}, the resulting system state 
    $\ket{\Psi^{(k+1)}}$ satisfies $ \left| \braket{\psi_{\mathrm{ideal},k+1}}{\Psi^{(k+1)}} \right|^2 \geq (1-\delta)^{2k}. $
\end{itemize}
In each iteration, the probability of correctly determining the next bit is at least $(1-\delta)^2$. The total success probability after $(k+1)$ iterations is the product of the success probability after $k$ iterations and the success probability of the $(k+1)$-th step. Therefore, we have
\begin{align*}
    \left| \braket{\psi_{\mathrm{ideal},k+1}}{\Psi^{(k+1)}} \right|^2 &\geq \left| \braket{\psi_{\mathrm{ideal},k}}{\Psi^{(k)}} \right|^2 \cdot (1-\delta)^2 \\
    &\geq (1-\delta)^{2k} \cdot (1-\delta)^2 \\
    &= (1-\delta)^{2(k+1)}.
\end{align*}
This shows that if $P(k)$ is true, then $P(k+1)$ is also true.

\textbf{Conclusion}
We have established the base case $P(1)$ and have shown that the inductive step holds. By the principle of mathematical induction, the hypothesis $P(\np)$ is true for all $\np \in [\Np]$. For $\np = \Np$, this implies that the final state of the algorithm $\ket{\Psi_{\mathrm{final}}} \equiv \ket{\Psi^{(\Np)}}$ satisfies
\begin{align*}
    \left| \braket{\psi_{\mathrm{ideal},\Np}}{\Psi_{\mathrm{final}}} \right|^2 \geq (1-\delta)^{2\Np}.
\end{align*}
This completes the proof of Item~\ref{prop: overall_success_probability: final}.

\subsection{Proof of Item~\ref{prop: overall_approximation_precision} (Approximation Precision)}

The total approximation error, $\gmax - \zdec(\mathbf{z}_{\max,\rmp,\rmc})$, can be decomposed into two distinct components: the state preparation error and the quantization error. We can write the total error as
\begin{align*}
    &\gmax - \zdec(\mathbf{z}_{\max,\rmp,\rmc}) 
    \\&\quad= \left( \gmax - \sum_{\nsfn \in [\Nsfn]} \zdec(\mathbf{z}_{\max,\rmp,\nsfn}) \right) 
    + \left( \sum_{\nsfn \in [\Nsfn]} \zdec(\mathbf{z}_{\max,\rmp,\nsfn}) - \zdec(\mathbf{z}_{\max,\rmp,\rmc}) \right).
\end{align*}
We now bound each of these terms.

\textbf{State Preparation Error:}
The first term, $\gmax - \sum_{\nsfn} \zdec(\mathbf{z}_{\max,\rmp,\nsfn})$, represents the error inherent in the state preparation oracle $\matr{U}_{\mathrm{ini}}$. This error arises because each local maximization algorithm $\matr{U}_{\max}$ (used within the implementation of $\matr{U}_{\mathrm{ini}}$) produces an $\Np$-bit approximation of the local maximum, not the exact value. As established in~\eqref{eqn: bound of gstar and znsfn}, which will be proven in Lemma~\ref{lem: properties_approximation_of_oracle_Uini}, this error is bounded as
\begin{align}
    0 \le \gmax - \sum_{\nsfn \in [\Nsfn]} \zdec(\mathbf{z}_{\max,\rmp,\nsfn}) \le \Nsfn \cdot 2^{-\Np}.
    \label{eqn:bound_state_prep_error_proof}
\end{align}

\textbf{Quantization Error:}
The second term, $\sum_{\nsfn} \zdec(\mathbf{z}_{\max,\rmp,\nsfn}) - \zdec(\mathbf{z}_{\max,\rmp,\rmc})$, is the error introduced by approximating the real-valued sum $\sum_{\nsfn} \zdec(\mathbf{z}_{\max,\rmp,\nsfn})$ with a finite, $\Np$-bit binary fraction. By the definition of $\mathbf{z}_{\max,\rmp,\rmc}$ in expression~\eqref{eqn: def_z_max_p_c_thm_main_overall} as the best $\Np$-bit lower approximation of this sum, this error is strictly non-negative and bounded above by the resolution of the $\Np$-bit representation. Specifically, we have
\begin{align}
    0 \le \sum_{\nsfn \in [\Nsfn]} \zdec(\mathbf{z}_{\max,\rmp,\nsfn}) - \zdec(\mathbf{z}_{\max,\rmp,\rmc}) < 2^{-\Np}.
    \label{eqn:bound_quantization_error_proof}
\end{align}

\textbf{Total Error Bound:}
Combining the bounds from inequalities~\eqref{eqn:bound_state_prep_error_proof} and~\eqref{eqn:bound_quantization_error_proof}, we obtain the bound on the total error:
\begin{align*}
    0 &\le \gmax - \zdec(\mathbf{z}_{\max,\rmp,\rmc}) 
    < (\Nsfn \cdot 2^{-\Np}) + 2^{-\Np} = (\Nsfn + 1) \cdot 2^{-\Np}.
\end{align*}
This completes the proof of the approximation precision.

\section{Transforming Classical Circuits to Unitary Operations}
\label{apx:quantum_is_universal_for_classical}

The simulation involves two main steps:
\begin{enumerate}

  \item \textbf{Converting Classical Logic to Reversible Quantum Circuits:} Any classical computation, including irreversible logical gates, can be transformed into an equivalent reversible circuit, typically by adding ancilla bits to store intermediate results that would otherwise be overwritten. Each gate in this reversible classical circuit can then be replaced by a corresponding universal quantum gate (\eg, Toffoli gates). This construction ensures that the quantum circuit size is polynomially related to the classical circuit size.

  \item \textbf{Simulating Classical Randomness with Quantum Superposition:} Classical randomness can be simulated by introducing auxiliary qubits initialized to $\svzero$ and applying Hadamard gates to create uniform superpositions. These qubits can then control different computational branches, mimicking random choices. More general probability distributions can also be prepared.

\end{enumerate}

This principle is formally captured by the complexity class inclusion $\mathrm{BPP} \subseteq \mathrm{BQP}$~\cite{Bernstein1997}, indicating that problems solvable by bounded-error probabilistic polynomial-time classical algorithms are also solvable by bounded-error quantum polynomial-time algorithms. For a detailed construction, see, \eg,~\cite[Chapter 4.5]{Nielsen2010}.


\section{A Generic Quantum Maximization Algorithm}
\label{apx:generic_maximization_subroutine}

This appendix provides the detailed construction of the generic quantum algorithm for function maximization, which serves as a fundamental building block for the distributed protocol in Section~\ref{sec:implementation_of_U_components}. The algorithm finds an $\Np$-bit approximation of the maximum value of an arbitrary function $F: \{0,1\}^N \to \sR$ by leveraging quantum phase estimation and amplitude amplification in an iterative binary search. (For details of quantum phase estimation and amplitude amplification, see, \eg,~\cite{Nielsen2010}.)

\begin{algorithm}[t!]
\caption{Quantum Algorithm for Maximization (Generic $F$)}
\label{alg:generic_maximization_appendix}
\begin{algorithmic}[1]
\Require State preparation oracle $\matr{U}_{F}$, marking oracles $\matr{U}_{\set{Z}_{F\ge \mathbf{z}}}$, precision bits $\Np$, error parameter $\delta \in (0,1)$, amplitude lower bound $p_{\min,F}>0$.
\Ensure Register $\vqpn{F}$ holds the information about $\mathbf{z}_{\max,\rmp,F}$, the best $\Np$-bit lower approximation of $F^{\max}$.
\State \textbf{Initialize} all registers to $\svzero$: $(\vqfrn{F},\,\vqsrn{F},\,\vqoran{F},\,\vqpn{F})\leftarrow \svzero$.
\For{$\np=1,\ldots,\Np$}
    \State \textbf{Determine bit $\np$ via controlled phase-testing:} Execute the subroutine from expression~\eqref{eqn: controll_phase_estimation_genericF_alg_detail_appendix} on registers $\bigl(\vqfrn{F},\vqsrn{F},\vqoran{F},\vqpn{F}(1:\np)\bigr)$. This operation tests the phase of the amplitude amplification operator $\matr{U}_{\mathrm{AA},F,(\mathbf{z}_{\mathrm{prefix}},1)}$ and writes the binary outcome to the qubit $\qpn{F}(\np)$.
\EndFor
\end{algorithmic}
\end{algorithm}

The high-level procedure is outlined in Algorithm~\ref{alg:generic_maximization_appendix}.
The formal specification of this procedure is as follows.
\begin{definition}[Quantum Algorithm for Maximization (Generic $F$)]
\label{def: quantum_maximization_algorithm_appendix}
Consider the setup in Section~\ref{sec: setup of quantum optimization algorithm}. The algorithm in Algorithm~\ref{alg:generic_maximization_appendix} aims to produce an $\Np$-bit string $\mathbf{z}_{\max,\rmp,F}$ representing the largest $\Np$-bit number whose decoded value does not exceed
\[
    F^{\max} = \max_{\mathbf{y} \in \{0,1\}^{N}} F(\mathbf{y}).
\]
Equivalently,
\begin{align*}
    \begin{aligned}
    \mathbf{z}_{\max,\rmp,F} \defeq \arg &\max\limits_{ \substack{ \mathbf{z}_{\rmp} \in \{0,1\}^{\Np} } } \zdec(\mathbf{z}_{\rmp})
    \\&\qquad \!\!
    \text{s.t. } \zdec(\mathbf{z}_{\rmp}) \le F^{\max}
    \end{aligned} \quad.
\end{align*}
The algorithm uses the registers $\vqfrn{F}$ (with $t(p_{\min,F},\delta)$ qubits), $\vqpn{F}$ (with $\Np$ qubits), $\vqsrn{F}$ (with $N$ qubits), and $\vqoran{F}$ (with $N_F$ ancilla qubits used to implement the state-preparation oracle $\matr{U}_{F}$). All registers are initialized to $\svzero$. The oracle $\matr{U}_{F}$ acts on $(\vqsrn{F},\vqoran{F})$ and satisfies~\eqref{eqn: state preparation property}.

For each iteration $\np \in [\Np]$, the algorithm applies an approximately reversible phase-testing step conditioned on the current prefix stored in $\vqpn{F}(1:\np-1)$. This controlled operation is
\begin{align}
    \sum_{ \mathbf{z}_{\mathrm{prefix}} \in \{0,1\}^{\np-1} }
    \mUpha{\matr{U}_{F}, \matr{U}_{\mathrm{AA},F, (\mathbf{z}_{\mathrm{prefix}}, 1) }}{p_{\min,F}, \delta }
    \otimes
    \ket{\mathbf{z}_{\mathrm{prefix}}} \bra{\mathbf{z}_{\mathrm{prefix}}}_{\vqpn{F}(1:\np-1)},
\label{eqn: controll_phase_estimation_genericF_alg_detail_appendix}
\end{align}
where, for each prefix $\mathbf{z}_{\mathrm{prefix}} \in \{0,1\}^{\np-1}$, the amplitude-amplification operator is
\begin{align*}
\begin{split}
    \matr{U}_{\mathrm{AA},F, (\mathbf{z}_{\mathrm{prefix}}, 1) }
    &\defeq
    \matr{U}_{F}
    \cdot
    \bigl( 2 \svzero \bra{\mathbf{0}} - \matr{I} \bigr)_{\vqsrn{F},\vqoran{F}}
    \cdot
    \matr{U}_{F}^{\Herm}
    \cdot
    \bigl( \matr{U}_{\set{Z}_{ F \geq (\mathbf{z}_{\mathrm{prefix}},1) }} \otimes \matr{I}_{\vqoran{F}} \bigr),
\end{split}
\end{align*}
and the phase-flip oracle $\matr{U}_{\set{Z}_{ F \geq (\mathbf{z}_{\mathrm{prefix}},1) }}$ acts on the search register according to
\begin{align*}
  \matr{U}_{\set{Z}_{ F \geq (\mathbf{z}_{\mathrm{prefix}},1) }} 
  \ket{\vect{y}}
  \defeq
  \begin{cases}
  -\ket{\vect{y}} & 
  \text{if } F(\vy) \geq \zdec(\mathbf{z}_{\mathrm{prefix}},1)\\
  \ket{\vect{y}} & \text{otherwise}
  \end{cases} \quad ,
\end{align*}
where $\zdec(\mathbf{z}_{\mathrm{prefix}},1)$ denotes the decoded value of the $\np$-bit string obtained by appending the bit $1$ to the prefix $\mathbf{z}_{\mathrm{prefix}}$. The subroutine $\mUphano$ (defined in Appendix~\ref{apx:reversible_phase_estimation}) uses $\vqfrn{F}$ as its estimation register, $(\vqsrn{F},\vqoran{F})$ as its target register, and $\qpn{F}(\np)$ as its output qubit.

Repeating this controlled phase test for $\np=1,\ldots,\Np$ yields the unitary operator
\[
    \matr{U}_{\max}(\Np, \delta, p_{\min,F}, \matr{U}_F).
\]
The corresponding final state is defined by
\[
    \ket{\psi_{\mathrm{final},F}}
    \defeq
    \matr{U}_{\max}(\Np, \delta, p_{\min,F}, \matr{U}_F)
    \cdot
    \svzero_{\vqfrn{F},\vqpn{F},\vqsrn{F},\vqoran{F}}.
\]
\end{definition}

\begin{theorem}[\textbf{Properties of Quantum Maximization Algorithm (Generic $F$)}]
 \label{thm: properties_estimated_quantum_algorithm}
 Let $F:\{0,1\}^{N}\to[0,1]$ and consider the algorithm in Definition~\ref{def: quantum_maximization_algorithm_appendix}. Suppose that the state-preparation oracle $\matr{U}_{F}$ satisfies~\eqref{eqn: state preparation property}. Then the following statements hold.
 \begin{enumerate}
  \item\label{prop_F_query_complexity} \textbf{Query complexity.} The total number of queries to the state-preparation oracle $\matr{U}_{F}$ and its inverse $\matr{U}_{F}^{\Herm}$ is
  \begin{align*}
    \Np \cdot \bigl( 2^{t(p_{\min,F}, \delta )+2} - 2 \bigr).
  \end{align*}
  \item\label{prop_F_success_pro_outcome} \textbf{Success probability and output.} Measuring all registers of the final state $\ket{\psi_{\mathrm{final},F}}$ yields the clean output state
  \[
      \svzero_{\vqfrn{F},\vqsrn{F},\vqoran{F}}
      \ket{\mathbf{z}_{\max,\rmp,F}}_{\vqpn{F}}
  \]
  with probability at least $(1-\delta)^{2\Np}$. Equivalently,
  \begin{align*}
    \bigl| 
      &\bra{0}_{\vqfrn{F},\vqsrn{F},\vqoran{F}}
      \bra{ \mathbf{z}_{\max,\rmp,F} }_{\vqpn{F}} 
      \ket{ \psi_{\mathrm{final},F} } 
    \bigr|^{2}
    \geq (1-\delta)^{2\Np}.
  \end{align*}
  If one requires $(1-\delta)^{2\Np} \geq 1-\delta'$ for some $0<\delta'<1$, then by Bernoulli's inequality it suffices to choose $\delta \leq \delta'/(2\Np)$.

  \item\label{prop_F_approx_precision} \textbf{Approximation precision.} The returned approximation $\mathbf{z}_{\max,\rmp,F}$ satisfies
  \[
      0 \leq F^{\max} - \zdec\bigl(\mathbf{z}_{\max,\rmp,F}\bigr) \leq 2^{-\Np}.
  \]
 \end{enumerate}
\end{theorem}
\begin{proof}
The proof follows the same bit-by-bit induction as the proof of Theorem~\ref{thm: overall_algorithm_properties}, specialized to the single-function setting.

For Item~\ref{prop_F_query_complexity}, each of the $\Np$ iterations invokes one instance of the approximately reversible phase-testing subroutine $\mUpha{\cdot}{\cdot}$. By Item~\ref{thm:property_of_phase_estimation_apx:query_complexity} of Theorem~\ref{thm:property_of_phase_estimation_apx}, one such invocation uses $2^{t(p_{\min,F},\delta)+1}-2$ applications of the amplitude-amplification operator and two uses of the initialization oracle $\matr{U}_{F}$ or $\matr{U}_{F}^{\Herm}$. Since each application of the amplitude-amplification operator $\matr{U}_{\mathrm{AA},F,\vz}$ contains one call to $\matr{U}_{F}$ and one call to $\matr{U}_{F}^{\Herm}$, the total number of queries per iteration is
\[
    2\bigl(2^{t(p_{\min,F},\delta)+1}-2\bigr)+2
    =
    2^{t(p_{\min,F},\delta)+2}-2.
\]
Multiplying by $\Np$ gives the claimed query complexity.

For Item~\ref{prop_F_success_pro_outcome}, the argument is identical to the inductive proof used for Theorem~\ref{thm: overall_algorithm_properties}: at iteration $\np$, the algorithm tests whether $F^{\max}$ is at least the threshold associated with the candidate prefix obtained by appending the bit $1$. If the threshold is feasible, then the corresponding marked subspace has weight at least $p_{\min,F}$ by~\eqref{eqn: state preparation property}; otherwise the marked weight is zero. In the first case, Item~\ref{thm:property_of_phase_estimation_apx:appro_reverse} of Theorem~\ref{thm:property_of_phase_estimation_apx} guarantees that the correct branch is produced with probability at least $(1-\delta)^2$, while in the second case Item~\ref{thm:property_of_phase_estimation_apx:pefect_reverse} gives the correct outcome with probability $1$. Repeating this argument over $\Np$ iterations shows that the final clean output state is obtained with probability at least $(1-\delta)^{2\Np}$.

Finally, Item~\ref{prop_F_approx_precision} follows from the construction of $\mathbf{z}_{\max,\rmp,F}$ as the largest $\Np$-bit lower approximation of $F^{\max}$. By definition, $\zdec(\mathbf{z}_{\max,\rmp,F}) \leq F^{\max}$, and the spacing between adjacent $\Np$-bit numbers in $[0,1]$ is $2^{-\Np}$. Therefore, the gap between $F^{\max}$ and its largest $\Np$-bit lower approximation is at most $2^{-\Np}$.
\end{proof}


\section{Proof of Lemma~\ref{lem: properties_approximation_of_oracle_Uini}}
\label{apx: properties_approximation_of_oracle}
The proof is structured to address each of the three claims in Lemma~\ref{lem: properties_approximation_of_oracle_Uini} sequentially.

We will establish the query complexity, the EPR pair consumption, and the state preparation fidelity by leveraging the formal procedures outlined in Algorithm~\ref{alg:u_ini_protocol} and the performance guarantees of the generic maximization algorithm from Theorem~\ref{thm: properties_estimated_quantum_algorithm}.

\subsection{Proof of Item 1 (Query Complexity to Local Oracles)}
The query complexity of the unitary operator $\matr{U}_{\mathrm{ini}}$ is determined by the total number of calls to the fundamental local oracles $\bigl\{ \matr{U}_{g_{\nsfn,\mathbf{x}_{\setVbnsfn}}}, \matr{U}_{g_{\nsfn,\mathbf{x}_{\setVbnsfn}}}^{\Herm} \bigr\}_{\nsfn \in [\Nsfn],\, \mathbf{x}_{\setVbnsfn} \in \set{X}_{\setVbnsfn}}$. The distributed procedure in Algorithm~\ref{alg:u_ini_protocol} reveals that these oracles are invoked exclusively within Stage 2 in Algorithm~\ref{alg:u_ini_protocol}. This step consists of the parallel application of the generic quantum maximization algorithm, $\matr{U}_{\max}$, by each of the $\Nsfn$ worker processors, coherently controlled by the qubits $\bigl(\vqB,(\vqBauxn)_{\nsfn \in [\Nsfn]}\bigr)$ corresponding to the boundary variable configuration $\mathbf{x}_{\setVb}$.

According to Item~\ref{prop_F_query_complexity} of Theorem~\ref{thm: properties_estimated_quantum_algorithm}, a single execution of $\matr{U}_{\max}(\Np, \delta, p_{\min,\nsfn}, \matr{U}_{g_{\nsfn,\mathbf{x}_{\setVbnsfn}}})$ requires $\Np \cdot \bigl(2^{t(p_{\min,\nsfn}, \delta)+2}-2 \bigr)$ queries to its respective local oracle and its inverse. Since these maximization procedures are executed in parallel for each of the $\Nsfn$ subgraphs, the total query complexity for one coherent application of $\matr{U}_{\mathrm{ini}}$ is the sum of the queries made by each worker processor. This sum is given by
\begin{align*}
\sum_{\nsfn \in [\Nsfn]} \Np \cdot \left(2^{t(p_{\min,\nsfn}, \delta)+2}-2\right),
\end{align*}
which validates the expression in expression~\eqref{eqn: Uini_query_complexity_of_oracle_lemma}.

\subsection{Proof of Item 2 (EPR Pair Consumption)}
The consumption of EPR pairs within the $\matr{U}_{\mathrm{ini}}$ protocol is entirely attributable to the quantum teleportation operations required for inter-processor communication. We can account for the total consumption by itemizing the teleportation events as described in Algorithm~\ref{alg:u_ini_protocol}.
\begin{itemize}
    \item \textbf{Distribution Phase:} In Stage 1 in Algorithm~\ref{alg:u_ini_protocol}, the coordinator $\QPUc$ teleports the auxiliary register $\vqBauxn$, which contains $|\setVbnsfn|$ qubits, to each worker processor $\QPU_{\nsfn}$. The total number of EPR pairs consumed during this distribution phase is the sum over all workers, which amounts to $\sum_{\nsfn\in[\Nsfn]} |\setVbnsfn|$.
    \item \textbf{Collection Phase:} In Stage 3 in Algorithm~\ref{alg:u_ini_protocol}, each worker processor $\QPU_{\nsfn}$ teleports its registers $\vqBauxn$ (containing $|\setVbnsfn|$ qubits) and $\vqpn{\nsfn}$ (containing $\Np$ qubits) back to the coordinator $\QPUc$. The total number of EPR pairs consumed during this collection phase is therefore $\sum_{\nsfn\in[\Nsfn]} (|\setVbnsfn| + \Np)$.
\end{itemize}
The total EPR pair consumption for a single application of $\matr{U}_{\mathrm{ini}}$ is the sum of the costs from both phases. This sum is
\begin{align*}
\sum_{\nsfn\in[\Nsfn]} |\setVbnsfn| + \sum_{\nsfn\in[\Nsfn]} (|\setVbnsfn| + \Np) = \sum_{\nsfn \in [\Nsfn]} (2|\setVbnsfn| + \Np),
\end{align*}
thereby confirming the expression for EPR pair consumption.

\subsection{Proof of Item 3 (State Preparation Fidelity and Accuracy)}
We prove separately the overlap bound~\eqref{eqn: property of psi vect with positive magnitude: final} and the accuracy bound~\eqref{eqn: bound of gstar and znsfn}.

\textbf{Proof of \texorpdfstring{\eqref{eqn: property of psi vect with positive magnitude: final}}{} (weight on the optimal outcome).}
Let $\vx_{\setVb}^{\ast}$ be an optimal boundary assignment attaining $\gmax$, and let
\[
  \mathbf{z}_{\max,\rmp,[\Nsfn]}
  =
  \bigl(
    \mathbf{z}_{\max,\rmp,\nsfn}
  \bigr)_{\nsfn\in[\Nsfn]}
\]
be the associated tuple of local $\Np$-bit approximations introduced below~\eqref{eqn: generated state vector}.

After Stage~1 of Algorithm~\ref{alg:u_ini_protocol}, the boundary register is in the uniform superposition
\begin{align*}
  \frac{1}{\sqrt{|\setx_{\setVb}|}}
  \sum_{\vx_{\setVb}\in\setx_{\setVb}}
  \ket{\vx_{\setVb}}_{\vqB}.
\end{align*}
For each worker $\nsfn$ and each boundary assignment $\vx_{\setVb}$, define
\begin{align*}
  \ket{\psi_{\nsfn,\vx_{\setVb}}}
  \defeq
  \matr{U}_{\max}\bigl(
    \Np,\delta,p_{\min,\nsfn},
    \matr{U}_{g_{\nsfn,\vx_{\setVbnsfn}}}
  \bigr)
  \cdot
  \svzero_{\vqstn{\nsfn},\vqpn{\nsfn},\vqsrn{\nsfn},\vqoran{\nsfn}}.
\end{align*}
Since Stage~2 applies these local maximization unitaries coherently and in parallel, and Stage~3 only teleports $(\vqBauxn,\vqpn{\nsfn})_{\nsfn}$ back to the coordinator and then uncomputes $(\vqBauxn)_{\nsfn}$, the state prepared by $\matr{U}_{\mathrm{ini}}$ can be written as
\begin{align}
  \ket{\psi_{\rmini}}
  =
  \frac{1}{\sqrt{|\setx_{\setVb}|}}
  \sum_{\vx_{\setVb}\in\setx_{\setVb}}
  \ket{\vx_{\setVb}}_{\vqB}
  \otimes
  \bigotimes_{\nsfn\in[\Nsfn]}
  \ket{\psi_{\nsfn,\vx_{\setVb}}}.
  \label{eqn:uini_branch_decomposition}
\end{align}
In particular, on the optimal branch $\vx_{\setVb}^{\ast}$, the worker states factorize across $\nsfn$.

For each worker $\nsfn$, we define the corresponding clean success state by
\begin{align*}
  \ket{\chi_{\nsfn}^{\ast}}
  \defeq
  \svzero_{\vqstn{\nsfn},\vqsrn{\nsfn},\vqoran{\nsfn}}
  \ket{\mathbf{z}_{\max,\rmp,\nsfn}}_{\vqpn{\nsfn}}.
\end{align*}
By Item~\ref{prop_F_success_pro_outcome} of Theorem~\ref{thm: properties_estimated_quantum_algorithm}, applied to the local function
$g_{\nsfn,\vx_{\setVbnsfn}^{\ast}}$, we have
\begin{align}
  \left|
    \braket{\chi_{\nsfn}^{\ast}}{\psi_{\nsfn,\vx_{\setVb}^{\ast}}}
  \right|^{2}
  \geq
  (1-\delta)^{2\Np},
  \qquad
  \forall\,\nsfn\in[\Nsfn].
  \label{eqn:uini_local_success_bound}
\end{align}
Therefore, the amplitude of the specific global branch in which the boundary register equals $\ket{\vx_{\setVb}^{\ast}}_{\vqB}$ and every worker is in its clean success state is
\begin{align*}
  &\left|
  \left(
    \bra{\vx_{\setVb}^{\ast}}_{\vqB}
    \otimes
    \bigotimes_{\nsfn\in[\Nsfn]}
    \bra{\chi_{\nsfn}^{\ast}}
  \right)
  \ket{\psi_{\rmini}}
  \right|^{2}
  \overset{(a)}{=}
  \frac{1}{|\setx_{\setVb}|}
  \prod_{\nsfn\in[\Nsfn]}
  \left|
    \braket{\chi_{\nsfn}^{\ast}}{\psi_{\nsfn,\vx_{\setVb}^{\ast}}}
  \right|^{2}
  \overset{(b)}{\geq}
  \frac{(1-\delta)^{2\Np\Nsfn}}{|\setx_{\setVb}|},
\end{align*}
where step $(a)$ follows from~\eqref{eqn:uini_branch_decomposition}, and where step $(b)$ follows from~\eqref{eqn:uini_local_success_bound}. Since the projector
$\matr{I}_{\mathrm{rem}}\otimes\matr{P}_{\{\mathbf{z}_{\max,\rmp,[\Nsfn]}\}}$
retains every component whose value registers equal
$\mathbf{z}_{\max,\rmp,[\Nsfn]}$, its expectation is at least the probability weight of this particular branch. Hence,
\begin{align*}
  \bra{\psi_{\rmini}}
  \bigl(
    \matr{I}_{\mathrm{rem}}
    \otimes
    \matr{P}_{\{\mathbf{z}_{\max,\rmp,[\Nsfn]}\}}
  \bigr)
  \ket{\psi_{\rmini}}
  \geq
  \frac{(1-\delta)^{2\Np\Nsfn}}{|\setx_{\setVb}|}
  =
  p_{\min,\rmc}.
\end{align*}
This proves~\eqref{eqn: property of psi vect with positive magnitude: final}.

\textbf{Proof of \texorpdfstring{\eqref{eqn: bound of gstar and znsfn}}{} (accuracy of the encoded value).}
Let $\vx_{\setVb}^{\ast}$ be as above. Then
\begin{align*}
  \gmax
  =
  \sum_{\nsfn\in[\Nsfn]}
  g_{\nsfn,\vx_{\setVbnsfn}^{\ast}}^{\max}.
\end{align*}
For each $\nsfn\in[\Nsfn]$, Item~\ref{prop_F_approx_precision} of Theorem~\ref{thm: properties_estimated_quantum_algorithm}, applied to the local maximization problem with boundary assignment $\vx_{\setVb}^{\ast}$, yields
\begin{align*}
  0
  \leq
  g_{\nsfn,\vx_{\setVbnsfn}^{\ast}}^{\max}
  -
  \zdec\bigl(\mathbf{z}_{\max,\rmp,\nsfn}\bigr)
  \leq
  2^{-\Np}.
\end{align*}
Summing over $\nsfn\in[\Nsfn]$ gives
\begin{align*}
  0
  &\leq
  \sum_{\nsfn\in[\Nsfn]}
  g_{\nsfn,\vx_{\setVbnsfn}^{\ast}}^{\max}
  -
  \sum_{\nsfn\in[\Nsfn]}
  \zdec\bigl(\mathbf{z}_{\max,\rmp,\nsfn}\bigr)
  =
  \gmax
  -
  \sum_{\nsfn\in[\Nsfn]}
  \zdec\bigl(\mathbf{z}_{\max,\rmp,\nsfn}\bigr)
  \leq
  \Nsfn\cdot 2^{-\Np}.
\end{align*}
This is exactly~\eqref{eqn: bound of gstar and znsfn}. The proof of Item~3 is complete.

\section{Proof of Theorem~\ref{thm: overall_algorithm_performance}}
\label{apx: overall_algorithm_performance}

The proof of Theorem~\ref{thm: overall_algorithm_performance} establishes the overall performance and resource requirements of the distributed quantum algorithm presented in Algorithm~\ref{alg:dist-opt}. The proof is a synthesis of previously established results. It combines the abstract performance guarantees of the main algorithmic structure, as proven in Theorem~\ref{thm: overall_algorithm_properties}, with the concrete resource costs of its distributed components. Specifically, we integrate the properties of the state preparation unitary $\matr{U}_{\mathrm{ini}}$, proven in Lemma~\ref{lem: properties_approximation_of_oracle_Uini}, and the distributed implementations of the reflection and controlled-phase-estimation operators, which are detailed in Appendix~\ref{apx:implementationof_Upe}. Each item of the theorem is proven sequentially below.

\subsection{Item~\ref{prop:total_qubits_count_overall} and Item~\ref{prop:qubits_per_processor_thm_overall_alg_perf}: Total Qubit Count and Allocation per Processor}

The total number of qubits and their allocation to each processor are derived by enumerating all quantum registers defined in Appendix~\ref{apx:registers} and accounting for their physical location throughout the algorithm's execution. The analysis distinguishes between registers that are statically allocated to a processor and those that are transient, being transferred via quantum teleportation during the protocol.

\subsubsection{Qubits Processed by the Coordinator Processor (\texorpdfstring{$\QPUc$}{QPU\_c})}
The coordinator processor manages a set of registers throughout the algorithm. The maximum number of qubits that it needs to simultaneously control, determines its hardware requirement.
\begin{itemize}
    \item \textbf{Statically Allocated Registers:} These registers reside permanently at $\QPUc$.
    \begin{itemize}
        \item Boundary variable register $\vqB$: $|\setVb|$ qubits.
        \item Coordinator processing registers $\vqcen$, which include the phase estimation register $\vqstn{\mathrm{c}}$ and the final result register $\vqpc$: $t(p_{\min,\rmc}, \delta) + \Np$ qubits.
        \item Ancilla qubits for the distributed reflection operator, $q_{\mathrm{CZ},\rmc}$ and $q_{\mathrm{global\_zero}}$: $2$ qubits.
    \end{itemize}
    \item \textbf{Transiently Held Registers:} These registers are teleported to $\QPUc$ from the worker processors at specific stages of the algorithm.
    \begin{itemize}
        \item Local result registers $(\vqpn{\nsfn})_{\nsfn}$, received from workers for implementing the marking oracle $\matr{U}_{\set{Z}_{\geq\mathbf{z}}}$ and for the final uncomputation step of $\matr{U}_{\mathrm{ini}}$: $\Nsfn \cdot \Np$ qubits.
        
        \item Auxiliary registers for distributing control information, $(\vqstn{\mathrm{c},\nsfn})_{\nsfn}$: $\Nsfn \cdot t(p_{\min,\rmc}, \delta)$ qubits.
        
        \item Auxiliary boundary registers $(\vqBauxn)_{\nsfn}$, received from workers during the uncomputation step of $\matr{U}_{\mathrm{ini}}$: $\sum_{\nsfn} |\setVbnsfn|$ qubits.
        \item Ancilla qubits for the distributed reflection, $(q_{\mathrm{CZ},\nsfn})_{\nsfn}$, received from workers to perform the multi-controlled gate: $\Nsfn$ qubits.
        
        \item An ancilla qubit $q_{\mathrm{ancilla}}$ for implementing $\matr{U}_{\set{Z}_{\ge\mathbf{z}}}$ and $\mUpha{\cdot}{\cdot}$ in a distributed manner.
    \end{itemize}
\end{itemize}
The total number of qubits processed by $\QPUc$ is the sum of the sizes of these registers, which yields the expression in Item~\ref{prop:qubits_per_processor_thm_overall_alg_perf} of the theorem.

\subsubsection{Qubits Processed by a Worker Processor (\texorpdfstring{$\QPU_{\nsfn}$}{QPU\_s})}
Each worker processor $\QPU_{\nsfn}$ manages its own set of local and auxiliary registers.
\begin{itemize}
    \item \textbf{Statically Allocated Registers:}
    \begin{itemize}
        \item Local processing registers $\vqlocn$, which include the local phase estimation register $\vqstn{\nsfn}$, the internal variable register $\vqsrn{\nsfn}$, and the local result register $\vqpn{\nsfn}$: $t(p_{\min,\nsfn}, \delta) + |\set{V}_{\nsfn}| + \Np$ qubits.
        \item Local oracle ancilla register $\vqoran{\nsfn}$: $N_{\mathrm{ora},\nsfn}$ qubits.
    \end{itemize}
    \item \textbf{Transiently Held Registers:}
    \begin{itemize}
        \item Auxiliary boundary register $\vqBauxn$, received from $\QPUc$ during the execution of $\matr{U}_{\mathrm{ini}}$: $|\setVbnsfn|$ qubits.
        \item Copies of distributed control information, $(\vqstn{\mathrm{c},\nsfn})$, received from $\QPUc$ during the main phase estimation loop: $t(p_{\min,\rmc}, \delta)$ qubits.
        \item An ancilla qubit for the distributed reflection, $q_{\mathrm{CZ},\nsfn}$: $1$ qubit.
    \end{itemize}
\end{itemize}
Summing the sizes of these registers gives the total qubit count for each worker processor, as stated in Item~\ref{prop:qubits_per_processor_thm_overall_alg_perf}. The total qubit count for the entire network, stated in Item~\ref{prop:total_qubits_count_overall}, is the sum of the coordinator's qubits and the qubits of all $\Nsfn$ workers, taking care not to double-count the transient registers. This corresponds to the sum of all entries in Table~\ref{table: number of qubits}.

\subsection{Item~\ref{prop:overall_query_complexity_main_thm}: Overall Query Complexity}

The overall query complexity is measured in the standard distributed model, where the fundamental cost is the number of calls to the local oracles $\bigl\{ \matr{U}_{g_{\nsfn,\mathbf{x}_{\setVbnsfn}}} \bigr\}_{\nsfn \in [\Nsfn],\, \mathbf{x}_{\setVbnsfn} \in \set{X}_{\setVbnsfn}}$ and their inverse, which encode the problem structure. Other resource costs, such as communication, are analyzed separately. The derivation follows a two-level analysis that reflects the nested structure of the algorithm.

First, we consider the high-level complexity of the main algorithm described in Algorithm~\ref{alg:dist-opt}. According to Item~\ref{prop: overall_query_complexity: final} of Theorem~\ref{thm: overall_algorithm_properties}, the main algorithm's iterative structure requires a total of $\Np \cdot (2^{t(p_{\min,\rmc}, \delta) + 2} - 2)$ applications of the state preparation unitary $\matr{U}_{\mathrm{ini}}$ and its inverse $\matr{U}_{\mathrm{ini}}^{\Herm}$.

Second, we consider the low-level complexity of a single application of $\matr{U}_{\mathrm{ini}}$. As established in Item~1 of Lemma~\ref{lem: properties_approximation_of_oracle_Uini}, each execution of the distributed procedure for $\matr{U}_{\mathrm{ini}}$ requires a total of $\sum_{\nsfn \in [\Nsfn]} \Np \cdot (2^{t(p_{\min,\nsfn}, \delta)+2} - 2)$ queries to the fundamental local oracles and their inverses.

The total query complexity of the entire distributed algorithm is the product of these two quantities. This multiplicative relationship arises because the main loop repeatedly calls the $\matr{U}_{\mathrm{ini}}$ subroutine, and each of these subroutine calls in turn triggers a sequence of local oracle queries. Therefore, the total number of queries to the fundamental local oracles is
\begin{align*}
    C_{\mathrm{distr}}(\graphG, \mathsf{N},\delta,\Np) 
    &= \Np^{2} \cdot \bigl( 2^{ t( p_{\min,\mathrm{c}}, \delta ) + 2 } - 2 \bigr) 
    \cdot \Biggl( \sum_{\nsfn \in [\Nsfn]} \bigl(2^{t(p_{\min,\nsfn}, \delta)+2} - 2 \bigr)\Biggr),
\end{align*}
which confirms the expression in~\eqref{eqn: overall_query_complexity_discussion_revised}.

\subsection{Item~\ref{prop:overall_EPR_consumption_main_thm}: Overall EPR Pair Consumption}

The total number of EPR pairs consumed is the sum of pairs used in all quantum teleportation steps throughout the algorithm, assuming an idealized, single-hop quantum network. We derive this total by systematically accounting for the three distinct sources of EPR pair consumption.

\subsubsection{Contribution from State Preparation and Uncomputation}
The state preparation unitary $\matr{U}_{\mathrm{ini}}$ and its inverse $\matr{U}_{\mathrm{ini}}^{\Herm}$ are called a total of $\Np \cdot (2^{t(p_{\min,\mathrm{c}}, \delta)+2} - 2)$ times during the execution of the main algorithm. According to Item~2 of Lemma~\ref{lem: properties_approximation_of_oracle_Uini}, each call consumes $\sum_{\nsfn \in [\Nsfn]} (2 |\setVbnsfn| + \Np)$ EPR pairs for distributing and collecting the auxiliary boundary registers. The total consumption from this source is therefore
\begin{align*}
    N_{\mathrm{EPR},1} = \Np \cdot \bigl( 2^{ t( p_{\min,\mathrm{c}}, \delta ) + 2 } - 2 \bigr) \cdot \Biggl( \sum_{\nsfn \in [\Nsfn]} (2|\setVbnsfn| + \Np) \Biggr).
\end{align*}

\subsubsection{Contribution from the Distributed Reflection Operator}
The reflection operator $(2\svzero\bra{\mathbf{0}} - \matr{I})_{\mathrm{all}}$ is a component of the amplitude amplification operator $\matr{U}_{\mathrm{AA}}$. The phase-testing subroutine, as detailed in Theorem~\ref{thm:property_of_phase_estimation_apx} in Appendix~\ref{apx:reversible_phase_estimation}, in each of the $\Np$ iterations of the main algorithm makes $2^{t(p_{\min,\mathrm{c}}, \delta)+1} - 2$ calls to the controlled-$\matr{U}_{\mathrm{AA}}$ operator. The distributed implementation of the reflection requires each of the $\Nsfn$ workers to teleport its ancilla qubit $q_{\mathrm{CZ},\nsfn}$ to the coordinator and to receive it back, consuming $2\Nsfn$ EPR pairs per reflection. The total consumption from this source is
\begin{align*}
    N_{\mathrm{EPR},2} = \Np \cdot \bigl(2^{t(p_{\min,\mathrm{c}}, \delta)+1} - 2\bigr) \cdot (2\Nsfn).
\end{align*}

\subsubsection{Contribution from Control Information Distribution}
This cost arises from the implementation of the controlled-unitary operations within the phase-testing subroutine, as detailed in Appendix~\ref{apx:implementationof_Upe}. For each of the $\Np$ bits being determined, the algorithm executes a phase-testing subroutine that uses $t(p_{\min,\rmc}, \delta)$ control qubits from the register $\vqstn{\mathrm{c}}$. For each of these control qubits, its quantum state is distributed to all $\Nsfn$ workers and subsequently returned to uncompute the entanglement. This process requires $2\Nsfn t(p_{\min,\rmc}, \delta)$ teleportations per control qubit. The total consumption from this source is the product of the number of main iterations, the number of control qubits per iteration, and the communication cost per control qubit, which is
\begin{align*}
    N_{\mathrm{EPR},3} = \Np \cdot t(p_{\min,\rmc}, \delta) \cdot (2\Nsfn).
\end{align*}

\subsubsection{Total EPR Pair Consumption}
The total number of EPR pairs consumed by the algorithm is the sum of these three contributions, $N_{\mathrm{EPR}} = N_{\mathrm{EPR},1} + N_{\mathrm{EPR},2} + N_{\mathrm{EPR},3}$. This sum yields the expression given in~\eqref{eqn:total_epr_consumption}. We note that the dominant contribution to the EPR pair consumption arises from the repeated calls to $\matr{U}_{\mathrm{ini}}$ and the distributed reflection operator, as both scale exponentially with the number of phase estimation qubits, $t(p_{\min,\rmc}, \delta)$. This highlights the communication cost associated with the phase-testing subroutines as the primary bottleneck for the algorithm's scalability in terms of network resources.

\section{Proof of Proposition~\ref{prop:generalization_to_arbitrary_topologies}}
\label{apx:generalization_to_arbitrary_topologies}

The construction of the algorithm $A_{\mathrm{dist}}'$ on the network $\mathsf{N}'$ follows directly from the implementation of $A_{\mathrm{dist}}$ on $\mathsf{N}$, with modifications only to the communication primitives. The proof proceeds in three steps.

\textbf{1. Remote Communication via Entanglement Swapping.}
The original algorithm $A_{\mathrm{dist}}$ requires communication between the coordinator $\QPUc$ and each worker $\QPU_{\nsfn}$. In an arbitrary connected network $\mathsf{N}'$, these processors may not be directly connected. However, since $\mathsf{N}'$ is a connected graph, a path of length $k = \mathrm{dist}_{\mathsf{N}'}(\mathrm{c}, \nsfn)$ exists between $\QPUc$ and $\QPU_{\nsfn}$,
where $\mathrm{dist}_{\mathsf{N}'}(\mathrm{c}, \nsfn)$ is defined to be shortest-path distance from the coordinator to the worker $\nsfn$ in the network $ \mathsf{N}' $. 
A single EPR pair can be established between these two remote processors by sequentially performing $k-1$ entanglement swapping operations along this path~\cite{ZelseZfiukowski1993, Bennett1993}.
This process consumes one elementary EPR pair for each of the $k$ links along the path.
Thus, teleporting a single qubit from $\QPUc$ to $\QPU_{\nsfn}$ (or vice versa), which consumes one EPR pair in a direct-link network, consumes $k = \mathrm{dist}_{\mathsf{N}'}(\mathrm{c}, \nsfn)$ EPR pairs in the network $\mathsf{N}'$.

\textbf{2. Adapting the Algorithm and Aggregating Communication Costs.}
The algorithm $A_{\mathrm{dist}}'$ is constructed by replacing every single-hop quantum communication step in $A_{\mathrm{dist}}$ (\eg, the quantum teleportation steps in Algorithm~\ref{alg:u_ini_protocol} and Appendix~\ref{apx:implementationof_Upe}) with the corresponding multi-hop entanglement swapping protocol on $\mathsf{N}'$.
Let $N_{\mathrm{EPR}, i}$ be the number of EPR pairs consumed by the $i$-th communication event in the original algorithm, where communication occurs between $\QPUc$ and a worker $\QPU_{\nsfn(i)}$.
In the adapted algorithm $A_{\mathrm{dist}}'$, this event now consumes $N'_{\mathrm{EPR}, i} = \mathrm{dist}_{\mathsf{N}'}(\mathrm{c}, \nsfn(i)) \cdot N_{\mathrm{EPR}, i}$ EPR pairs.
The total EPR pair consumption for $A_{\mathrm{dist}}'$ is the sum over all such events:
\begin{align*}
    N'_{\mathrm{EPR}} = \sum_{i} N'_{\mathrm{EPR}, i} = \sum_{i} \mathrm{dist}_{\mathsf{N}'}(\mathrm{c}, \nsfn(i)) \cdot N_{\mathrm{EPR}, i}.
\end{align*}
By the definition of the network diameter, $\mathrm{dist}_{\mathsf{N}'}(\mathrm{c}, \nsfn(i)) \le d(\mathsf{N}')$ for all $i$.
Therefore, we can establish the upper bound:
\begin{align}
    N'_{\mathrm{EPR}} &\le \sum_{i} d(\mathsf{N}') \cdot N_{\mathrm{EPR}, i} = d(\mathsf{N}') \cdot \sum_{i} N_{\mathrm{EPR}, i} = d(\mathsf{N}') \cdot N_{\mathrm{EPR}}.\nonumber
\end{align}
This proves the communication cost property.

\textbf{3. Invariance of Computational Metrics.}
The adaptation from $A_{\mathrm{dist}}$ to $A_{\mathrm{dist}}'$ only modifies the implementation of communication between processors.
The sequence and nature of local quantum operations (gates, oracles, measurements) performed within each QPU remain unchanged.
Consequently, the total number of queries to local oracles, the number of qubits required at each processor, the overall success probability, and the final approximation precision are all independent of the network topology and remain identical to those specified in Theorem~\ref{thm: overall_algorithm_performance}.
This completes the proof.

\section{Gate-level Wavefunction Validation: Query Scaling and Topology-Driven Entanglement Cost}
\label{sec:gate_level_validation}

\begin{figure}[t]
\centering
\includegraphics[width=0.6\linewidth]{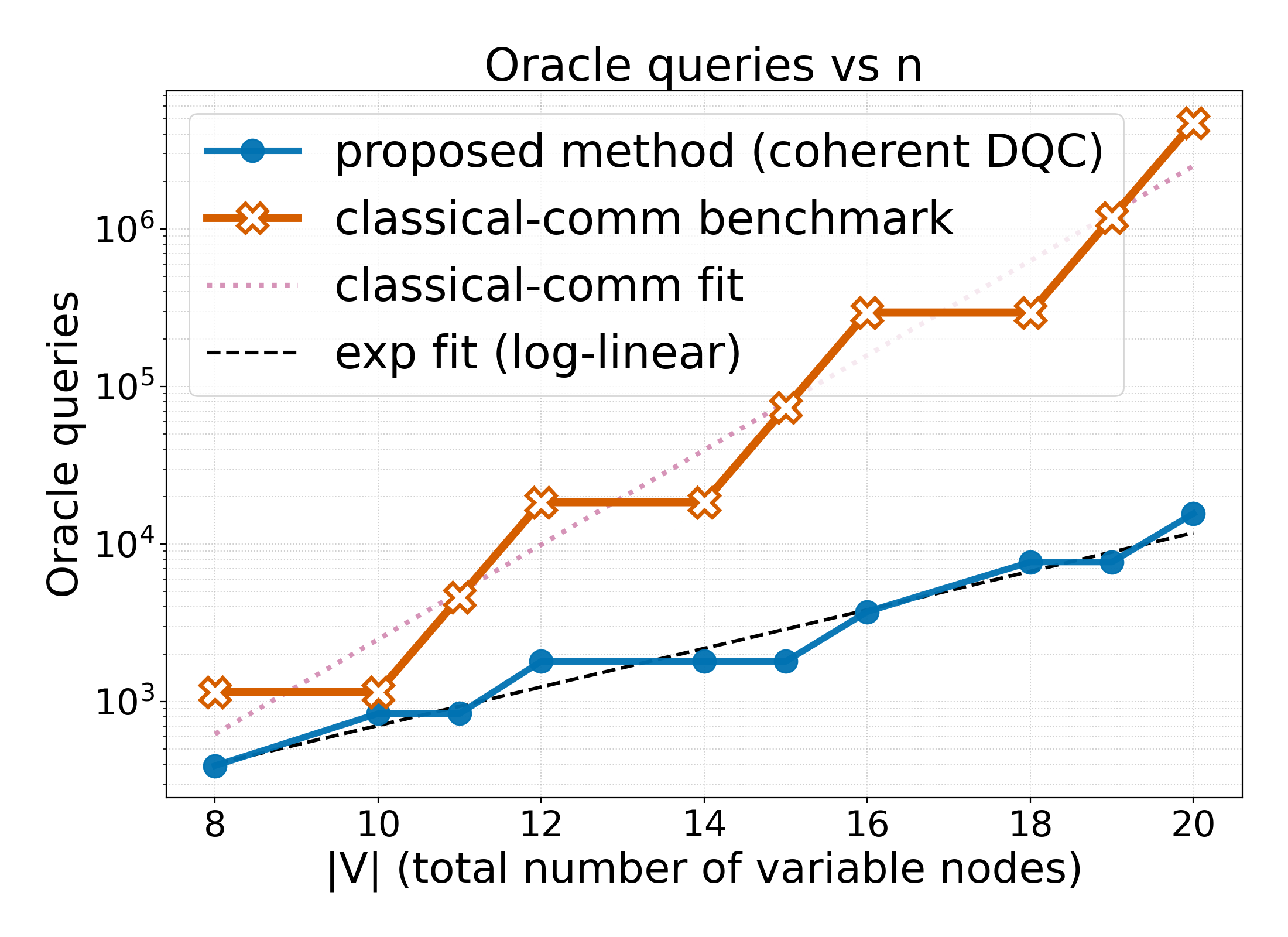}
\caption{Total oracle queries per full execution of the distributed primitive versus the total number of variables $|\setV|$ (gate-level state vector simulation). The $y$-axis is logarithmic. Dashed lines are log-linear fits for both the proposed method and the classical-communication benchmark. Benchmark markers are all actual simulation points (no interpolation).}
\label{fig:queries_vs_n_gate}
\end{figure}
\begin{figure}[t]
\centering
\includegraphics[width=0.6\linewidth]{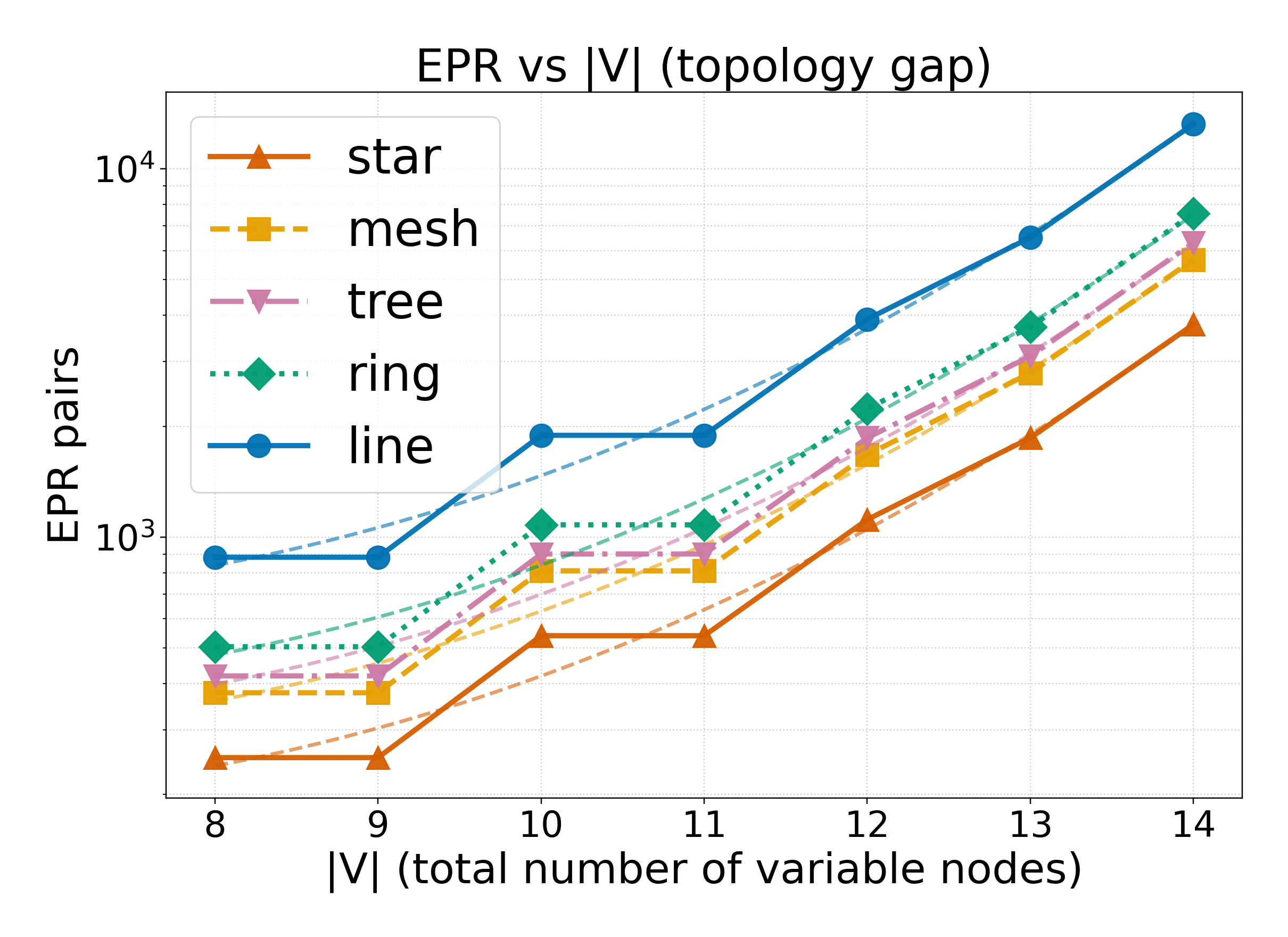}
\caption{Total EPR pairs consumed versus $|\setV|$ across canonical coordinator--worker topologies (gate-level state vector simulation). Here $|\setV|$ increases via the boundary-size schedule under fixed worker settings. Dashed curves are quadratic trend fits in log-space.}
\label{fig:epr_vs_vb_gate}
\end{figure}

\begin{figure}[t]
\centering
\includegraphics[width=0.6\linewidth]{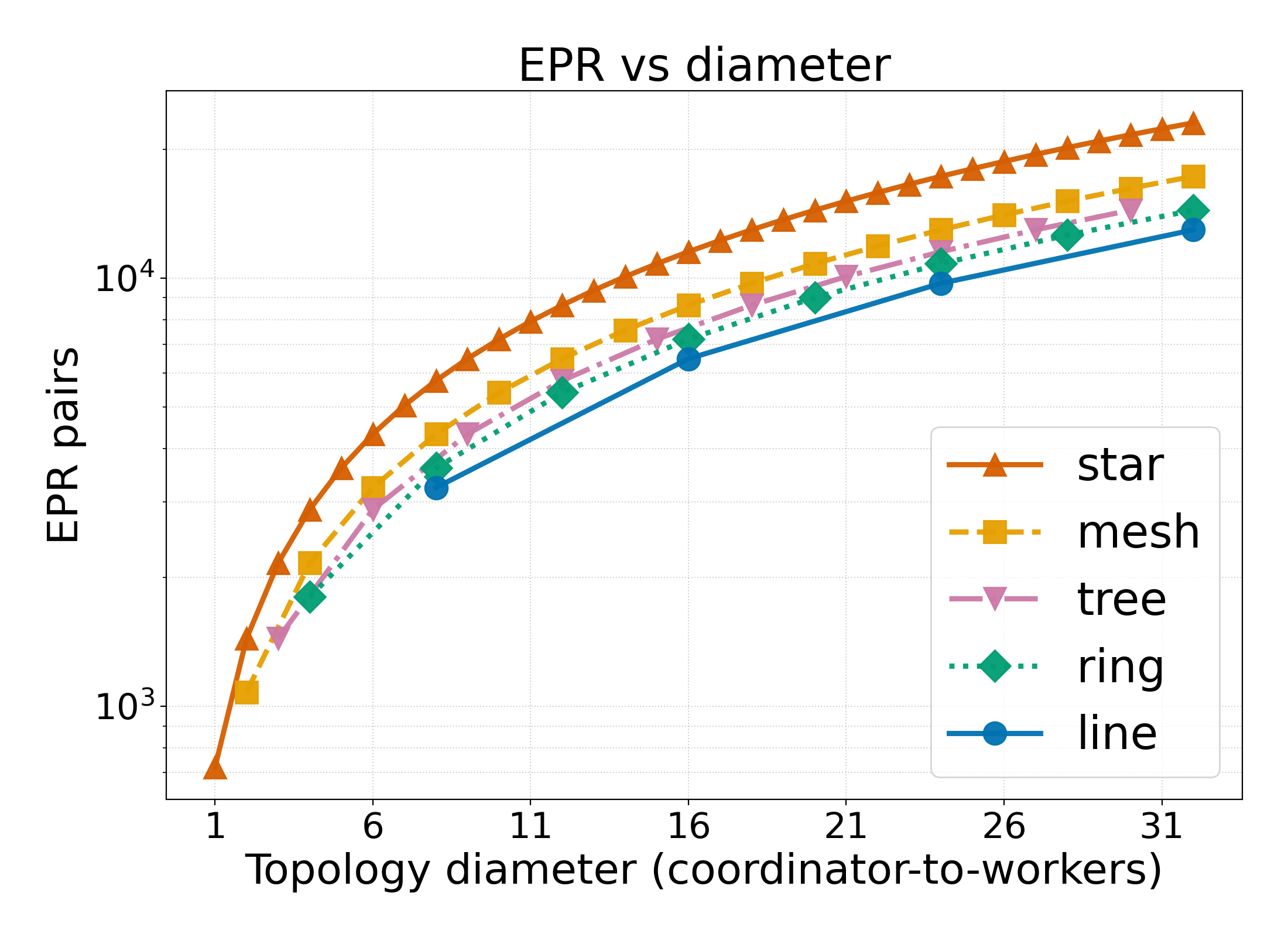}
\caption{Total EPR pairs consumed versus coordinator-to-workers diameter (hops) under topology stretching (gate-level state vector simulation).}
\label{fig:epr_vs_diameter_gate}
\end{figure}

\textbf{Goal.}
These experiments execute the same boundary-conditioned phase-estimation and amplitude-amplification control flow as the
distributed primitive, so they serve as an end-to-end check of the protocol assumptions used in
Theorem~\ref{thm: overall_algorithm_performance} and the topology accounting in
Proposition~\ref{prop:generalization_to_arbitrary_topologies}. Two full-execution counters are considered:
(i) total leaf-oracle evaluations (oracle queries), and (ii) total elementary EPR pairs consumed under the
teleportation-and-swapping hop accounting.

\textbf{Fig.~\ref{fig:queries_vs_n_gate}: query scaling and the coherence gap.}
We sweep $|\setV|\in\{8,9,\ldots,20\}$ (13 simulated points) with two workers ($\Nsfn=2$). The baseline is a
classical-communication benchmark (\textit{cf.}~\cite{Avron2021}): workers run local quantum maximum-finding and exchange
only classical messages, with boundary-level repetition chosen to match the same precision and success targets as our
primitive. The coherent distributed primitive uses fewer oracle queries at every point, with a ratio that ranges from
$1.37\times$ to $302\times$ over the sweep. Both curves show plateau--jump transitions because key protocol parameters
(\eg, the phase-estimation register length) are integer-valued ceilings of logarithmic bounds: they remain constant over
intervals and increase by one at thresholds, producing discrete multiplicative increases in the oracle-query counter.

\textbf{Fig.~\ref{fig:epr_vs_vb_gate}: boundary growth separates topologies.}
We fix $\Nsfn=6$, $|\setV_{\nsfn}|=1$ for all $\nsfn\in[\Nsfn]$, $\Np=1$, $\delta=0.25$, and sweep
$|\setVb|\in\{2,\ldots,8\}$ (hence $|\setV|\in\{8,\ldots,14\}$). The EPR cost increases rapidly with $|\setV|$, and the
topology ordering is stable:
\[
\text{line} > \text{ring} > \text{tree} > \text{mesh} > \text{star}.
\]
The gap is systematic (\eg, line is about $3.5\times$ star across the sweep), which matches
Proposition~\ref{prop:generalization_to_arbitrary_topologies}: at fixed algorithmic workload, EPR consumption is governed
by topology-dependent hop aggregation.

\textbf{Fig.~\ref{fig:epr_vs_diameter_gate}: stretching isolates hop-geometry dependence.}
We fix $|\setVb|=5$, $\Nsfn=8$, $|\setV\setminus\setVb|=1$, $\Np=1$, $\delta=0.25$, and increase hop distances by edge
stretching up to diameter $32$. This keeps the logical computation fixed and changes only routed distances. EPR
consumption increases monotonically with diameter and remains topology-separated, confirming that hop geometry is a
first-order driver. At matched diameter, aggregate coordinator-to-worker hop sums dominate, so diameter alone does not
predict the full ordering; under stretching, star becomes most expensive because every spoke scales with the stretch
(\eg, at diameter $24$: star $17280$ vs.\ line $9720$ EPR pairs).

\textbf{Summary.}
The gate-level data align with the theory: coherent distributed execution reduces oracle-query volume relative to
classical communication only (Theorem~\ref{thm: overall_algorithm_performance}), while entanglement cost is controlled by
boundary size and topology-dependent hop aggregation (Proposition~\ref{prop:generalization_to_arbitrary_topologies}).

\section{Proof of Theorem~\ref{prop: for an oracle solving the maximum how to find an oracle finds the location}}
\label{apx: for an oracle solving the maximum how to find an oracle finds the location}
we assume, without loss of generality, that the variable set is $ \set{V} = \{1, 2, \ldots, |\set{V}|\}$ and that each variable $x_{i}$ is binary, \ie, $x_{i} \in \{0,1\}$, consistent with the factor graph assumptions in this paper (see Section~\ref{sec:factor_graph_main_definition}).

\begin{definition}[Algorithm $ A_{\mathrm{dist,config}}(\graphG, \mathsf{N}, \delta, \Np) $]
\label{def:A_dist_config_alg_revised_V2}
The algorithm to find an optimal configuration $\hat{\vx}$ proceeds as follows:

\textbf{Initial Step:}
Run $ A_{\mathrm{dist}}(\graphG, \mathsf{N}, \delta, \Np) $ on the original factor graph $\graphG$ to obtain an $\Np$-bit approximation $\mathbf{z}_{\max,\rmp,\rmc}^{(\mathrm{ref})}$ of $\gmax$.
Let $\hat{g}_{\mathrm{ref}} = \zdec( \mathbf{z}_{\max,\rmp,\rmc}^{(\mathrm{ref})} )$. Initialize an empty list for the configuration $\hat{\vx}$.
For consistency, we define $\mathbf{z}_{\max,\rmp,\rmc}^{(0)} \defeq \mathbf{z}_{\max,\rmp,\rmc}^{(\mathrm{ref})}$.

\textbf{Iterative Variable Fixing:}
For each variable index $i \in \{1,\ldots, |\set{V}|\}$:
\begin{enumerate}
    \item \textbf{Construct Modified Graphs and Evaluate:}
    For each possible assignment $x' \in \{0,1\}$ for the current variable $x_{i}$:
    \begin{itemize}
        \item Define a new factor graph $ \graphG^{(i,x')} $. This graph is identical to the original graph $\graphG$ except that the alphabets of variables $x_{1}, \ldots, x_{i-1}$ are restricted to their previously determined optimal values $\{\hat{x}_{1}\}, \ldots, \{\hat{x}_{i-1}\}$, and the alphabet of $x_{i}$ is restricted to $\{x'\}$. The alphabets of $x_{i+1}, \ldots, x_{|\set{V}|}$ remain unchanged.
        \item Let $\gmaxn{i,x'}$ be the true maximum of the global function corresponding to the factor graph $\graphG^{(i,x')}$.
        \item Run $ A_{\mathrm{dist}}(\graphG^{(i,x')}, \mathsf{N}, \delta, \Np) $ to obtain an $\Np$-bit approximation $\mathbf{z}_{\max,\rmp,\rmc}^{(i,x')}$ of $\gmaxn{i,x'}$. By Theorem~\ref{thm: overall_algorithm_properties}, with probability at least $(1-\delta)^{2\Np}$, this approximation satisfies:
        \begin{align}
        \bigl| \gmaxn{i,x'} - \zdec( \mathbf{z}_{\max,\rmp,\rmc}^{(i,x')} ) \bigr| \leq \Delta_{\Np},
        \label{eqn: bound_of_error_for_g_hat_and_g_star_iz_proof_revised_V3}
        \end{align}
        where $\Delta_{\Np} \defeq (\Nsfn + 1) \cdot 2^{-\Np}$.
    \end{itemize}
    \item \textbf{Set Current Variable $\hat{x}_{i}$:}
    If
    \begin{align}
    \Bigl| \zdec( \mathbf{z}_{\max,\rmp,\rmc}^{(i,0)} ) - \zdec\bigl( \mathbf{z}_{\max,\rmp,\rmc}^{(i-1)} \bigr) \Bigr| \leq 2 \cdot \Delta_{\Np},
    \label{eqn: the_inequality_of_the_sequence_proof_revised}
    \end{align}
    then set $ \hat{x}_{i} = 0 $. Otherwise, set $ \hat{x}_{i} = 1 $.
    Append $\hat{x}_{i}$ to $\hat{\vx}$ and define $\mathbf{z}_{\max,\rmp,\rmc}^{(i)} \defeq \mathbf{z}_{\max,\rmp,\rmc}^{(i, \hat{x}_i)}$.
\end{enumerate}
\textbf{Output:} The algorithm outputs the configuration $ \hat{\vx} = (\hat{x}_{1}, \ldots, \hat{x}_{|\set{V}|}) $.
\end{definition}

Now we prove the performance guarantee statements.

The algorithm $A_{\mathrm{dist,config}}$ involves one initial call to $A_{\mathrm{dist}}$ and then $2|\set{V}|$ subsequent calls to $A_{\mathrm{dist}}$ during the iterative variable fixing stage (two calls for each of the $|\set{V}|$ variables).
Thus, there are a total of $1 + 2|\set{V}|$ calls to $A_{\mathrm{dist}}$.
The query complexity of $A_{\mathrm{dist}}(\graphG^{(i,x')}, \mathsf{N},\delta,\Np)$ is at most $C_{\mathrm{distr}}(\graphG, \mathsf{N},\delta,\Np)$ because fixing variables does not increase (and typically reduces) the size or complexity of the optimization problem. Thus the total query complexity to the fundamental local oracles is upper bounded by $(2|\set{V}| +1) \cdot C_{\mathrm{distr}}(\graphG, \mathsf{N},\delta,\Np) \in O(2^{|\set{V}|/2})$.

The probability that all these $1 + 2|\set{V}|$ calls to $A_{\mathrm{dist}}$ succeed in providing approximations within the bound $\Delta_{\Np}$ (as specified in inequalities~\eqref{eqn:prop_approx_gmax_revised} for the initial call and~\eqref{eqn: bound_of_error_for_g_hat_and_g_star_iz_proof_revised_V3} for subsequent calls) is at least $ ( 1 - \delta 
)^{2 \Np \cdot (1+2|\set{V}|)} $.

Assuming all calls to $A_{\mathrm{dist}}$ are successful (\ie, their outputs satisfy the respective approximation bounds), we analyze the error accumulation.
Let $\gmaxn{i}$ denote the true maximum value of $g(\vx)$ when variables $x_{1}, \ldots, x_{i}$ are fixed to the determined values $\hat{x}_{1}, \ldots, \hat{x}_{i}$, and the remaining variables $x_{i+1}, \ldots, x_{|\set{V}|}$ are optimized.
Let $\gmaxn{0} \defeq \gmax$. 
From the properties of $A_{\mathrm{dist}}$, we have $|\gmaxn{i} - \zdec(\mathbf{z}_{\max,\rmp,\rmc}^{(i, \hat{x}_i)})| \leq \Delta_{\Np}$ for all $i \geq 1$, and $|\gmaxn{0} - \zdec(\mathbf{z}_{\max,\rmp,\rmc}^{(0)})| \leq \Delta_{\Np}$.
Consider the $i$-th step of variable fixing, where $\hat{x}_i$ is determined.

\textbf{Case 1: $\hat{x}_{i} = 0$.}
This assignment is chosen if the condition in inequality~\eqref{eqn: the_inequality_of_the_sequence_proof_revised} holds: $|\zdec(\mathbf{z}_{\max,\rmp,\rmc}^{(i,0)}) - \zdec(\mathbf{z}_{\max,\rmp,\rmc}^{(i-1)})| \leq 2 \cdot \Delta_{\Np}$.
We bound the difference between the true maxima of consecutive optimally extended partial assignments:
\begin{align*}
& \bigl| \gmaxn{i-1} - \gmaxn{i,0} \bigr| \nonumber \\
&\quad = \Bigl| \bigl( \gmaxn{i-1} - \zdec(\mathbf{z}_{\max,\rmp,\rmc}^{(i-1)}) \bigr) + \bigl(\zdec(\mathbf{z}_{\max,\rmp,\rmc}^{(i-1)}) 
- \zdec(\mathbf{z}_{\max,\rmp,\rmc}^{(i,0)}) \bigr) + \bigl(\zdec(\mathbf{z}_{\max,\rmp,\rmc}^{(i,0)}) - \gmaxn{i,0}\bigr) \Bigr| \nonumber \\
&\quad \overset{(a)}{\leq} \underbrace{\bigl| \gmaxn{i-1} - \zdec(\mathbf{z}_{\max,\rmp,\rmc}^{(i-1)}) \bigr|}_{\leq \Delta_{\Np}} 
+ \underbrace{\bigl| \zdec(\mathbf{z}_{\max,\rmp,\rmc}^{(i-1)}) - \zdec(\mathbf{z}_{\max,\rmp,\rmc}^{(i,0)}) \bigr|}_{\leq 2 \cdot \Delta_{\Np} \text{ (by condition for this case)}} 
+ \underbrace{\bigl| \zdec(\mathbf{z}_{\max,\rmp,\rmc}^{(i,0)}) - \gmaxn{i,0} \bigr|}_{\leq \Delta_{\Np}} \nonumber \\
&\quad \leq \Delta_{\Np} + 2 \cdot \Delta_{\Np} + \Delta_{\Np} \nonumber
\\&\quad= 4 \cdot \Delta_{\Np},
\end{align*}
where step $(a)$ follows from the triangle inequality. The subsequent inequality uses the approximation bounds from $A_{\mathrm{dist}}$ and the condition for choosing $\hat{x}_i=0$.
Since $\hat{x}_{i}=0$ in this case, $\gmaxn{i} = \gmaxn{i,0}$. Thus, $| \gmaxn{i-1} - \gmaxn{i} | \leq 4 \cdot \Delta_{\Np}$.

\textbf{Case 2: $\hat{x}_{i} = 1$.}
This assignment is chosen if $\bigl| \zdec(\mathbf{z}_{\max,\rmp,\rmc}^{(i,0)}) - \zdec(\mathbf{z}_{\max,\rmp,\rmc}^{(i-1)}) \bigr| > 2 \cdot \Delta_{\Np}$.
Since $x_{i}$ is binary, the true maximum $\gmaxn{i-1}$ (optimizing over $x_i, \dots, x_{|\set{V}|}$ given $\hat{x}_1, \dots, \hat{x}_{i-1}$) is equal to either $\gmaxn{i,0}$ (if $x_i=0$ is optimal for the $i$-th variable) or $\gmaxn{i,1}$ (if $x_i=1$ is optimal for the $i$-th variable).
Suppose, for contradiction, that $\gmaxn{i-1} = \gmaxn{i,0}$ (\ie, setting $x_i=0$ would have been optimal). Then:
\begin{align*}
&\bigl| \zdec(\mathbf{z}_{\max,\rmp,\rmc}^{(i,0)}) - \zdec(\mathbf{z}_{\max,\rmp,\rmc}^{(i-1)}) \bigr|
\\&\leq \bigl| \zdec(\mathbf{z}_{\max,\rmp,\rmc}^{(i,0)}) - \gmaxn{i,0} \bigr| 
+ \bigl| \gmaxn{i,0} - \gmaxn{i-1} \bigr|
+ \bigl| \gmaxn{i-1} - \zdec(\mathbf{z}_{\max,\rmp,\rmc}^{(i-1)}) \bigr| \\
&\leq \Delta_{\Np} + 0 + \Delta_{\Np} 
\\&= 2 \cdot \Delta_{\Np}.
\end{align*}
This contradicts the condition for setting $\hat{x}_{i}=1$. Therefore, our supposition is false, implying that $\gmaxn{i-1} = \gmaxn{i,1}$.
In this case, since $\hat{x}_{i}=1$, we have $\gmaxn{i} = \gmaxn{i,1}$.
Thus, $| \gmaxn{i-1} - \gmaxn{i} | = | \gmaxn{i,1} - \gmaxn{i,1} | = 0 \leq 4 \cdot \Delta_{\Np}$.

In both cases (whether $\hat{x}_i=0$ or $\hat{x}_i=1$), we establish that for each step $i$:
\begin{align*}
| \gmaxn{i-1} - \gmaxn{i} | \leq 4 \cdot \Delta_{\Np}.
\end{align*}
The total difference between the true global maximum $\gmax = \gmaxn{0}$ and the value of the function at the constructed configuration $g(\hat{\vx}) = \gmaxn{|\set{V}|}$ can be bounded by a telescoping sum:
\begin{align*}
\Bigl| \gmax - g(\hat{\vx}) \Bigr|
&= \Bigl| \gmaxn{0} - \gmaxn{|\set{V}|} \Bigr| \\
&= \Biggl| \sum_{i = 1}^{|\set{V}|} \bigl( \gmaxn{i-1} - \gmaxn{i} \bigr) \Biggr| \\
&\leq \sum_{i = 1}^{|\set{V}|} \bigl| \gmaxn{i-1} - \gmaxn{i} \bigr| \\
&\leq \sum_{i = 1}^{|\set{V}|} (4 \cdot \Delta_{\Np}) \\
&= 4 \cdot |\set{V}| \cdot \Delta_{\Np} \\
&= 4 \cdot |\set{V}| \cdot (\Nsfn + 1) \cdot 2^{-\Np}.
\end{align*}
This establishes the accuracy bound in inequality~\eqref{eqn:prop_approx_x_accuracy_revised}.

\section{Proof of Theorem~\ref{thm:hierarchical_performance}}
\label{apx:hierarchical_performance}

We use the notation from the Setup and Theorem~\ref{thm:hierarchical_performance}.
In particular, $\mathcal{T}$ denotes the decomposition tree, and each internal node $\graphG^{(\ell)}_{\vk}$ is decomposed by $\set{V}^{(\ell+1)}_{\mathrm{B},\vk}$ into children $\{\graphG^{(\ell+1)}_{\vk,k}\}_{k\in\mathcal{K}_{\ell+1}(\vk)}$.
Each such internal node runs a single-level instance of $A_{\mathrm{dist}}$ (Algorithm~\ref{alg:dist-opt}) on its children, and each leaf runs the local maximization routine $\matr{U}_{\max}$ (Appendix~\ref{apx:generic_maximization_subroutine}).
The proof consists of instantiating the single-level guarantees (Theorems~\ref{thm: overall_algorithm_properties} and~\ref{thm: overall_algorithm_performance}) at each node and composing them along the tree.

\textbf{1. Qubits per processor (Item~\ref{itm:hier_qubits_per_proc}).}
Fix an internal node $\graphG^{(\ell)}_{\vk}$.
When its sub-coordinator executes $A_{\mathrm{dist}}$ on its children, it plays the role of the coordinator in Theorem~\ref{thm: overall_algorithm_performance}, with the substitutions
\[
\Nsfn \leftarrow K_{\ell+1}(\vk),\qquad
\setVb \leftarrow \set{V}^{(\ell+1)}_{\mathrm{B},\vk},\qquad
\setVbnsfnn{k} \leftarrow \set{V}^{(\ell+1)}_{\mathrm{B},\vk,k}.
\]
Applying Item~\ref{prop:qubits_per_processor_thm_overall_alg_perf} of Theorem~\ref{thm: overall_algorithm_performance} and taking the dominant problem-dependent term yields the stated scaling in $\lvert \set{V}^{(\ell+1)}_{\mathrm{B},\vk} \rvert$ (up to factors polynomial in $K_{\ell+1}(\vk)$, $\Np$, and $\log(1/\delta)$).
The leaf bound follows analogously from the worker expression in Item~\ref{prop:qubits_per_processor_thm_overall_alg_perf}.

\textbf{2. Query complexity (Item~\ref{itm:hier_query_complexity}).}
For each node $\graphG^{(\ell)}_{\vk}\in\mathcal{T}$, let $C(\graphG^{(\ell)}_{\vk})$ denote the total number of queries to the \emph{leaf-level} local oracles incurred by one execution of the hierarchical subroutine rooted at $\graphG^{(\ell)}_{\vk}$.

\emph{Base case (leaf).}
If $\ell=L$, then $\graphG^{(L)}_{\vk}$ executes $\matr{U}_{\max}$ on local data.
By Proposition~\ref{prop:generic_maximization_performance}, this contributes
$
C(\graphG^{(L)}_{\vk}) = C_{\mathrm{leaf}}(\graphG^{(L)}_{\vk})
$
with $C_{\mathrm{leaf}}(\cdot)$ as defined in~\eqref{eqn:hier_leaf_query_cost_def}.

\emph{Recursive case (internal node).}
Let $\ell<L$.
If $\ell\notin \set{L}_{\mathrm M}$ (coherent), then the level-$\ell$ sub-coordinator executes $A_{\mathrm{dist}}$ and, by Item~\ref{prop: overall_query_complexity: final} of Theorem~\ref{thm: overall_algorithm_properties}, makes
$\Np\cdot(2^{t(p_{\min,\vk}^{(\ell)},\delta)+2}-2)$ applications of its state-preparation oracle $\matr{U}^{(\ell)}_{\mathrm{ini}}$ (and its inverse).
Each such application invokes \emph{each} child subroutine once, hence
\begin{align}
  C(\graphG^{(\ell)}_{\vk})
  =
  C_{\mathrm{eff}}(\graphG^{(\ell)}_{\vk})
  \cdot
  \sum_{k\in\mathcal{K}_{\ell+1}(\vk)} C(\graphG^{(\ell+1)}_{\vk,k}),
  \label{eqn:hier_query_recursion}
\end{align}
where $C_{\mathrm{eff}}(\graphG^{(\ell)}_{\vk})$ is given by~\eqref{eqn:hier_Ceff_def}.
If $\ell\in \set{L}_{\mathrm M}$ (measured), the coordinator enumerates all boundary configurations, so the number of child invocations is $\lvert \setx_{\set{V}^{(\ell+1)}_{\mathrm{B},\vk}}\rvert$, which again gives~\eqref{eqn:hier_query_recursion}.
Unrolling the recursion from the root to the leaves yields
\[
C_{\mathrm{hier}} = C(\graphG^{(0)})
=
\sum_{\graphG^{(L)}_{\vk}\in\mathcal{T}}
N_{\mathrm{inv}}(\graphG^{(L)}_{\vk})\cdot C_{\mathrm{leaf}}(\graphG^{(L)}_{\vk}),
\]
with invocation counts $N_{\mathrm{inv}}(\cdot)$ defined in~\eqref{eqn:hier_invocation_count_def}.
Finally, the coherent big-$O$ bound~\eqref{eqn:hier_Ceff_coherent_bigO} follows by substituting~\eqref{eqn:pmin_hier_level_ell} into the asymptotic bound in Item~\ref{prop: overall_query_complexity: final} of Theorem~\ref{thm: overall_algorithm_properties}.

\textbf{3. EPR consumption (Item~\ref{itm:hier_epr}).}
Let $E(\graphG^{(\ell)}_{\vk})$ be the total number of EPR pairs consumed by one execution of the hierarchical subroutine rooted at $\graphG^{(\ell)}_{\vk}$.
If $\ell=L$ (leaf), then $E(\graphG^{(L)}_{\vk})=0$ under our accounting model because the computation is local.
If $\ell<L$ (internal node), then one execution consists of:
(i) the level-$\ell$ single-level routine $A_{\mathrm{dist}}$ itself, which consumes $N_{\mathrm{EPR}}(\graphG^{(\ell)}_{\vk})$ EPR pairs for coherent communication with the children (expression~\eqref{eqn:total_epr_consumption} instantiated with the substitutions above, and $N_{\mathrm{EPR}}(\graphG^{(\ell)}_{\vk})\defeq 0$ when $\ell\in\set{L}_{\mathrm M}$), and
(ii) the child subroutines, each invoked $C_{\mathrm{eff}}(\graphG^{(\ell)}_{\vk})$ times.
Therefore,
\begin{align}
  E(\graphG^{(\ell)}_{\vk})
  =
  N_{\mathrm{EPR}}(\graphG^{(\ell)}_{\vk})
  +
  C_{\mathrm{eff}}(\graphG^{(\ell)}_{\vk})
  \cdot
  \sum_{k\in\mathcal{K}_{\ell+1}(\vk)} E(\graphG^{(\ell+1)}_{\vk,k}).
  \label{eqn:hier_epr_recursion}
\end{align}
Unrolling~\eqref{eqn:hier_epr_recursion} from the root yields the node-wise sum in Item~\ref{itm:hier_epr}, where each node's local EPR cost is weighted by its invocation count $N_{\mathrm{inv}}(\cdot)$.

\textbf{4. Success probability and precision (Item~\ref{itm:hier_success_precision}).}
\begin{itemize}
  \item \textbf{Coherent-cascade success.}
  We prove by induction on the level (from $L$ up to $0$) that each coherent subroutine succeeds with probability at least $(1-\delta)^{2\Np}$.
  For $\ell=L$, a leaf runs $\matr{U}_{\max}$, which succeeds with probability at least $(1-\delta)^{2\Np}$ by Proposition~\ref{prop:generic_maximization_performance}.
  For the inductive step, fix an internal node $\graphG^{(\ell)}_{\vk}$ with $\ell<L$ executed coherently.
  By the induction hypothesis, each child routine succeeds with probability at least $(1-\delta)^{2\Np}$.
  Hence the success branch of the state-preparation oracle for $\graphG^{(\ell)}_{\vk}$ has overlap at least
  $
  (1-\delta)^{2\Np\,K_{\ell+1}(\vk)} / \lvert\setx_{\set{V}^{(\ell+1)}_{\mathrm{B},\vk}}\rvert
  $,
  which is exactly $p_{\min,\vk}^{(\ell)}$ from~\eqref{eqn:pmin_hier_level_ell}.
  Applying Item~\ref{prop: overall_success_probability: final} of Theorem~\ref{thm: overall_algorithm_properties} to this level-$\ell$ instance of $A_{\mathrm{dist}}$ yields success probability at least $(1-\delta)^{2\Np}$.
  The claim for the root follows.

	  \item \textbf{Hybrid success.}
	  In hybrid execution, success requires that every coherently executed subroutine whose output is subsequently read out, namely each node in $\mathcal{T}_{\mathrm{out}}$ from Theorem~\ref{thm:hierarchical_performance}, returns the correct $\Np$-bit result across all its $N_{\mathrm{inv}}(\cdot)$ invocations.
	  By the coherent-cascade argument above, each such coherent execution succeeds with probability at least $(1-\delta)^{2\Np}$ (for leaves this is Proposition~\ref{prop:generic_maximization_performance}, and for internal coherent nodes it follows from Item~\ref{prop: overall_success_probability: final} of Theorem~\ref{thm: overall_algorithm_properties} applied at that node).
	  Therefore, by the chain rule for conditional probability,
	  \(
	  \Pr(\text{success})
	  \ge
	  (1-\delta)^{2\Np\,N_{\mathrm{exec}}}.
	  \)
	  The sufficient choice of $\delta$ for a target overall failure probability $\delta'$ then follows from Bernoulli's inequality, as stated in Theorem~\ref{thm:hierarchical_performance}.

  \item \textbf{Additive-error certificate.}
  For each node $\graphG^{(\ell)}_{\vk}$, let $\gmaxn{\ell,\vk}$ denote the true optimum value of that subproblem, and let $\hat{g}^{(\ell)}_{\vk}\defeq \zdec(\mathbf{z}^{(\ell)}_{\vk})$ be the decoded value returned by the corresponding hierarchical subroutine upon success.
  Define the additive error $\epsilon^{(\ell)}_{\vk}\defeq \gmaxn{\ell,\vk}-\hat{g}^{(\ell)}_{\vk}\ge 0$.
  For a leaf, Proposition~\ref{prop:generic_maximization_performance} gives $\epsilon^{(L)}_{\vk}\le 2^{-\Np}$.
  For an internal node, we decompose the error as in the single-level analysis (\textit{cf.}~\eqref{eqn: bound of gstar and znsfn} and the discussion around~\eqref{eqn:prop_approx_gmax_revised}):
  \begin{align}
    \epsilon^{(\ell)}_{\vk}
    &=
    \Biggl(\gmaxn{\ell,\vk} - \sum_{k\in\mathcal{K}_{\ell+1}(\vk)} \hat{g}^{(\ell+1)}_{\vk,k}\Biggr)
    +
    \Biggl(\underbrace{\sum_{k\in\mathcal{K}_{\ell+1}(\vk)} \hat{g}^{(\ell+1)}_{\vk,k} - \hat{g}^{(\ell)}_{\vk}}_{\overset{(a)}{\leq}2^{-\Np} }\Biggr)
    \nonumber\\
    &\leq
    \sum_{k\in\mathcal{K}_{\ell+1}(\vk)} \epsilon^{(\ell+1)}_{\vk,k}
    + 2^{-\Np},
    \label{eqn:hier_precision_recursion}
  \end{align}
  where step $(a)$ follows from the fact that $\hat{g}^{(\ell)}_{\vk}$ is, by construction of $A_{\mathrm{dist}}$ (Item~\ref{prop: overall_success_probability: final} of Theorem~\ref{thm: overall_algorithm_properties}), the best $\Np$-bit lower approximation of the summed child values, and any such lower approximation has resolution at most $2^{-\Np}$.
  Applying~\eqref{eqn:hier_precision_recursion} recursively down the tree shows that the root error is bounded by one $2^{-\Np}$ contribution per node in $\mathcal{T}$, yielding
  $
  0 \le \gmax - \hat{g}^{(0)} \le |\mathcal{T}|\cdot 2^{-\Np},
  $
  which is the claimed precision bound.
\end{itemize}

\section{Hierarchical Execution: Gate-level Policy Study (Queries and Entanglement)}
\label{sec:hier_chain_multilevel}

\begin{figure}[t]
\centering
\includegraphics[width=0.6\linewidth]{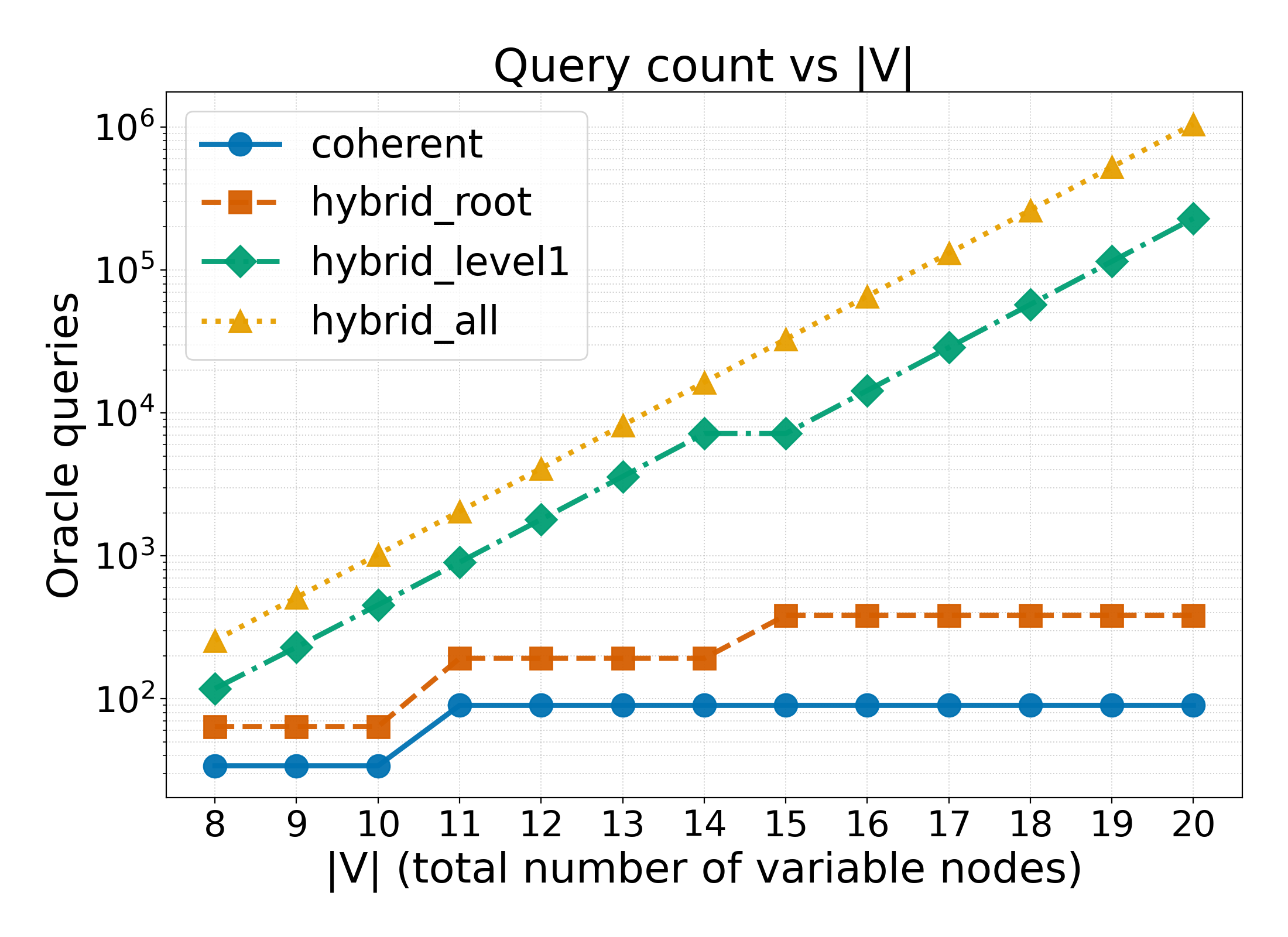}
\caption{Total oracle-query count $Q_{\mathrm{total}}$ versus $|\setV|$ (total number of variable nodes) for the coherent policy and three hybrid-policy benchmarks. The $y$-axis is logarithmic.}
\label{fig:hier_queries_vs_n_chain}
\centering
\includegraphics[width=0.6\linewidth]{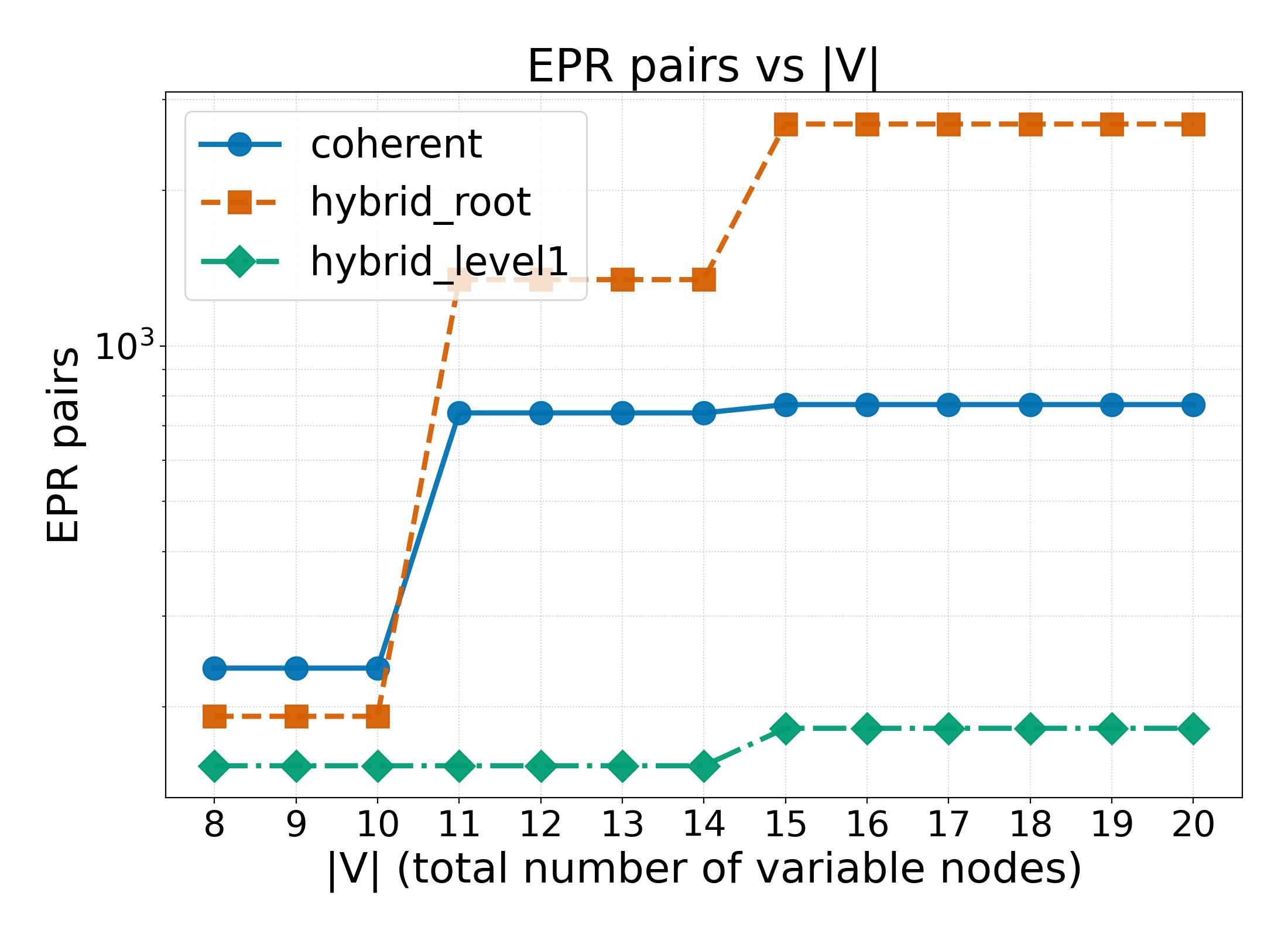}
\caption{Total EPR-pair count $E_{\mathrm{total}}$ versus $|\setV|$ for the same benchmark. We plot only policies with nonzero coherent inter-node communication; \texttt{hybrid\_all} is identically zero in this accounting model. The $y$-axis is logarithmic.}
\label{fig:hier_epr_vs_n_chain}
\end{figure}

\textbf{Goal and policies.}
This experiment isolates a single architectural choice in the hierarchical framework of Section~\ref{sec:hierarchical_algorithm}: the placement of intermediate measurements. We compare the fully coherent policy (\texttt{coherent}) with three hybrid policies: \texttt{hybrid\_root} (measurement at the root only), \texttt{hybrid\_level1} (measurement at level~1 only), and \texttt{hybrid\_all} (measurement at both internal levels). All policies use the same benchmark family, decomposition path, accuracy parameters, and accounting rules. The observed differences therefore isolate the effect of measurement placement itself.

\textbf{Parameter setting (gate-level).}
We use exact gate-level state-vector simulation with
\begin{align*}
  (b_0,b_1,S_0,S_1,r)\in
\{(5,0,1,1,r)\}_{r=3}^{5}
\cup
\{(5,1,1,1,r)\}_{r=5}^{8}
\cup
\{(6,1,1,1,r)\}_{r=8}^{13}
\end{align*}
and fixed $(\Np,\delta,t_{\max})=(1,0.2,2)$.
Here $b_0$ and $b_1$ are boundary widths at levels 0 and 1, respectively, $S_0$ and $S_1$ are branching factors at levels 0 and 1, respectively, and $r$ is the leaf-local bit width.
This monotone parameter path yields the contiguous sweep $|\setV|=8,\ldots,20$ while keeping the state dimension feasible for exact state-vector simulation. All reported oracle-query and EPR totals are extracted directly from the simulator's gate-level counters, so the figures validate the implemented circuits.

\textbf{Query ordering.}
The observed ordering is
\[
  \texttt{coherent} < \texttt{hybrid\_root} < \texttt{hybrid\_level1} < \texttt{hybrid\_all}.
\]
This ordering is exactly the behavior predicted by Item~\ref{itm:hier_query_complexity}. The total query count is built from the leaf cost in~\eqref{eqn:hier_leaf_query_cost_def}, weighted by the recursive invocation count in~\eqref{eqn:hier_invocation_count_def}. In this two-level hierarchy, \texttt{hybrid\_root}, \texttt{hybrid\_level1}, and \texttt{hybrid\_all} replace the coherent factor in~\eqref{eqn:hier_Ceff_def} at the root only, at level~1 only, and at both internal levels, respectively. Each such replacement increases the number of lower-level calls, so the ordering is strict across the full sweep. At $|\setV|=20$, the query multipliers relative to \texttt{coherent} are $4.27$, $2.55\times 10^{3}$, and $1.17\times 10^{4}$ for \texttt{hybrid\_root}, \texttt{hybrid\_level1}, and \texttt{hybrid\_all}, respectively.

\textbf{Entanglement ordering.}
The EPR accounting follows Item~\ref{itm:hier_epr}, which counts only coherent inter-node communication and inherits the per-node cost from~\eqref{eqn:total_epr_consumption}. Accordingly, \texttt{hybrid\_all} is identically zero in this accounting model because both internal interfaces are measured. At $|\setV|=20$, the EPR multipliers relative to \texttt{coherent} are $3.49$ and $0.236$ for \texttt{hybrid\_root} and \texttt{hybrid\_level1}, respectively. This contrast is important. Measuring at the root removes the root-level coherent primitive, but it also increases the number of coherent calls executed below the root; in the present benchmark, that replication dominates and leads to higher total EPR consumption. By contrast, measuring at level~1 removes the deeper coherent interface itself and therefore lowers the total EPR count. The figure thus confirms that Item~\ref{itm:hier_epr} must be read as a node-wise sum weighted by invocation multiplicity, rather than as a purely local saving rule.

\textbf{Interpreting the curve shapes.}
To obtain a dense \emph{contiguous} sweep in $|\setV|$ while keeping the state-vector dimension feasible, we vary the decomposition parameters $(b_0,b_1,S_0,S_1,r)$ along a \emph{monotone} path. 
The resulting piecewise behavior is expected. Policies that keep a level coherent are sensitive to integer-valued circuit parameters and to discrete changes in boundary width, so their curves exhibit plateau--jump behavior. Policies that measure a level inherit explicit branching through~\eqref{eqn:hier_Ceff_def}; along the present benchmark path, this yields approximately log-linear growth whenever the dominant cost comes from enumeration.
Concretely, the sweep comprises three controlled regimes:
\begin{itemize}
    \item \emph{Regime A} ($|\setV|=8,\ldots,10$): we fix $(b_0,b_1,S_0,S_1)=(5,0,1,1)$ and increase only the leaf-local width $r$. This primarily enlarges the measured enumeration domain at the lower interface. Consequently, \texttt{hybrid\_level1} and \texttt{hybrid\_all} rise rapidly, whereas the coherent policy and \texttt{hybrid\_root} remain nearly flat while the relevant integer circuit parameters stay unchanged.
    \item \emph{Regime B} ($|\setV|=11,\ldots,14$): we increase the level-1 boundary from $b_1=0$ to $b_1=1$ and continue increasing $r$. This changes the level-1 effective factor in~\eqref{eqn:hier_Ceff_def} and also modifies the coherent communication cost inherited from~\eqref{eqn:total_epr_consumption}, producing a visible step in the policies that still use or replicate the level-1 coherent primitive.
    \item \emph{Regime C} ($|\setV|=15,\ldots,20$): we increase the root boundary from $b_0=5$ to $b_0=6$ and continue increasing $r$. This changes the root-level factor in~\eqref{eqn:hier_Ceff_def}, produces a visible jump in \texttt{hybrid\_root}, and increases the coherent EPR cost per invocation. By contrast, in this two-level benchmark \texttt{hybrid\_all} effectively enumerates the full assignment space, so its query curve remains close to a straight line on the logarithmic scale.
\end{itemize}

\textbf{Implication.}
Measurement placement trades coherent depth against multiplicative replication of lower-level work. The simulations make this trade-off explicit at the circuit level: the query penalty is governed by the invocation recursion in~\eqref{eqn:hier_invocation_count_def}, while the entanglement cost depends on where coherent interfaces remain. The main value of the study is therefore not a single universally best policy, but a concrete validation that Theorem~\ref{thm:hierarchical_performance} predicts how the preferred execution mode should depend on the available entanglement budget, allowable circuit depth, and latency constraints of the target network.


\bibliographystyle{quantum}
\bibliography{biblio_used}

\fi
\end{document}